\documentclass[11pt,a4paper]{article}

\usepackage{amsmath} % AMS Math Package
\usepackage{amsthm} % Theorem Formatting
\usepackage{amssymb}	% Math symbols such as \mathbb
\usepackage{graphicx} % Allows for eps images
\usepackage{xurl} % Allow linkbreaks inside hyperlinks
\usepackage[pagebackref,plainpages=false,colorlinks,linkcolor=NavyBlue,anchorcolor=blue,citecolor=NavyBlue]{hyperref}

\hypersetup{breaklinks=true}

\usepackage[dvipsnames]{xcolor}
\usepackage[dvips,margin=1in,bottom=1in]{geometry}

\usepackage[algoruled]{algorithm2e}
\usepackage[capitalize,noabbrev]{cleveref}
\crefname{algorithm}{algorithm}{algorithms}
\Crefname{algorithm}{Algorithm}{Algorithms}
\crefname{protocol}{protocol}{protocols}
\Crefname{protocol}{Protocol}{Protocols}

\crefname{appendix}{appendix}{appendices} 
\Crefname{appendix}{Appendix}{Appendices}

\usepackage{tcolorbox}

\allowdisplaybreaks[3]

\usepackage{xspace}
\usepackage[T1]{fontenc}
\AtBeginDocument{%
  \DeclareFontShape{T1}{cmr}{m}{scit}{<->ssub*cmr/m/sc}{}%
}
\usepackage[utf8]{inputenc}
\usepackage[english]{babel}

\usepackage{adjustbox} % Adjust the maximum width of tables
\usepackage{footnotehyper} 
\makesavenoteenv{table}  
\usepackage{multirow}
\usepackage{makecell}
\usepackage{booktabs}

\usepackage{placeins}

\usepackage{xparse}

\usepackage{enumerate}
\usepackage{enumitem}

\usepackage[normalem]{ulem}

\usepackage{tikz}
\usetikzlibrary{quantikz2}

\usepackage{authblk}

\usepackage{diagbox}
\usepackage{mathtools}
\usepackage{stmaryrd}

\usepackage{thmtools}
\makeatletter
\@ifundefined{newcounteralias}{}{%
  \renewcommand\thmt@autorefsetup{%
    \@xa\def\csname\thmt@envname autorefname\@xa\endcsname
      \@xa{\thmt@thmname}%
  }%
}
\makeatother

\usepackage{thm-restate}
\usepackage{xpatch}
\makeatletter

\newcounter{protocol}
\xpatchcmd{\algocf@caption@algo}
  {\renewcommand{\theHalgocf}{\thealgocf}}
  {\relax}
  {}
  {\PackageWarning{algorithm-numbering}{Could not patch algorithm2e hyperlinks}}
\renewcommand*{\theHalgocf}{algorithm.\thealgocf}
\newcommand{\UseProtocolCounter}{%
  \let\c@algocf\c@protocol
  \let\thealgocf\theprotocol
  \renewcommand*{\theHalgocf}{protocol.\theprotocol}%
  \crefalias{algocf}{protocol}%
  \SetAlgorithmName{Protocol}{protocol}{List of Protocols}%
}

\xpatchcmd{\thmt@restatable}
  {\csname #2\@xa\endcsname\ifx\@nx#1\@nx\else[{#1}]\fi}
  {\ifthmt@thisistheone
     \csname #2\@xa\endcsname\ifx\@nx#1\@nx\else[{#1}]\fi
   \else
     \csname #2\@xa\endcsname
       \ifx\@nx#1\@nx
         [restated]%
       \else
         [{#1, restated}]%
       \fi
   \fi}
  {}
  {\PackageWarning{thm-restate}{Patch for restated theorem failed}}
\makeatother

\declaretheorem[style=plain,numberwithin=section]{theorem}
\declaretheorem[style=plain,numberlike=theorem]{lemma,corollary}
\declaretheorem[style=remark,numberlike=theorem]{remark}
\declaretheorem[style=plain,numberlike=theorem]{definition}
\declaretheorem[style=plain,numberlike=theorem]{proposition}

\declaretheorem[style=definition,numberlike=theorem]{problem,conjecture}
\numberwithin{equation}{section}

\usepackage{comment}

\DeclarePairedDelimiter\rbra{\lparen}{\rparen}
\DeclarePairedDelimiter\sbra{\lbrack}{\rbrack}
\DeclarePairedDelimiter\cbra{\{}{\}}
\DeclarePairedDelimiter\ssbra{\llbracket}{\rrbracket}
\DeclarePairedDelimiter\abs{\lvert}{\rvert}
\DeclarePairedDelimiter\norm{\lVert}{\rVert}
\DeclarePairedDelimiter\ceil{\lceil}{\rceil}
\DeclarePairedDelimiter\floor{\lfloor}{\rfloor}

\DeclarePairedDelimiterX{\sbraCond}[2]{(}{)}{#1 \;\delimsize|\; #2} % for conditional probability
\DeclarePairedDelimiterX{\innerprod}[2]{\langle}{\rangle}{#1 \delimsize| #2}

\RequirePackage{mleftright}
\mleftright
\RequirePackage{physics2}
\usephysicsmodule{ab}

\newcommand{\ketbra}[2]{\ensuremath{\ket{#1}\!\bra{#2}}}

\DeclareMathOperator*{\E}{\mathbb{E}}
\DeclareMathOperator*{\Var}{\mathrm{Var}}

\makeatletter
\def\@buildmath#1{%
  \expandafter\def\csname bb#1\endcsname{\ensuremath{\mathbb{#1}}}%
  \expandafter\def\csname bf#1\endcsname{\ensuremath{\mathbf{#1}}}%
  \expandafter\def\csname sf#1\endcsname{\ensuremath{\mathsf{#1}}}%
  \expandafter\def\csname cal#1\endcsname{\ensuremath{\mathcal{#1}}}%
  \expandafter\def\csname rm#1\endcsname{\ensuremath{\mathrm{#1}}}%
  \expandafter\def\csname tt#1\endcsname{\ensuremath{\mathtt{#1}}}%
  \expandafter\def\csname frak#1\endcsname{\ensuremath{\mathfrak{#1}}}%
}
\def\@buildmathletters#1{%
  \ifx#1\relax\else
    \@buildmath{#1}%
    \expandafter\@buildmathletters
  \fi
} % Cycling through letters
\@buildmathletters ABCDEFGHIJKLMNOPQRSTUVWXYZabcdefghijklmnopqrstuvwxyz\relax
\makeatother

\newcommand{\Tr} {\operatorname{Tr}}
\newcommand{\poly} {\operatorname{poly}}

\newcommand{\polylog} {\operatorname{polylog}}

\newcommand{\supp} {\operatorname{supp}}

\newcommand{\yes}{{\rm yes}}
\newcommand{\no}{{\rm no}}

\newcommand{\Out}{\mathrm{out}}
\newcommand{\hsep}[2]{h_{\mathrm{Sep}\rbra*{ #1, #2 }}}
\newcommand{\hsepM}[2]{h_{\mathrm{Sep}\sbra*{ #1; #2 }}}
\newcommand{\hsepPlus}[2]{h_{\mathrm{Sep}^+\rbra*{ #1, #2 }}}
\newcommand{\hsepPlusM}[2]{h_{\mathrm{Sep}^+\sbra*{ #1; #2 }}}
\newcommand{\totimes}{\!\otimes\!}
\newcommand{\wCG}{\widetilde{C}_G}

\newcommand{\PTIME}{\textnormal{\textsf{P}}\xspace}

\newcommand{\BPP}{\textnormal{\textsf{BPP}}\xspace}
\newcommand{\BQP}{\textnormal{\textsf{BQP}}\xspace}
\newcommand{\NP}{\textnormal{\textsf{NP}}\xspace}

\newcommand{\PH}{\textnormal{\textsf{PH}}\xspace}
\newcommand{\MA}{\textnormal{\textsf{MA}}\xspace}
\newcommand{\StoqMA}{\textnormal{\textsf{StoqMA}}\xspace}
\newcommand{\StoqMAtwo}{\textnormal{\textsf{StoqMA}(2)}\xspace}
\newcommand{\StoqMAk}{\textnormal{\textsf{StoqMA}(\textit{k})}\xspace}
\newcommand{\SymStoqMA}{\textnormal{\textsf{SymStoqMA}}\xspace}

\newcommand{\SymStoqMAk}{\textnormal{\textsf{SymStoqMA}(\textit{k})}\xspace}

\newcommand{\SepQMAtwo}{\textnormal{\textsf{SepQMA}(2)}\xspace}

\newcommand{\ProdStoqMAk}{\textnormal{\textsf{ProdStoqMA}(\textit{k})}\xspace}
\newcommand{\AM}{\textnormal{\textsf{AM}}\xspace}

\newcommand{\QMA}{\textnormal{\textsf{QMA}}\xspace}
\newcommand{\QCMA}{\textnormal{\textsf{QCMA}}\xspace}
\newcommand{\QMAtwo}{\textnormal{\textsf{QMA}(2)}\xspace}
\newcommand{\QMAk}{\textnormal{\textsf{QMA}}(\textit{k})\xspace}
\newcommand{\PreciseQMA}{\textnormal{\textsf{PreciseQMA}}\xspace}
\newcommand{\PreciseQMAtwo}{\textnormal{\textsf{PreciseQMA}(2)}\xspace}
\newcommand{\PreciseStoqMA}{\textnormal{\textsf{PreciseStoqMA}}\xspace}
\newcommand{\PreciseStoqMAtwo}{\textnormal{\textsf{PreciseStoqMA}(2)}\xspace}
\newcommand{\PreciseStoqMAk}{\textnormal{\textsf{PreciseStoqMA}(\textit{k})}\xspace}

\newcommand{\SBP}{\textnormal{\textsf{SBP}}\xspace}

\newcommand{\SBQP}{\textnormal{\textsf{SBQP}}\xspace}

\newcommand{\PSPACE}{\textnormal{\textsf{PSPACE}}\xspace}

\newcommand{\EXP}{\textnormal{\textsf{EXP}}\xspace}
\newcommand{\NEXP}{\textnormal{\textsf{NEXP}}\xspace}

\newcommand{\StoqCIU}{\textnormal{\textsc{StoqCIU}}\xspace}
\newcommand{\SepStoqCIU}{\textnormal{\textsc{SepStoqCIU}}\xspace}
\newcommand{\RCD}{\textnormal{\textsc{RCD}}\xspace}
\newcommand{\SepRCD}{\textnormal{\textsc{SepRCD}}\xspace}
\newcommand{\CleanCC}{\textnormal{\textsc{CleanCC}}\xspace}

\newcommand{\GSCON}{\textnormal{\textsc{GSCON}}\xspace}

\newcommand{\GapCG}{\textnormal{\textsc{GapCG}}\xspace}

\newcommand{\SAT}{\textnormal{\textsc{SAT}}\xspace}

\newcommand{\Real} {\operatorname{Re}}

\renewcommand{\H}{\mathrm{H}}

\newcommand{\td}{\mathrm{T}}
\newcommand{\TV}{\mathrm{TV}}
\newcommand{\F}{\mathrm{F}}
\newcommand{\dH}{d_{\mathrm H}} %Hellinger distance
\newcommand{\KL}{\mathrm{KL}}

\newcommand{\binset}{\{0,1\}}
\newcommand{\innerprodF}[2]{\langle #1 , #2 \rangle}

\newcommand{\Sep}{\mathrm{Sep}}
\DeclareMathOperator*{\conv}{\mathrm{conv}}

\newcommand{\negl}{\mathrm{negl}}

\newcommand{\full}{\mathrm{full}}
\newcommand{\unif}{\mathrm{unif}}
\newcommand{\cons}{\mathrm{cons}}
\newcommand{\branch}{\mathrm{branch}}
\newcommand{\match}{\mathrm{match}}
\newcommand{\route}{\mathrm{route}}
\newcommand{\dum}{\mathrm{dum}}
\newcommand{\ver}{\mathrm{ver}}
\newcommand{\dyad}{\mathrm{dyad}}
\newcommand{\sym}{\mathrm{sym}}
\newcommand{\DTIME}{\textnormal{\textrm{DTIME}}\xspace}
\newcommand{\DTISP}{\textnormal{\textrm{DTISP}}\xspace}
\newcommand{\BPTIME}{\textnormal{\textrm{BPTIME}}\xspace}
\newcommand{\Uniform}{\mathrm{Unif}}

\newcommand{\EE}{\widetilde{\mathbb E}}
\newcommand{\GammaZero}{\Gamma_{\bar{0}}}

\newcommand{\algoref}[1]{Algorithm \ref{#1}}

\newcommand{\parheading}[1]{%
  \par\addvspace{1em}%
  \noindent\emph{#1}\enspace\ignorespaces%
}
\newcommand{\subparheading}[1]{%
  \par\addvspace{0.8em}%
  \noindent\underline{#1}\enspace\ignorespaces%
}
\newcommand{\parheadingItem}[1]{%
  \par\addvspace{0.75em}%
  \noindent\emph{#1}\enspace\ignorespaces%
}

\NewDocumentCommand{\set}{m g}{%
  \left\{ #1
  \IfNoValueF{#2}{ \,\middle|\, #2 }
  \right\}%
}

\newtcolorbox{takeawaybox}{
    colback=purple!2,
    colframe=purple!50!black,
    boxrule=0.3pt,
    arc=0pt,
    left=45pt,
    right=45pt,
    top=6pt,
    bottom=6pt,
    before skip=8pt,
    after skip=8pt
}
\newcommand{\takeaway}[1]{%
\begin{takeawaybox}
\centering\itshape #1
\end{takeawaybox}%
}

\renewcommand{\epsilon}{\varepsilon}
\begin{document}
% Reduce the line spacing between equation and text
\setlength{\abovedisplayskip}{6pt}
\setlength{\belowdisplayskip}{6pt}

\title{The power of unentanglement without destructive interference}

\author[1]{Yupan Liu\thanks{Email: \url{yupan.liu@epfl.ch}}}
\author[2]{Pei Wu\thanks{Email: \url{pei.wu@psu.edu}}}
\affil[1]{School of Computer and Communication Sciences, \'Ecole Polytechnique F\'ed\'erale de Lausanne}
\affil[2]{Department of Computer Science and Engineering, The Pennsylvania State University}
\date{}

\maketitle
\pagenumbering{roman}
\thispagestyle{empty}

\begin{abstract}
    Stoquasticity, originating in physical systems free of the so-called \emph{sign problem} and avoiding destructive interference from sign cancellations, gives rise to $\sf StoqMA$, introduced by \hyperlink{cite.BBT06}{Bravyi, Bessen, and Terhal (2006)}, a quantum-inspired intermediate class between $\sf MA$ and $\sf AM$. Unentanglement similarly gives rise to ${\sf QMA}(2)$, introduced by \hyperlink{cite.KMY09}{Kobayashi, Matsumoto, and Yamakami (CJTCS 2009)}, which generalizes $\sf QMA$ to two unentangled proofs and still has only the trivial $\sf NEXP$ upper bound. 
    
    In this work, we initiate a systematic study of the power of unentanglement without destructive interference via ${\sf StoqMA}(2)$, the class of unentangled stoquastic Merlin--Arthur proof systems. Beyond its complexity-theoretic interest, ${\sf StoqMA}(2)$ is closely connected to the optimality of non-negative tensor optimization algorithms. We highlight: 
    \begin{enumerate}[label={\arabic*.}]
        \item ${\sf NP} \subseteq {\sf StoqMA}(2)$ with $\widetilde{O}(\sqrt{n})$-qubit proofs and completeness error $2^{-{\rm polylog}(n)}$, paralleling the best known ${\sf QMA}(2)$ lower bounds except for perfect completeness. 
        Meanwhile, the Sum-of-Squares algorithm of \hyperlink{cite.BKS14}{Barak, Kelner, and Steurer (STOC 2014)} gives an exponential-time upper bound for ${\sf StoqMA}(2)$. Our tightened analysis shows the essential \emph{optimality} of the parameters in both our protocol and the BKS algorithm under the Exponential-Time Hypothesis (ETH). 
        \item For logarithmic-size proofs, ${\sf NP} \subseteq {\sf StoqMA}(2)_{\log}$ with completeness error $O(n^{-2})$ and implicitly vanishing gap. Meanwhile, ${\sf StoqMA}(2)_{\log} \subseteq {\sf MA}$, and consequently quantum-inspired randomness enables \emph{exponentially} shorter unentangled proofs even under the standard derandomization assumption ${\sf MA}={\sf NP}$.
        \item ${\sf PreciseStoqMA}(2)$, a variant of ${\sf StoqMA}(2)$ with exponentially small promise gap, \emph{cannot} achieve perfect completeness unless ${\sf EXP}={\sf NEXP}$. In contrast, ${\sf PreciseStoqMA}$ achieves perfect completeness, since ${\sf PSPACE} \subseteq {\sf PreciseStoqMA}_1$. 
        \item Under nearly perfect completeness, ${\sf StoqMA}(2)$ admits sharper deterministic upper bounds: for ${\sf StoqMA}(2)_1$, the parameter dependence in the general ETH-optimal time bound can be \emph{exponentially} improved, or the bound achieved \emph{simultaneously} with polynomial space.        
    \end{enumerate}

    Our lower bounds are obtained by stoquastizing the short-proof ${\sf QMA}(2)$ protocols using \emph{distribution testing} techniques. Our upper bounds for the nearly perfect completeness case are proved via our \emph{rectangular closure testing} framework, a new combinatorial technique tailored to ${\sf StoqMA}(2)_1$. 
\end{abstract}

\newpage
\tableofcontents
\thispagestyle{empty}

%%%%%%%%%%%%%%%%%%%%%%%%%%%%%%%%%%%%%%%%%%%%%%%%%%%%%%%%%%%
%%%%%%%%%%%%%%%%%%%%%%%%%%%%%%%%%%%%%%%%%%%%%%%%%%%%%%%%%%%
%%%%%%%%%%%%%%%%%%%%%%%%%%%%%%%%%%%%%%%%%%%%%%%%%%%%%%%%%%%

\newpage
\pagenumbering{arabic}

%%%%%%%%%%%%%%%%%%%%%%%%%%%%%%%%%%%%%%%%%%%%%%%%%%%%%%%%%%%
%%%%%%%%%%%%%%%%%%%%%%%%%%%%%%%%%%%%%%%%%%%%%%%%%%%%%%%%%%%
%%%%%%%%%%%%%%%%%%%%%%%%%%%%%%%%%%%%%%%%%%%%%%%%%%%%%%%%%%%

\section{Introduction}
The notion of \emph{stoquasticity}, a blend of \emph{sto}cha\emph{stic} and \emph{qua}ntum, was first introduced in~\cite{BDOT06} to describe Hamiltonians whose off-diagonal matrix elements are real and non-positive, a condition that naturally connects the model to \emph{stochastic processes} such as Markov chains and random walks. Motivated by this connection, Bravyi, Bessen, and Terhal~\cite{BBT06} introduced the complexity class $\StoqMA$, for which the stoquastic local Hamiltonian problem is complete. 

A useful way to understand \StoqMA{} is to compare it directly with \MA{}~\cite{Bab85}, the class of Merlin--Arthur proof systems and the standard randomized generalization of \NP{}. Indeed, \MA{} admits a natural ``quantum'' formulation~\cite{BDOT06}: the verifier uses classical reversible circuits built from Toffoli, CNOT, and X gates, together with the usual $\ket{0}$ ancillary qubits and $\ket{+}$ ancillary qubits that provide randomness, and then measures a designated output qubit in the \emph{computational} basis. Allowing the witness to be quantum does not increase the power of this model, as quantum witnesses offer \emph{no} advantage over classical ones. 

By contrast, in \StoqMA{} the only formal difference is that the designated output qubit is measured in the \emph{Hadamard} basis instead. Consequently, \StoqMA{} naturally contains \MA{}.\footnote{ As in~\cite[Appendix A.4]{BBT06}, let the \MA{} verifier be in the state $\ket{0}\ket{R_0} + \ket{1}\ket{R_1}$ before its final measurement. After appending $\ket{+}\ket{0}$ and applying one Toffoli gate, the resulting \StoqMA{} verifier has the state
$\rbra*{ \ket{000}\ket{R_0} +\allowbreak \ket{100}\ket{R_0} +\allowbreak \ket{001}\ket{R_1} +\allowbreak \ket{111}\ket{R_1} }/\sqrt{2}$
before the final measurement. Hence, the acceptance probability $p$ becomes $(1+p)/2$: the accepting branch interferes constructively, while the rejecting branch is split into orthogonal labels. }
This apparently minor syntactic change has important structural consequences. By the Perron--Frobenius theorem in linear algebra, the acceptance probability of a \StoqMA{} verifier is always maximized by a witness state whose amplitudes are \emph{non-negative} in the computational basis,\footnote{In this work, we refer to such states as \emph{non-negative states}.} and this structural property makes the underlying optimization substantially easier than in \QMA{}, the fully quantum analogue of \NP{}. 

Whereas the local Hamiltonian problem, which asks one to approximate the minimum eigenvalue of an underlying Hermitian matrix, is complete for \QMA{}, and \QMA{} is thus at least contained in \PSPACE{}~\cite{Kitaev99,KSV02}, the stoquastic local Hamiltonian problem can be reformulated as the problem of estimating the trace of a power of an entrywise non-negative matrix with efficiently computable local contributions. Using the techniques in~\cite{BGM06,GS89}, this places \StoqMA{} at least inside \AM{}~\cite{BBT06}, and hence inside the polynomial-time hierarchy (\PH{})~\cite{Bab85}. In this sense, \StoqMA{} is a \emph{semi-quantum} class: it is quantum enough to go beyond \MA{}, yet still structured enough to admit unexpectedly classical upper bounds, unlike \BQP{}, for which even proving containment in \PH{} is technically difficult in light of oracle separations~\cite{RT19}.

\paragraph{Further perspectives on \StoqMA{}.}
Beyond serving as a natural quantum-inspired intermediate class between \MA{} and \AM{}, \StoqMA{} admits several further perspectives: 
    \parheadingItem{Physical perspective.} Traditionally, \StoqMA{} corresponds to Hamiltonians without the \emph{sign problem}: a regime with no destructive interference from sign cancellations, where random-walk-based Monte Carlo simulations become applicable. A standard example is stoquastic $k$-\SAT{}, which is both complete for the perfect-completeness version of \StoqMA{} ($\StoqMA_1$) and contained in \MA{}~\cite{BBT06,BT10}.
    More recently, this containment was extended slightly to \StoqMA{} with negligible completeness error~\cite{AGL20}.\footnote{This improvement is achieved via a careful analysis of the underlying random walk for \emph{yes} instances in~\cite[Section 5]{AG19}, together with a probabilistic argument for the lazy random walk, such as~\cite[Proposition 2.5]{ST13}.} Beyond being a formal model, \StoqMA{} also captures physically meaningful Hamiltonians, such as the transverse-field Ising model whose associated Hamiltonian problem is \StoqMA{}-complete~\cite{BH17} and is closely related to adiabatic quantum computation~\cite{AL18}.
    
    \parheadingItem{Complexity-theoretic perspective.} From a more modern viewpoint, \StoqMA{} occupies a distinguished place in the classification of local Hamiltonian problems. In the qubit $2$-local setting, every problem falls into one of four categories: in \PTIME{}, \NP{}-complete, \StoqMA{}-complete, or \QMA{}-complete~\cite{CM16}, yielding a quantum analogue of Schaefer's dichotomy theorem~\cite{Schaefer78}.\footnote{Other examples include local Hamiltonian problems with symmetric $2$-local interactions and no 1-local terms, which are either \QMA{}-complete or in \StoqMA{}~\cite{PM15}, as well as a recent classification of EPR Hamiltonians~\cite{MS26}, whose levels are \QMA{}-complete, \StoqMA{}-complete, \NP{}-complete, and reducible to the $\textsc{EPR}^*$ problem conjectured to be in \BPP{}.} Under derandomization assumptions implying $\AM=\NP$~\cite{KvM02,MV05}, the intermediate \StoqMA{}-complete level would disappear. Likewise, under the standard derandomization assumption that $\MA=\NP$, proving $\StoqMA=\MA$ would eliminate the same level.
    
    \parheadingItem{Distribution-testing perspective.} More recently, an intrinsic connection between \emph{distribution testing} and the final Hadamard-basis measurement in a \StoqMA{} verifier has been identified in~\cite{Liu21}. Roughly speaking, the verifier's final measurement computes a Hellinger-affinity-type quantity between two induced distributions that are determined jointly by the verification circuit and the witness. 
    This perspective also yields another non-trivial subclass $\mathsf{eStoqMA}$ of \StoqMA{} that is contained in \MA{}
    via distribution-testing techniques~\cite{CR14}, as well as parallel repetition for \StoqMA{}~\cite{Liu21}, which immediately implies  error reduction for \StoqMA{} with negligible completeness error.

\paragraph{Unentangled Merlin--Arthur proof systems and their motivations.}
Parallel to the development of stoquastic Merlin-Arthur proof systems, the notion of \emph{unentanglement} gives rise to the class \QMAtwo{} introduced by Kobayashi, Matsumoto, and Yamakami~\cite{KMY09}, in which the verifier has full quantum power but receives two unentangled quantum proofs, or equivalently the witness is required to be a separable state across two registers.\footnote{See also~\cite{JLW26} for a comprehensive survey about \QMAtwo{}.} This viewpoint was then substantially clarified by Harrow and Montanaro~\cite{HM13}, who showed that \QMAk{}, for up to polynomially many provers, collapses to \QMAtwo{}. Their work established error reduction for \QMAtwo{}, showing the class is robust under the usual choice of constant completeness and soundness parameters.

On the upper-bound side, however, \QMAtwo{} is notoriously difficult. The underlying optimization is over separable states and is therefore \emph{more subtle} than that over arbitrary quantum states, which obstructs many of the techniques available for \QMA{}. For the general class, the trivial upper bound of \NEXP{} remains the best known. Nevertheless, several restricted settings are better understood. For example, when the verifier's induced measurements are restricted to Bell or one-way LOCC, the corresponding subclasses are contained in \QMA{}~\cite{BCY11,BH13,LS15}. Several further perspectives on \QMAtwo{} are as follows:

    \parheadingItem{Power of short unentangled proofs.} \QMAk{} with short proofs already captures \NP{} in striking ways. 
    Blier and Tapp~\cite{BT07} first established that $\NP$ admits two $O(\log n)$-qubit quantum proofs, with a \emph{vanishing} gap. A scaled-up version of this idea later implied that $\mathsf{PreciseQMA}(2)=\NEXP$~\cite{Pereszlenyi12}, and the soundness analysis was subsequently improved in~\cite{Beigi10,CF13,LNN12}.
    Later work focused on obtaining a \emph{constant} gap, a goal first realized in~\cite{ABDFS09}, and later refined to use two $\widetilde{O}(\sqrt{n})$-qubit quantum proofs by combining PCP machinery with prover-elimination techniques from~\cite{HM13}. This line of work also yields conditional hardness of \QMAtwo{} under the Exponential-Time Hypothesis (ETH). Moreover, the measurement underlying this \QMAtwo{} protocol can be relaxed further, even to Bell measurements~\cite{CD10}, with subsequent improvements in~\cite{CF13,NZ23}. More recently, imposing a \emph{non-negative} promise on both \emph{yes} and \emph{no} instances of \QMAtwo{}, thereby defining $\QMA^+(2)$, further strengthens the power of unentanglement, yielding $\QMA^+(2)=\NEXP$~\cite{JW23}, with subsequent developments in~\cite{JW24,BFM24,BFLMW24}. 

    \parheadingItem{Optimization over separable states.} A central way to understand \QMAtwo{} is as tensor optimization in disguise: the verifier's maximum acceptance probability is obtained by optimizing over the separable states, rather than by estimating a top eigenvalue. This connects \QMAtwo{} to a broader landscape of tensor optimization and hardness of approximation, including themes around the Unique Games Conjecture~\cite{BBH+12}.
    This viewpoint has led to several distinct algorithmic approaches. One sophisticated line of developments uses semidefinite-programming hierarchies based on relaxations of separable states, such as PPT and extendibility~\cite{DPS04, NOP09, HNW17, FF21}. 
    The convergence rate of the corresponding SDP hierarchies is often guaranteed by quantum information-theoretic or de Finetti type arguments asserting that symmetric states must be close to separable. In restricted measurement regimes such as one-way LOCC, this SDP hierarchy yields quasipolynomial-time algorithms~\cite{BCY11, BH13, LS15}. Another line of attacks includes smarter search algorithms exploring the set of quantum states such as epsilon-net methods~\cite{SW15, BH15}, and Sum-of-Squares-based algorithms~\cite{BKS14, BKS17}. 
    More recently, algebraic-geometric based approaches have been explored~\cite{HNW17, KR24}. 
    Altogether, this optimization perspective places \QMAtwo{} at the interface of separability testing, convex optimization, and classical hardness of approximation.
    
    \parheadingItem{Physical and information-theoretic perspective.} One canonical \QMAtwo{}-complete problem is \textsc{Separable Sparse Hamiltonian}~\cite{CS12}. By contrast, \textsc{Separable Local Hamiltonian}, though once conjectured to be \QMAtwo{}-complete, is in fact in \QMA{}, and the hardness techniques of~\cite{CS12} apply only in the \emph{sparse} setting. Another natural \QMAtwo{}-complete problem is \textsc{Pure Product Isometry Output} (\textsc{PIPO}), which asks whether a given isometry $U$ maps some pure state $\ket{\psi}$ to a state close to a pure product state~\cite{GHMW15}. Interestingly, \textsc{PIPO} also serves as a starting point for proving that \textsc{Low Entropy Low Energy State}, namely the intersection of the local Hamiltonian problem and the task of deciding whether a state has low bipartite entropy, is \QMAtwo{}-hard via a new channel-to-Hamiltonian simulation that yields a \emph{local} Hamiltonian~\cite{GK25}. 

\paragraph{From stoquastic to unentangled stoquastic.} 
Since stoquasticity originates from \emph{sign-problem-free} physical systems, which avoid destructive interference from sign cancellations, our model \StoqMAtwo{} naturally combines the two themes above and captures the power of \emph{unentanglement without destructive interference}. Algebraically, this absence of destructive interference is reflected by \emph{non-negativity}. This interaction leads to the following two questions:
\begin{quote}
\begin{enumerate} [label={\upshape(\alph*)}]
    \item \label{question:lower-bound}
    Can unentanglement remain computationally powerful even without destructive interference?
    More concretely, can \StoqMAtwo{} recover some known \QMAtwo{}-type lower bounds?
\end{enumerate}
\end{quote}

Prior work suggests that the ability to perform distribution-testing-type tasks makes \StoqMA{} appear more powerful than \MA{}, with this extra power closely tied to non-negativity~\cite{Liu21}; moreover, non-negativity can dramatically strengthen unentanglement~\cite{JW23}. Together, these results support the intuition behind Question \ref{question:lower-bound}. 

\begin{quote}
\begin{enumerate} [label={\upshape(\alph*)}]
    \setcounter{enumi}{1}
    \item \label{question:upper-bound}
    Can stoquasticity impose new computational limitations on unentangled proof systems? More precisely, does $\StoqMAtwo{}$ admit (algorithmic) upper bounds beyond \NEXP{}, still out of reach for general \QMAtwo{}?
\end{enumerate}
\end{quote}

Question \ref{question:upper-bound} reflects the viewpoint that the intrinsic non-negativity of \StoqMA{} makes its convex structure easier to exploit and therefore imposes extra structure, as supported by the best known upper bounds: $\QMA \subseteq \SBQP$~\cite{Vyalyi03,MW05,Kuperberg15}, whereas $\StoqMA \subseteq \SBP$~\cite{BBT06}, where \SBQP{} and \SBP{} are variants of \BQP{} and \BPP{}, respectively, with a completeness-to-soundness ratio at least $2$. In the general polynomial optimization literature, a similar sign-free phenomenon has been exploited~\cite{BGGLT17}, admitting faster approximation algorithms.

\vspace{0.5em}
We provide surprisingly sharp answers to these questions.

\subsection{Main results}
\label{subsec:main-results}
We now present and explain our main results from both the proof-system and algorithmic perspectives. These two perspectives are closely connected and together establish ETH-based \emph{optimality} for both short stoquastic unentangled proofs for \NP{}, mirroring~\cite{ABDFS09,CF13}, and the Barak--Kelner--Steurer algorithm~\cite{BKS14} for non-negative tensor optimization problems.

\paragraph{The proof-system perspective.} 
We begin by stating our first main result concerning \StoqMAtwo{} with short proofs, given in \Cref{thm-informal:StoqMAtwo-short-proofs}. For convenience, we write $\StoqMA_\ell\ssbra{k,c,\Delta}$ for the class of \StoqMAk{} proof systems in which each proof has length $\ell(n)$, completeness $c(n)$, and promise gap $\Delta(n) \coloneqq c(n) - s(n)$, where $s(n)$ denotes the soundness parameter. For simplicity, we may omit the subscript $\ell$ when $\ell(n)$ is polynomially bounded. 

\begin{theorem}[The power of \StoqMAtwo{} with short proofs, informal version of \Cref{thm:power-stoqma-np,thm:power-stoqma-log-np}] 
\label{thm-informal:StoqMAtwo-short-proofs}
The following inclusions hold:\footnote{The notation $\widetilde{O}(\cdot)$ suppresses poly-logarithmic factors; see the beginning of \Cref{sec:preliminaries} for the formal definition.}
\begin{enumerate}[label={\upshape(\arabic*)}]
    \item \label{thmitem:StoqMAtwo-short-constantGap}$\displaystyle\NP \subseteq \StoqMA_{\widetilde{O}(\sqrt{n})}\ssbra*{2,1-2^{-\polylog(n)},\Omega(1)} \subseteq \MA$.
    \item \label{thmitem:StoqMAtwo-short-vanishGap}$\displaystyle\NP \subseteq \StoqMA_{O(\log{n})}\ssbra*{2,1-O\rbra*{n^{-2}},\Omega\rbra*{n^{-2}}} \subseteq \MA$. 
\end{enumerate}
\end{theorem}

Surprisingly, \Cref{thm-informal:StoqMAtwo-short-proofs}\ref{thmitem:StoqMAtwo-short-constantGap} matches the prior constant-gap \QMAtwo{} protocols with sublinear-size proofs~\cite{ABDFS09,CD10}, apart from the loss of perfect completeness.
A similar short-proof phenomenon is unlikely to hold even for \QMA{}, since otherwise it would violate the Quantum Exponential-Time Hypothesis~\cite{CCZZ22}. Likewise, this phenomenon appears unlikely for $\MA(2)$, which trivially collapses to \MA{} and would therefore violate randomized ETH (rETH)~\cite{DHMTW14}. 
Therefore, since \StoqMA{} naturally contains \MA{}, this argument suggests that collapsing \StoqMAtwo{} to \StoqMA{}, referred to as a \emph{multi-to-single witness reduction}, would incur at least a \emph{quadratic} blow-up in the proof length. It turns out that this quadratic loss in the multi-to-single witness reduction (e.g.,~\cite{GJ26}) is \emph{unavoidable}, assuming rETH.\footnote{This conditional optimality result follows by combining \Cref{thm-informal:StoqMAtwo-short-proofs}\ref{thmitem:StoqMAtwo-short-constantGap} with $\SAT \in \StoqMA_\ell\ssbra{1,1,1/2}$ (see \Cref{subsec:omitted-rETH-opt-BPTIME-bound}), where $\ell(n)=n$ and $n$ denotes the number of variables in the \SAT{} formula. Consequently, applying a multi-to-single witness reduction to \Cref{thm-informal:StoqMAtwo-short-proofs}\ref{thmitem:StoqMAtwo-short-constantGap} without incurring a quadratic blow-up would refute rETH.}

Similarly, \Cref{thm-informal:StoqMAtwo-short-proofs}\ref{thmitem:StoqMAtwo-short-vanishGap} matches the prior inverse-polynomial-gap \QMAtwo{} protocols with logarithmic-size proofs~\cite{BT07,Beigi10,CF13,LNN12}, except for the loss of perfect completeness. Together with a sufficiently structured PCP for \NEXP{}, such as~\cite[Theorem 7.4]{JW23}, we obtain the following scaling-up version of \Cref{thm-informal:StoqMAtwo-short-proofs}\ref{thmitem:StoqMAtwo-short-vanishGap}: 

\begin{corollary}
    \label{corr:NEXP-in-PreciseStoqMAtwo}
    $\NEXP \subseteq \StoqMA\ssbra*{2,1-2^{-\poly(n)},\Omega\rbra*{2^{-\poly(n)}}} \subseteq \PreciseStoqMAtwo$.
\end{corollary}

Here, ``precise'' means that the promise gap of the class variant is \emph{exponentially} small. As a consequence of \Cref{corr:NEXP-in-PreciseStoqMAtwo} and the equivalence $\NEXP=\PreciseQMAtwo$ established in~\cite{Pereszlenyi12}, we obtain the following equivalence:
\begin{corollary}
    \label{corr:PreciseStoqMAtwo-eq-PreciseQMAtwo}
    $\PreciseStoqMAtwo = \NEXP = \PreciseQMAtwo$.
\end{corollary}

As we will discuss later in the algorithmic perspective, the protocol underlying \Cref{thm-informal:StoqMAtwo-short-proofs}\ref{thmitem:StoqMAtwo-short-constantGap} is \emph{optimal} under the Exponential-Time Hypothesis (ETH), when combined with the Sum-of-Squares algorithm of~\cite{BKS14} for non-negative tensor optimization problems.  
Interestingly, our results (\Cref{thm-informal:StoqMAtwo-short-proofs,thm-informal:StoqMAk-in-EXP}), achieving ETH-based optimality, use the class \StoqMAtwo{} in regimes that collapse to \MA{}. 
Here, the \MA{} upper bound in \Cref{thm-informal:StoqMAtwo-short-proofs}\ref{thmitem:StoqMAtwo-short-constantGap} is obtained by combining~\cite{GJ26,AGL20}, while the one in \Cref{thm-informal:StoqMAtwo-short-proofs}\ref{thmitem:StoqMAtwo-short-vanishGap} follows from our new upper bound for $\StoqMAk_{\log}$, where the subscript $\log$ denotes logarithmic-size proofs:\footnote{Notably, the inclusion $\StoqMA_{\log} \subseteq \BPP$ in \Cref{thm-informal:upper-bound-StoqMAk-short-witness} can be extended to a bounded-error randomized algorithm for $\StoqMA_\ell$ with time complexity \emph{singly exponential} in the proof length $\ell$, and this dependence is \emph{optimal} under randomized ETH, up to poly-logarithmic factors.}

\begin{theorem}[Upper bounds for \StoqMAk{} with logarithmic-size proofs, informal version of~\Cref{thm:StoqMA-BPTIME-upper-bound,thm:StoqMAk_log-in-MA}]
\label{thm-informal:upper-bound-StoqMAk-short-witness}
For every integer-valued function $k(n)$, it holds that
\[ \StoqMAk_{\log} \subseteq \MA \quad\text{while}\quad \StoqMA_{\log} = \BPP. \]
\end{theorem}

Since the proof-length advantages in stoquastic short unentangled proofs remain even assuming the standard derandomization assumption $\MA=\NP$, our first result and its corollaries give a positive answer to Question \ref{question:lower-bound}: 
\takeaway{Even assuming $\MA=\NP$, quantum-inspired randomness enables \emph{sublinear} unentangled proofs and (at least) \emph{quadratic} proof-length advantages.}

\paragraph{The algorithmic perspective.}
Next, we turn to our second main result:
\begin{theorem}[Upper bounds for \StoqMAk{} with perfect completeness, informal version of~\Cref{thm:StoqMAtwo-perfect-completeness,thm:PreciseStoqMAtwo-perfect-completeness}]
\label{thm-informal:upper-bound-StoqMAtwo-noYesError}
For all integer-valued functions $\ell(n)$, $k(n)$, $m_0(n)$, and $m_+(n)$, and every $\Delta(n) \in (0,1/2]$, where $m_0(n)$ and $m_+(n)$ denote the numbers of $\ket{0}$ and $\ket{+}$ ancillary qubits in the underlying stoquastic verification circuit, the following inclusion holds:\footnote{The simultaneous exponential-time and polynomial-space upper bound in \Cref{thm:StoqMAtwo-perfect-completeness} directly extends to the regime where the completeness error is doubly exponentially small, while the singly exponential-time upper bound in \Cref{thm:PreciseStoqMAtwo-perfect-completeness} directly extends to the regime where the completeness error is triply exponentially small.}
\begin{align*}
    \StoqMA_\ell\ssbra{k,1,\Delta} &\subseteq 
    \DTISP\sbra[\bigg]{ \exp\rbra[\Big]{\widetilde{O}\rbra[\Big]{\frac{k^2\ell^2}{\Delta^2}}} \cdot 2^{m_+}, \widetilde{O}\rbra[\Big]{\frac{k^2\ell^2}{\Delta^2}}+O(m_0+m_+)}\\
    &\qquad \,\cap\, \DTIME\sbra*{ \exp\rbra[\big]{\widetilde{O}\rbra{\max\cbra{k\ell,m_+}}}}.
\end{align*}
\end{theorem}

Here, $\DTISP[t,s]$ denotes the class of problems decidable simultaneously in deterministic time $t$ and space $s$.
Taken together, \Cref{corr:PreciseStoqMAtwo-eq-PreciseQMAtwo} and the deterministic time bound in \Cref{thm-informal:upper-bound-StoqMAtwo-noYesError} imply the following separation under the standard assumption $\EXP\neq\NEXP$:
\begin{corollary}
    \label{corr:PreciseStoqMAtwo-no-perfect-completeness}
    $\PreciseStoqMAtwo_1\subsetneq \PreciseStoqMAtwo$ unless $\EXP=\NEXP$. 
\end{corollary}

Here, the subscript $1$ indicates perfect completeness.
Notably, the failure of perfect-completeness closure for \PreciseStoqMAtwo{}, as provided in \Cref{corr:PreciseStoqMAtwo-no-perfect-completeness}, contrasts with the closure under perfect completeness for the general quantum classes: specifically, $\PreciseQMA=\PreciseQMA_1$~\cite{FL16,FL18} and $\PreciseQMAtwo=\PreciseQMAtwo_1$~\cite{Pereszlenyi12}. Nevertheless, this failure \emph{does not extend} to the single-prover precise stoquastic setting: our inclusion $\PSPACE\subseteq \PreciseStoqMA_1$, established in \Cref{sec:cleanCC-in-PreciseStoqMA}, together with~\cite{FL16,FL18}, also implies that \PreciseStoqMA{} is closed under perfect completeness:\footnote{\Cref{thm:PreciseStoqMA-eq-PreciseQMA} follows from \Cref{thm:PreciseStoqMA-perfect-completeness}, and we refer the reader to \Cref{sec:cleanCC-in-PreciseStoqMA} for details.}
\begin{theorem}
    \label{thm:PreciseStoqMA-eq-PreciseQMA}
    $\PreciseStoqMA = \PreciseStoqMA_1 = \PSPACE = \PreciseQMA$.
\end{theorem}

Furthermore, \Cref{corr:PreciseStoqMAtwo-eq-PreciseQMAtwo,corr:PreciseStoqMAtwo-no-perfect-completeness,thm:PreciseStoqMA-eq-PreciseQMA} together reveal a subtle distinction among precise stoquastic proof systems: although the precise stoquastic complexity classes match their quantum counterparts in both the single-prover and two-prover settings without perfect completeness, this strength does not persist under perfect completeness once unentanglement is present.

\vspace{1em}
Our third main result records a general upper bound for \StoqMAk{}:\footnote{It is important to note that, as in \Cref{thm-informal:upper-bound-StoqMAtwo-noYesError}, the time bound in \Cref{thm-informal:StoqMAk-in-EXP} also has an \emph{implicit exponential dependence} on $m_+$ of the form $2^{O(m_+)}$, arising from deterministically transforming a stoquastic verification circuit into an entrywise non-negative matrix of dimension $2^{k\ell}$. See \Cref{remark:m+-dependence-in-BKS} for details.}
\begin{theorem}[Informal version of \Cref{thm:StoqMAk-in-EXP}]
    \label{thm-informal:StoqMAk-in-EXP}
    For every integer-valued $\ell(n)$ and $k(n)$, and for all functions $c(n) \in (1/2,1]$ and $\Delta(n) \in (0,1/2]$ such that $c(n)-\Delta(n)\geq 1/2$, it holds that
    \[\StoqMA_\ell\ssbra{k,c,\Delta} \subseteq \DTIME\sbra*{\exp\rbra*{ \widetilde{O}\rbra[\bigg]{\frac{k^2\ell^2}{\Delta^2}} }}.\] 
\end{theorem}

\Cref{thm-informal:StoqMAk-in-EXP} follows from the Sum-of-Squares algorithm of Barak, Kelner, and Steurer~\cite{BKS14}, as stated in \Cref{thm:bks-sos}. Remarkably, together with \Cref{thm-informal:StoqMAtwo-short-proofs}\ref{thmitem:StoqMAtwo-short-constantGap}, our refined analysis implies that the BKS algorithm is essentially \emph{optimal} under ETH: any improvement over the stated dependence, up to poly-logarithmic factors, on the tensor order $t$, corresponding to the number of provers $k$, or on the dimension $d=2^\ell$, corresponding to the witness length, would yield an algorithm for \SAT{} that runs faster than ETH permits. 
This ETH-based optimality also explains an interesting subtlety concerning the \emph{loss of perfect completeness} in \Cref{thm-informal:StoqMAtwo-short-proofs}: it occurs precisely in the regime in which the dependence on both $k$ and $\ell$ in the time complexity in \Cref{thm-informal:upper-bound-StoqMAtwo-noYesError} is \emph{exponentially} improved over that of the general algorithm in \Cref{thm-informal:StoqMAk-in-EXP}.

Finally, \Cref{thm-informal:upper-bound-StoqMAtwo-noYesError,thm-informal:StoqMAk-in-EXP} together establish non-trivial algorithmic upper bounds for \StoqMAtwo{}. These results contrast with the trivial \NEXP{} upper bound for \QMAtwo{}, which remains the state of the art more than two decades after the introduction of \QMAtwo{}~\cite{KMY09}. Consequently, our upper bounds give a positive answer to Question \ref{question:upper-bound} and indicate that 

\takeaway{Stoquasticity imposes exploitable structure in unentangled proof systems, yielding upper bounds much better than \NEXP{}.}

\paragraph{A stoquastic complexity zoo.}
For readability, \Cref{fig:classes-relation} summarizes the relationships among the complexity classes discussed in this subsection: the classes introduced in this work are shown in {\color{purple}\emph{purple}} shaded boxes with dashed borders, while previously known classes are shown in {\color{gray}\emph{gray}} shaded boxes with solid borders. New inclusions are indicated by bold {\color{purple}\emph{purple}} lines, and new equivalences appear in {\color{purple}\emph{purple}} shaded boxes with bold borders. In addition, the new results proved in the follow-up work~\cite{GJ26} are shown in {\color{RoyalBlue}\emph{blue}} shaded boxes with dashed borders.

\begin{figure}[!ht]
    \centering
    \includegraphics[width=0.63\linewidth]{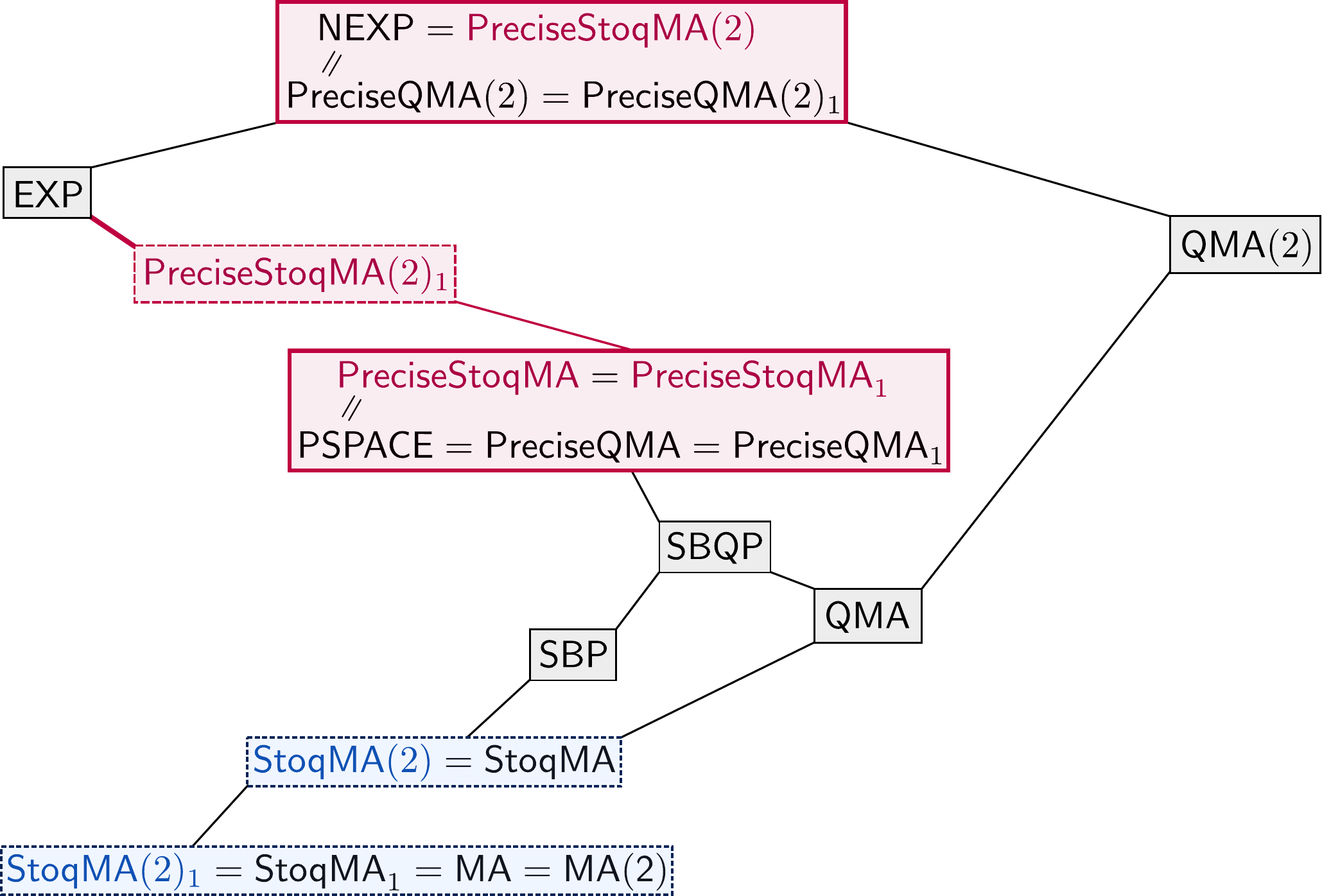}
    \caption{Relationships among \StoqMAtwo{}, its variants, and other complexity classes.}
    \label{fig:classes-relation}
\end{figure}

\subsection{Proof techniques}

To explain our main results in \Cref{subsec:main-results}, we first present basic properties of \StoqMAk{} and its symmetric variant, which lead to the robustness of \StoqMAtwo{}. Equipped with these properties, we then explain how our lower and upper bounds are achieved. 

\paragraph{Basic properties of \StoqMAk{}.} 
We start with a consequence of parallel repetition for \StoqMAk{} (\Cref{thm:StoqMAk-w-parallel-repetition-loss}), where the subscript ``$1-\negl$'' denotes negligible completeness error:

\begin{theorem}[Error reduction for $\StoqMAk_{1-\negl}$, informal version of \Cref{thm:StoqMAk-w-parallel-repetition-loss,corr:error-reduction-StoqMAk-negl}] 
\label{thm-informal:StoqMAk-error-reduction-yes}
Let $\epsilon(n)$ be a negligible function, and let $\Delta(n) \geq 1/\poly(n)$ denote the promise gap. Then, for every polynomially bounded $r(n)$ satisfying $r(n) \geq \ceil*{n \big/ \log_2\rbra[\big]{\frac{1}{1-\epsilon-\Delta}}}$, it holds that
    \[ \StoqMA_\ell\ssbra*{k,1-\epsilon,\Delta} \subseteq \StoqMA_{r\ell}\ssbra*{k,1-\epsilon',\Omega\rbra*{1}}. \]
Here, $\epsilon'(n)$ is also a negligible function. 
\end{theorem}

It is noteworthy that \Cref{thm-informal:StoqMAk-error-reduction-yes} mirrors parallel repetition for \StoqMA{}~\cite[Section 5]{Liu21}, but with parameter loss. The main obstacle to extending the approach of~\cite[Section 5]{Liu21} to \StoqMAtwo{} is that its soundness analysis relies crucially on the \emph{eigenvalue properties} of stoquastic Hermitian matrices,\footnote{For convenience, we refer to entrywise non-negative matrices as \emph{stoquastic} matrices.} whereas the analogous argument fails for \emph{separable values} of stoquastic Hermitian matrices in general;\footnote{Roughly speaking, for any \StoqMAtwo{} verification circuit $V_x$, the maximum acceptance probability of $V_x$ corresponds to the separable value of the matrix $\bra{\bar{0}}\bra{\bar{+}}V_x^\dagger \ketbra{+}{+}_\Out V_x \ket{\bar{0}}\ket{\bar{+}}$. See \Cref{subsec:separable-values} for a formal definition.} see \Cref{remark:StoqMAk-parallel-repetition-parameter-loss}. Instead, we work directly with separable values of stoquastic positive semi-definite (PSD) matrices, thereby showing that separable values are multiplicative for stoquastic PSD matrices (\Cref{lemma:stoqSepVal-multiplicative-PSD}). Interestingly, for general \QMAtwo{}, such a multiplicative property holds only for \emph{separable} Hermitian matrices; cf.~\cite[Lemma 10]{HM13}.

\vspace{1em}
However, unlike \QMAtwo{}, which admits error reduction~\cite{KMY09,ABDFS09,HM13}, general error reduction for \StoqMA{} remains open and would collapse \StoqMA{} to \MA{}~\cite{AGL20}. Nevertheless, two basic properties for \QMAtwo{}, namely \emph{prover compression}~\cite{HM13} and \emph{witness symmetrization}~\cite{ABDFS09}, still admit stoquastic analogues with only a \emph{quadratic gap loss}:

\begin{theorem}[Robustness of \StoqMAtwo{}, informal version of \Cref{thm:StoqMAk-in-StoqMAtwo,thm:StoqMAk-eq-SymStoqMAk}]
\label{thm-informal:robustness-StoqMAtwo}
    Let $\epsilon(n)$ be the completeness error, and let $\Delta(n)$ be the promise gap. Then the following inclusions hold: 
    \begin{enumerate}[label={\upshape(\arabic*)}]
    \item \label{thmitem:prover-compression}\emph{\textbf{Prover compression}.} For every polynomially bounded integer-valued $k(n) > 2$,
    \[ \StoqMA_{\ell}\ssbra*{k,1-\epsilon,\Delta} \subseteq \StoqMA_{k\ell}\ssbra*{2,1-O(\epsilon \Delta),\Omega(\Delta^2)}. \]
    \item \label{thmitem:symmetrize}\emph{\textbf{Witness symmetrization}.}
    For every polynomially bounded integer-valued $k(n) \geq 2$,
    \begin{subequations}
    \begin{align}
        \StoqMA_{\ell}\ssbra*{k,1-\epsilon,\Delta} &\subseteq \SymStoqMA_{O(\ell\log{k})}\ssbra*{k,1-\epsilon-O(\Delta),\Omega(\Delta)}, \label{eq:StoqMAk-in-SymStoqMAk}\\
        \SymStoqMA_{\ell}\ssbra*{k,1-\epsilon,\Delta} &\subseteq \StoqMA_{\ell}\ssbra*{k,1-O(\epsilon\Delta),\Omega(\Delta^2)}. \label{eq:SymStoqMAk-in-StoqMAk}
    \end{align}
    \end{subequations}
    \end{enumerate}
\end{theorem}

When the completeness error is \emph{negligible}, by combining \Cref{thm-informal:StoqMAk-error-reduction-yes,thm-informal:robustness-StoqMAtwo}\ref{thmitem:prover-compression}, we obtain $\StoqMAk_{1-\negl}=\StoqMAtwo_{1-\negl}$ for every polynomially bounded $k\geq 2$.\footnote{As stated in \Cref{thm:StoqMAk-in-StoqMAtwo}, without error reduction for \StoqMAk{}, which is unknown in general, our current prover-compression technique would introduce a quadratic gap loss.}
The proof of \Cref{thm-informal:robustness-StoqMAtwo}\ref{thmitem:prover-compression} adopts the construction underlying \Cref{thm-informal:StoqMAk-error-reduction-yes} and the improved analysis of~\cite{SW22} to implement the product test (cf.~\cite[Protocol 1]{HM13}). Without error reduction, the quadratic gap loss is obtained via optimizing the branch occurrence probability in~\cite[Protocol 2]{HM13}, rather than using the original value $1/2$. To prove \Cref{thm-informal:robustness-StoqMAtwo}\ref{thmitem:symmetrize}, the inclusion in \Cref{eq:SymStoqMAk-in-StoqMAk} is obtained similarly by refining the construction in~\cite[Lemma 4.8]{ABDFS09}. 
The most interesting direction is the inclusion in \Cref{eq:StoqMAk-in-SymStoqMAk}, which improves the witness length from $k\ell$ in~\cite[Lemma 4.8]{ABDFS09} to $O(\ell \log{k})$ while preserving an $\Omega(\Delta)$ promise gap. 

The main idea behind \Cref{eq:StoqMAk-in-SymStoqMAk} is a \emph{coupon-collector-type} branch argument: instead of asking the prover to send $k$ identical copies of the $k$-tuple witness state $\ket{\psi_1}\ket{\psi_2}\cdots\ket{\psi_k}$, the verifier asks for $k$ bundles of \emph{coherent labeled} states, of the form $\rbra[\big]{ \frac{1}{\sqrt{k}}\sum_{j\in[k]} \ket{j}\ket{\psi_j} }^{\otimes r}$. Each bundle contains $r=O(\log{k})$ independent labeled basis branches, so in the branch-weight analysis the verifier encounters many labels from $[k]$. 
By the coupon-collector intuition, a constant total branch weight is supported on configurations in which these labels contain enough information to select one branch labeled $1$, one branch labeled $2$, and so on, using distinct bundles. When this event occurs, the verifier routes the chosen registers to the original \StoqMAk{} verifier; otherwise, it applies a fixed dummy behavior. This construction preserves soundness and loses only a constant fraction of the original promise gap.

\paragraph{Lower bounds via distribution testing.}
The connection between distribution testing and \StoqMA{} was originally revealed in~\cite{Liu21}. More precisely, any stoquastic verification circuit can be viewed as a \emph{reversible branch-overlap test}, implicitly in~\cite[Theorem 22]{Liu21}: given a pair of reversible circuits $(R_0,R_1)$ and a non-negative witness state $\ket{\Psi}$, the corresponding stoquastic verification circuit accepts with probability  
\begin{equation}
    \label{eq:RCD-intro}
    \frac{1+\innerprod{R_0(\Psi)}{R_1(\Psi)}}{2}, \quad\text{where } \quad \ket{R_b(\Psi)} \coloneqq R_b \ket{\Psi}\ket{\bar{0}}\ket{\bar{+}}, \;\forall b\in\binset.
\end{equation}

We extend this viewpoint to the \emph{symmetric} unentangled setting \SymStoqMAk{}, where the witness state is constrained to have the form $\ket{\Psi} \coloneqq \ket{\psi}^{\otimes K}$.\footnote{Although $K=k$, we use $K$ here to emphasize $K$ identical copies of $\ket{\psi}$, rather than the number of provers.} This variant is equivalent to \StoqMAk{} up to a quadratic gap loss (\Cref{thm-informal:robustness-StoqMAtwo}\ref{thmitem:symmetrize}). Instead of using one witness to induce \emph{one distributional object}, the symmetric tensor-power witness creates \emph{many coherent factor branches} whose squared amplitudes behave as a product distribution. 

Consequently, our lower-bound constructions can be viewed as coherent analogues of distribution testing. In a classical distribution-testing algorithm, one draws independent samples from an unknown distribution. In \SymStoqMAk{}, however, no such intermediate measurement is performed. 
Let $G=(V,E)$ be the constraint graph, where $V$ contains $n$ vertices. Instead, the verifier receives the witness $\ket{\Psi}$, where a computational-basis branch is labeled by a tuple
\[ \rbra*{(v_1,a_1),\ldots,(v_K,a_K)}\in (V\times \Sigma)^K. \]
We call $(v_i,a_i)$ the \emph{vertex-label branch} in the $i$-th tensor factor of $\ket{\Psi}$, and the full tuple a \emph{$K$-vertex-label branch}. The squared amplitude of this $K$-vertex-label branch is
\[ \prod_{i=1}^K p(v_i,a_i), \quad\text{where } p(v_i,a_i)\coloneqq\alpha_{v_i,a_i}^2 \;\;\text{and}\;\; \ket{\psi}=\sum_{(v,a)\in V\times\Sigma}\alpha_{v,a}\ket{v,a}. \]

Therefore the coherent branch weights are governed by the product distribution $p^{\otimes K}$. This product-distribution interpretation is used only in the analysis: it describes the total squared amplitude of accepting or rejecting branches, not outcomes of actual measurements. This distinction is important, since actual polynomially many classical samples would suggest an \MA{} containment, as in~\cite[Section 3.1]{Liu21}. 

\vspace{1em}
We first explain \Cref{thm-informal:StoqMAtwo-short-proofs}\ref{thmitem:StoqMAtwo-short-vanishGap} as a warm-up, which can be viewed as a stoquastic analogue of~\cite{BT07,Beigi10,CF13,LNN12}. Our starting point is the \textsc{Gap Constraint Graph Problem} (\GapCG{}) used in Dinur's PCP theorem~\cite{Dinur07}. For a satisfying labeling $\iota \colon V \to \Sigma$, where $(\iota(u),\iota(v)) \in R_{uv}$ for each $(u,v)\in E$, the honest one-copy witness is:
\[ \ket{\phi_\rmL} =\frac{1}{\sqrt{|V|}}\sum_{v\in V}\ket{v,\iota(v)}.\]
The verifier receives the witness state $\ket{\Phi} \coloneqq \ket{\phi_\rmL}\otimes\ket{\phi_\rmL}$. A cheating witness may try to concentrate branch weight on a small region of the graph, making an unsatisfying labeling appear satisfying. This observation motivates the following simple branch-local protocol $\calA_{\branch}$ for a \emph{fixed $2$-vertex-label branch} $((u,a),(v,b))$. The protocol chooses the uniformity or consistency test with probability $1/2$ each and accepts if the chosen test accepts: 
\begin{itemize}
    \item \textbf{Uniformity test:} Reject if $u=v$;
    \item \textbf{Consistency test:} Reject if $u=v$ and $a\neq b$; or $(u,v)\in E$ and $(a,b) \notin R_{uv}$. 
\end{itemize}

Since only two vertex-label branches are present, a direct calculation shows that the relevant collision and edge-detection probabilities are of order $1/\poly(n)$. To implement $\calA_\branch$ as a stoquastic verification circuit, we express the two acceptance conditions corresponding to the above tests, as predicates $A_\unif$ and $A_\cons$, implement them by classical reversible circuits $\Gamma_\unif$ and $\Gamma_\cons$, and combine these circuits using an additional $\ket{+}$ ancillary qubit to obtain a reversible circuit $\Gamma_\branch$ for the uniform mixture of the two tests. The actual stoquastic verifier is then obtained by applying the branch-overlap test to the circuit pair $(\Gamma_\branch,I)$ on the input state $\ket{+}\ket{\Psi}\ket{\bar{0}}$, where $\ket{\Psi}=\ket{\Phi}$ for \emph{yes} instances. 

\vspace{1em}
We next explain \Cref{thm-informal:StoqMAtwo-short-proofs}\ref{thmitem:StoqMAtwo-short-constantGap}, a stoquastic analogue of~\cite{ABDFS09,CD10}. We use the same starting point, \GapCG{} from~\cite{Dinur07}, and the same one-copy witness $\ket{\phi_\rmL}$. To achieve a constant promise gap, the short-proof protocols of~\cite{CD10} use $\widetilde{O}(\sqrt{n})$ identical copies of $\ket{\phi_\rmL}$, which can be compressed into two $\widetilde{O}(\sqrt{n})$-qubit proofs using the prover-compression technique~\cite{HM13}. Their perfect completeness relies on (at least) Bell-type measurements, while constant soundness follows from a birthday-paradox analysis. 

To stoquastize these protocols, the main technical challenge is that the uniformity test is typically implemented using the quantum Fourier transform~\cite{ABDFS09,CD10}, and hence is not stoquastic. One might instead use the two-branch uniformity test from $\calA_\branch$, but this gives only a $1/\poly(n)$ promise gap. To recover a constant gap, we use \emph{Paninski's coincidence-based uniformity-testing idea}~\cite{Paninski08}: in distribution testing, deviations from uniformity can be detected from the collision pattern among very sparse samples, with the optimal $\Theta(\sqrt n)$ behavior for constant-distance testing. Paninski's test is particularly useful because it separates marginal uniformity from label consistency, making the soundness analysis modular and leading to a predicate $\widehat{A}_\unif$. With a refined soundness analysis, a similar birthday-paradox mechanism applies to coherent $K$-vertex-label branches.
We thus obtain a branch-local protocol $\widehat{\calA}_{\branch}$ for a fixed $\widetilde{O}(\sqrt{n})$-vertex-label branch: the protocol accepts if both tests $\widehat{A}_\unif$ and $A_\cons$ accept. The actual stoquastic verifier is then obtained as before by applying the branch-overlap test to the circuit pair $\rbra[\big]{\widehat{\Gamma}_\branch,I}$, where $\widehat{\Gamma}_\branch$ is the classical reversible circuit for $\widehat{A}_\unif \wedge A_\cons$. 

\paragraph{Upper bounds via rectangular closure testing.}
To derive upper bounds for $\StoqMAtwo_1$, a natural starting point is to extend the known random-walk approach used to prove $\StoqMA_1 \subseteq \MA$ to the unentangled setting, which relies crucially on the local optimality phenomenon of stoquastic $k$-\SAT{}. This  approach~\cite{BBT06,BT10} uses this structure \emph{dynamically}: perfect completeness can be interpreted as the existence of a clean connected component in the configuration graph of a stoquastic $k$-SAT instance, and a random walk from a single seed string suffices to test this component. 

However, this local structure does not directly extend to \StoqMAtwo{}. Given a pair of reversible circuits $(R_0,R_1)$, \Cref{eq:RCD-intro} allows us equivalently to consider the circuit pair $(\Gamma,I)$, where $\Gamma=R_0^\dagger R_1$. Consequently, for any $2\ell$-qubit non-negative witness state $\ket{\psi_1}\otimes\ket{\psi_2}$, perfect completeness yields the following identity, with $r=m_+$ in the notation of \Cref{thm-informal:upper-bound-StoqMAtwo-noYesError}:
\begin{equation}
    \label{eq:projection-intro}
    \Gamma \rbra*{\ket{\psi_1}\otimes\ket{\psi_2}}\ket{\bar{0}}\ket{+}^{\otimes r} = \rbra*{\ket{\psi_1}\otimes\ket{\psi_2}}\ket{\bar{0}}\ket{+}^{\otimes r}.
\end{equation}
Since $\Gamma$ is a bijective permutation, \Cref{eq:projection-intro} implies \emph{support invariance}. More precisely, if $\ket{\psi_1}$ and $\ket{\psi_2}$ have supports $S$ and $T$, respectively, then the joint support is the rectangle $S\times T$; moreover, one may take the corresponding product subset state $\ket{S}\otimes \ket{T}$, where the subset state of $S \subseteq \binset^\ell$ is denoted by $\ket{S}=\frac{1}{\sqrt{|S|}}\sum_{j\in S} \ket{j}$. 
This rectangular structure is \emph{global} and therefore breaks the direct \emph{local} optimality analogy with $\StoqMA_1$: a clean connected component may be highly non-rectangular and hence may not support any product witness.

Nevertheless, our combinatorial approach recovers a \emph{global} analogue of \emph{local} random-walk exploration by testing a property we call \emph{rectangular closure}. The key observation from \Cref{eq:projection-intro} is that, whenever $(s,t,u)$ satisfies $s\in S$, $t\in T$, and $u\in U\coloneqq \binset^r$, every transition of $\Gamma$ from $(s,t,\bar{0},u)$ to  $(s',t',a_0,u')$ must satisfy $a_0=\bar{0}$ and $(s',t',u') \in S\times T\times U$. In other words, a closed rectangle never sends a valid configuration outside the zero-ancilla sector. Our rectangular closure testing algorithm $\calT$ formalizes this idea: starting from every seed pair $(s_0,t_0)$, it repeatedly closes the current rectangle under all transitions of $\Gamma$, adding every pair $(s',t')$ forced by the transition relation underlying $\Gamma$. 
This \emph{deterministic} process tests not only pairs reached by an ordinary walk, but also all crossed pairs required by separability.

The resulting algorithm is therefore \emph{global} rather than \emph{local}. For \emph{yes} instances, \Cref{eq:projection-intro} guarantees that a seed inside a closed rectangle never changes the zero-ancilla register. For \emph{no} instances, the soundness promise implies that every non-closed rectangle has noticeable escaping mass; by reversibility, this escaping mass forces multiplicative growth of the rectangular closure. If $\gamma$ denotes the rejection probability for \emph{no} instances, which is comparable to the promise gap $\Delta$, then it suffices to run $L=\Theta\rbra*{\ell/\gamma}$ closure rounds so that every seed encounters a nonzero ancilla transition. 
Implementing our $L$-round rectangular closure testing algorithm recursively yields an upper bound that is \emph{simultaneously} exponential-time and polynomial-space for inverse-polynomial gaps, while an explicit-table implementation yields a \emph{singly} exponential-time upper bound for exponentially small gaps. 

Finally, by replacing exact invariance (\Cref{eq:projection-intro}) with an approximate version, one can extend our upper bounds to the nearly perfect-completeness regime. Combining our upper bounds with our prover-compression technique (\Cref{corr:StoqMAk-eq-StoqMAtwo-small-yesError}) establishes \Cref{thm-informal:upper-bound-StoqMAtwo-noYesError}.

\paragraph{A general upper bound via Sum-of-Squares.}
The upper bound of $\StoqMAk\subseteq\EXP$ is a corollary of the Sum-of-Squares algorithm of Barak, Kelner, and Steurer~\cite{BKS14}. 

After passing to the symmetric form of the protocol, the acceptance probability can be written as a tensor optimization over the unit sphere
\[
    \max_{\|x\|=1} \Tr M (xx^T)^{\otimes t},
\]
where $M$ acting on $(\mathbb C^d)^{\otimes t}$ is the entrywise non-negative matrix induced by the stoquastic verifier, $t$ is the tensor order corresponding to the number of provers, and $d=2^\ell$ is the dimension of each proof. The BKS algorithm relaxes this optimization to a bounded-degree pseudoexpectation and rounds it using the non-negative vector
$x_i^*=\sqrt{\widetilde{\mathbb E}[x_i^2]}.$
The role of non-negativity is crucial here: it allows the rounded vector to be compared directly with the pseudoexpectation without cancellations from signs or phases.

At a high level, the analysis views the pseudoexpectation
as defining a joint exchangeable distribution on random variables $A_1,\ldots,A_t$. If these variables are close to independent, then the above direct rounding already gives a good approximation. 
If they are noticeably correlated, the BKS conditioning step reweights the pseudo-distribution by a square monomial, which decreases an entropy potential. Since the entropy is bounded by $\log d$, only finitely many conditioning steps are possible before rounding succeeds. BKS's analysis yields a degree-$O(t^3\log d)$ SoS upper bound, which corresponds to an SDP of size $\exp(O(t^3\log^2 d))$.

We slightly tighten BKS's analysis, exploiting the symmetry in the entropy argument, avoiding a wasteful triangle inequality, yielding an improved degree-$O(t^2\log d)$ SoS upper bound. Although the improvement looks quantitatively modest, it makes the parameter
dependence \emph{tight}. Combined with
\Cref{thm-informal:StoqMAtwo-short-proofs}, the degree-$O(t^2\log d)$
SoS gives rise to a $\SAT$ algorithm which matches the ETH barrier up to logarithmic and gap-dependent factors:
an improvement to $\exp((t\log d)^{2-\Omega(1)})$ time of the SoS algorithm, or any improvement of our protocols~\Cref{thm-informal:StoqMAtwo-short-proofs} with shorter proofs or smaller number of provers would yield a subexponential-time
algorithm for $\SAT$. 

\subsection{Discussion and open problems}

Our results in \Cref{thm-informal:StoqMAtwo-short-proofs} stoquastize known short unentangled quantum proofs for \NP{}, with either constant gap~\cite{ABDFS09,CD10,NZ23} or vanishing gap~\cite{BT07,Beigi10,CF13,LNN12}, and push the upper bound down to \MA{}, which is the best possible target under the standard derandomization assumption $\MA=\NP$. This short-proof picture raises a natural question of whether short unentangled quantum proofs can yield lower bounds beyond \NP{}.

An intriguing target is \QCMA{}. One possible starting point is the \textsc{Ground State Connectivity Problem} (\GSCON{}): as established in~\cite{CNS18}, $\GSCON \in \QMA_{\widetilde{O}(n)}\ssbra*{2,1,1/\poly(n)}$, which suggests the following question:
\begin{enumerate}[label={\upshape(\roman*)}]
    \item \label{open:QCMA-lower-bound} Can this protocol be made substantially shorter, or have constant promise gap (assuming PCP for \QCMA{})? More broadly, can one identify a natural \QCMA{}-complete or \QMA{}-complete problem that admits a genuinely short \QMAtwo{} proof system?
\end{enumerate}

From the algorithmic perspective, our upper bounds leave open the possibility of better time-space trade-off algorithms for \StoqMAtwo{} and, more generally, for non-negative tensor optimization problems. The Barak--Kelner--Steurer algorithm~\cite{BKS14}, with a refined analysis in~\Cref{thm-informal:StoqMAk-in-EXP}, achieves near-optimal exponential-time dependence under ETH, but its space usage is not obviously polynomial.
In parameter regimes where the ETH-based optimality bound does not apply, \Cref{thm-informal:upper-bound-StoqMAtwo-noYesError} matches the dependence on the parameters $k$ and $\ell$ of the upper bound that is ETH-optimal in the applicable regimes, while using only polynomial space. By contrast, our fastest algorithm in \Cref{thm-informal:upper-bound-StoqMAtwo-noYesError} achieves an exponential improvement over the ETH-optimal bound in the parameters $k$ and $\ell$. A natural open problem is therefore:
\begin{enumerate}[label={\upshape(\roman*)}]
    \setcounter{enumi}{1}
    \item Is there a randomized or deterministic algorithm for \StoqMAtwo{} with proof length $\ell$ and constant gap that runs in time $\exp(\widetilde O(\ell))$ and uses only $\poly(n,\ell)$ space? More generally, can algorithms for non-negative tensor optimization avoid explicitly storing exponentially large matrices or high-degree SDP objects while preserving their time complexity?
\end{enumerate}

\subsection{Related work}

Most prior work connecting to \StoqMA{} concerns stoquastic Hamiltonians. A line of work starting from~\cite{Bravyi15,Liu21} studies restricted variants of stoquastic local Hamiltonian problems, especially under additional promises on the ground state or its approximation that lead to \MA{} containments. This perspective has proved useful beyond the stoquastic setting: analogous \MA{} containments extend to general local Hamiltonians either with constant precision~\cite{GLG22,LG25} or when the ground state admits a succinct representation~\cite{Jiang25}, with improved hardness results in~\cite{WL25}.\footnote{In addition, a related constant-precision setting with an \NP{} containment was studied in~\cite{WFC24}.}  A common formulation of the former includes a guiding state as part of the input, and its stoquastic counterpart was recently studied from several perspectives~\cite{Waite25,HBS26}. 

Further work on stoquastic Hamiltonians has studied the distinction between global and termwise stoquasticity~\cite{IPM+20}, the complexity of curing the sign problem~\cite{KT20,KMP+20}, oracle separations between stoquastic adiabatic and classical computation~\cite{Hastings21,GHV21}, the equivalence of logarithmically many adaptive queries and polynomially many parallel queries to a \StoqMA{} oracle~\cite{GPY20}, and the \StoqMA{}-hardness of geometrically local stoquastic Hamiltonians~\cite{REG26,WB25}. 
In addition, a stoquastic version of the {Consistency of Local Density Matrices} Problem was proposed in~\cite{Liu07}, although its \StoqMA{}-completeness remains open. Very recently, a homology-related \MA{}-complete problem was identified in~\cite{HGR+25}.

For \QMAtwo{}, alternative proofs of parallel repetition for \SepQMAtwo{} appear in~\cite{GSU13,LW17}. Quantum upper bounds for \QMAtwo{} were given in~\cite{GSSSY22,GR25} in terms of certain low levels of different quantum analogues of the polynomial-time hierarchy. In addition, a recent work identifies a natural \QMAtwo{}-complete problem related to quantum channel testing~\cite{CM25}. Finally, another recent work~\cite{BHW25} shows that \QMAtwo{} collapses to \QMA{} assuming a certain non-standard version of the quantum PCP conjecture.

\paragraph{Concurrent and follow-up work.} The concurrent work of~\cite{GR26}, which appeared two days after ours, establishes a complete stoquastic counterpart of~\cite{CS12} with respect to \emph{sparsity}: \textsc{Stoquastic Sparse Hamiltonians Problem} is \StoqMA{}-complete, while its \emph{separable} version is \StoqMAtwo{}-complete. That work also considers a stoquastic version of prover compression~\cite{HM13,SW22}, similar in spirit to \Cref{thm:stoq-k-to-2} of our work. Their result obtains a quadratic loss in the promise gap under the additional assumption that the completeness error is negligible, whereas our corresponding result achieves the same gap loss without this assumption.

The follow-up work of~\cite{GJ26}, which was released fifteen days after our work, establishes a \emph{multi-to-single witness reduction} for stoquastic Merlin--Arthur proof systems, thereby collapsing \StoqMAk{} to \StoqMA{} for every polynomially bounded $k\geq 2$. While our algorithmic upper-bound results (\Cref{thm-informal:upper-bound-StoqMAtwo-noYesError,thm-informal:StoqMAk-in-EXP}) remain state of the art, the complexity-theoretic upper bound for \StoqMAk{} was significantly improved in~\cite{GJ26}. It is important to note, however, that multi-to-\emph{single} witness reduction and prover compression are \emph{different phenomena}. More precisely, the latter is \emph{more efficient} than the former with respect to the proof length $\ell$:
\begin{itemize}
    \item Applying the multi-to-single witness reduction in~\cite[Theorem~1.3]{GJ26} to \StoqMAk{} with proof length $\ell$ results in \StoqMA{} with proof length $O(k^3\ell^2/\Delta^2)$.
    \item Applying prover compression in \Cref{thm-informal:robustness-StoqMAtwo}\ref{thmitem:prover-compression} to the same class results in \StoqMAtwo{} with proof length $k\ell$. 
\end{itemize}
This comparison aligns with the rETH-based quadratic lower bound with respect to the proof length $\ell$ for multi-to-single witness reductions, as sketched immediately after \Cref{thm-informal:StoqMAtwo-short-proofs}.

\paragraph{Recent developments.}
Several developments appeared around three months after the release of our work, or later. 
For unentangled proof systems, an optimal analysis of the product test is given in~\cite{BJW26}, while near-optimal error reduction~\cite{Miyamoto26} strengthens the $\QMA^+(2)=\NEXP$ characterization of~\cite{JW23}.
Dimension-free bosonic argmax stability collapses $\mathsf{PureSuperQMA}$, a candidate intermediate class between \QMA{} and \QMAtwo{}~\cite{KR24}, to \QMA{}~\cite{GJL+26}. Allowing the two provers to share $O(\log n)$ EPR pairs also preserves the power of \QMAtwo{}~\cite{DM26}.

On the optimization side, new sum-of-squares rounding methods improve algorithms for Best Separable State with perfect-completeness and related optimization problems~\cite{JWX26a,AHR26}. Improved de Finetti bounds remove the perfect completeness assumption and yield algorithms running in subexponential time in the local dimension for constant-gap Best Separable State and trace-distance separability testing~\cite{JWX26b}.
Polynomial-time separability algorithms were also obtained for a fixed constant Hilbert--Schmidt distance gap~\cite{Mal26,JWX26b}.

%%%%%%%%%%%%%%%%%%%%%%%%%%%%%%%%%%%%%%%%%%%%%%%%%%%%%%%%%%%
%%%%%%%%%%%%%%%%%%%%%%%%%%%%%%%%%%%%%%%%%%%%%%%%%%%%%%%%%%%
%%%%%%%%%%%%%%%%%%%%%%%%%%%%%%%%%%%%%%%%%%%%%%%%%%%%%%%%%%%

\section{Preliminaries}
\label{sec:preliminaries}

We assume that the reader has basic familiarity with quantum computation and quantum information theory. For background, the textbooks~\cite{NC10} and~\cite{deWolf19} provide useful starting points. We adopt the following notation throughout the paper:
\begin{itemize}
    \item For every positive integer $n$, we write $\sbra{n}\coloneqq \cbra{1,2,\ldots,n}$ and $\ssbra{n}\coloneqq \cbra{0,1,\ldots,n-1}$.
    \item The notation $\widetilde{O}(f)$ stands for $O(f\polylog(f))$. 
    \item The notations $\ln$ and $\log$ denote the natural logarithm and the logarithm to the base $2$, respectively.
    \item A function $\mu\colon \bbN\to\bbR_{\ge 0}$ is \emph{negligible} if, for every integer $c\geq 1$, there exists an integer $n_c>0$ such that $\mu(n)<n^{-c}$ for all $n\geq n_c$.  
    \item The notation $\poly(\cdot)$ denotes an unspecified polynomial, and $\polylog(\cdot)$ denotes an unspecified polylogarithmic factor. We also write $\negl$ for an unspecified negligible function. 
    \item We adopt $\bbI[\cdot]$ to denote the indicator function of an event or predicate. 
\end{itemize}

In addition, we denote a \emph{promise problem} by a pair
$\calI=(\calI_{\yes},\calI_{\no})$ of disjoint subsets of $\binset^*$.
Inputs in $\calI_{\yes}$ are called \emph{yes} instances,
and inputs in $\calI_{\no}$ are called \emph{no} instances. Whenever convenient, we write $x\in\calI$ to mean that $x\in \calI_{\yes}\cup\calI_{\no}$. No condition is imposed on inputs outside $\calI_{\yes} \cup \calI_{\no}$. 

\paragraph{Linear algebra conventions.}
We use $\ket{\bar{0}}$ and $\ket{\bar{+}}$ as shorthand for $\ket{0}^{\otimes a}$ and $\ket{+}^{\otimes a}$, respectively, whenever $a>1$ and is clear from the context. We also adopt the notation $\norm{\ket{\psi}}_2^2 \coloneqq \innerprod{\psi}{\psi}$ for the Euclidean norm of a vector $\ket{\psi}$ in Dirac notation. 

\parheadingItem{Non-negative states and matrices.}
A $d$-dimensional quantum state $\ket{\psi}$, viewed as a unit vector, is called \textit{non-negative} if $\innerprod{x}{\psi} \geq 0$ for all $x \in [d]$. 
For any non-negative state $\ket{\psi}$, the support of $\ket{\psi}$ is denoted by $\supp(\ket{\psi}) \coloneqq \set{s \in [d]}{\innerprod{s}{\psi} > 0}$, which consists of the basis indices (typically binary strings) corresponding to nonzero amplitudes. 
For any nonempty set $S \subseteq \{0,1\}^n$ with $d=2^n$, the corresponding \emph{subset state}, a notion first considered in~\cite[Section 3]{Wat00}, is defined by $\ket{S}  \coloneqq \frac{1}{\sqrt{|S|}} \sum_{i \in S} \ket{i}$. 
In addition, a $d_1\times d_2$ matrix $M$ is called \textit{entrywise non-negative} if $\bra{x}M\ket{y} \geq 0$ for all $x\in[d_1]$ and $y\in[d_2]$. 

\parheadingItem{Matrices and matrix norms.} A square matrix $A$ is \emph{Hermitian} if $A=A^\dagger$. A Hermitian matrix $A$ is \emph{positive semi-definite}, written $A\succeq 0$ or $0\preceq A$, if $\bra{\psi}A\ket{\psi}\geq 0$ for every vector $\ket{\psi}$. Moreover, we say that a positive semi-definite matrix $A$ is a density matrix if $\Tr(A)=1$. 
For Hermitian matrices $A$ and $B$ of the same dimension, we write $A\preceq B$ if $B-A$ is positive semi-definite. 
In addition, we denote the trace norm by $\norm{A}_1 \coloneqq \Tr(|A|)$, the Frobenius norm by $\norm{A}_2 \coloneqq \sqrt{\Tr\rbra*{A^\dagger A}}$, and the operator norm by $\norm{A}_{\infty} \coloneqq \sigma_{\max}(A)$, where $\sigma_{\max}(A)$ denotes the largest singular value of $A$. For simplicity, we also write $\norm{A} = \norm{A}_{\infty}$. 

\parheadingItem{State spaces.}
Let $\calS(d)$ be the set of $d$-dimensional unit vectors, and let
$\calS_+(d)$ be the set of $d$-dimensional \emph{non-negative} unit
vectors. Similarly, let $\calD(d)$ be the set of $d$-dimensional density
matrices, and let $\calD_+(d)$ be the set of $d$-dimensional
\emph{entrywise non-negative} density matrices.

\subsection{Information-theoretic measures}

We first list several useful information-theoretic measures for quantum states, together with their basic properties. 

\begin{definition}[Trace distance]
    Let $\rho_0$ and $\rho_1$ be quantum states in $\calD(N)$.
    The trace distance between $\rho_0$ and $\rho_1$ is defined by
    \[\td(\rho_0,\rho_1) \coloneqq \frac{1}{2}\Tr\rbra{|\rho_0-\rho_1|}.\]
\end{definition}

\begin{definition}[Fidelity]
    Let $\rho_0$ and $\rho_1$ be quantum states in $\calD(N)$.
    The (Uhlmann) fidelity and the squared fidelity between $\rho_0$ and $\rho_1$ are defined by
    \[\F(\rho_0,\rho_1) \coloneqq \Tr\abs{\sqrt{\rho_0}\sqrt{\rho_1}} = \Tr\sqrt{\sqrt{\rho_0}\rho_1\sqrt{\rho_0}} \quad\text{and}\quad \F^2(\rho_0,\rho_1) \coloneqq \F(\rho_0,\rho_1)^2.\]
\end{definition}
When one of the states is \emph{pure}, say $\rho_0=\ketbra{\psi}{\psi}$ and $\rho_1=\rho$, it is easy to verify that
\begin{equation}
    \label{eq:one-pure-state}
    \F^2(\ketbra{\psi}{\psi},\rho) = \bra{\psi}\rho\ket{\psi}.
\end{equation}
When both states are pure, the following identity holds; see, e.g.,~\cite[Section 9.2.3]{NC10}:
\begin{equation}
    \label{eq:pure-pure-Fuchs-vanGraaf}
    \td(\ketbra{\psi_0}{\psi_0},\ketbra{\psi_1}{\psi_1}) = \sqrt{1-\F^2\rbra*{\ketbra{\psi_0}{\psi_0},\ketbra{\psi_1}{\psi_1}}}.
\end{equation}
By the data-processing inequality for the trace distance, e.g.,~\cite[Theorem 9.2]{NC10}, we obtain:
\begin{lemma}[Measurement bound for trace distance]
    \label{lemma:meas-bound-td}
    Let $\rho_0$ and $\rho_1$ be quantum states in $\calD(N)$.
    For every two-outcome measurement $\cbra*{\Pi,I-\Pi}$ with $0 \preceq \Pi \preceq I$, it holds that
    \[ \abs*{\Tr(\Pi\rho_0) - \Tr(\Pi\rho_1)} \leq \td(\rho_0,\rho_1). \]
\end{lemma}

\vspace{1em}
Next, we list several information-theoretic measures for probability distributions, together with their basic properties. 

\begin{definition}[Total variation distance]
    Let $p_0$ and $p_1$ be probability distributions over $[N]$. The total variation distance between $p_0$ and $p_1$ is defined by
	\[\TV(p_0,p_1) \coloneqq \frac{1}{2}\sum_{x\in[N]} |p_0(x)-p_1(x)|.\]
\end{definition}

\begin{definition}[Hellinger distance]
    Let $p_0$ and $p_1$ be probability distributions over $[N]$. The squared Hellinger distance and Hellinger distance between $p_0$ and $p_1$ are defined by
	\[\dH^2(p_0,p_1) \coloneqq 1-\sum_{x\in[N]}\sqrt{p_0(x)p_1(x)} \quad\text{and}\quad \dH(p_0,p_1) \coloneqq \sqrt{\dH^2(p_0,p_1)}.\]
\end{definition}

\begin{definition}[Kullback--Leibler divergence]
    Let $p_0$ and $p_1$ be probability distributions over $[N]$. The Kullback--Leibler (KL) divergence between $p_0$ and $p_1$ is defined by
	\[\KL(p_0\|p_1) \coloneqq \sum_{x \colon p_0(x)>0} p_0(x) \log\frac{p_0(x)}{p_1(x)}.\]
\end{definition}

\begin{lemma}[Adapted from~{\cite[Equation (3.3.9)]{Reiss89}}]
    \label{lemma:Hsquare-leq-KL}
    Let $p_0$ and $p_1$ be probability distributions over $[N]$. Then the following inequality holds:
    \[ 2\dH^2(p_0,p_1) \leq \KL(p_0\|p_1). \]
\end{lemma}

\begin{definition}[Shannon entropy]
    Let $p$ be a probability distribution over $[N]$. The Shannon entropy of $p$ is defined by 
    \[ \H(p) \coloneqq - \sum_{x \in \sbra{N}} p\rbra{x} \log \rbra{p\rbra{x}}, \]
    with the convention that $0\log(0)=0$.   
    
    \noindent Equivalently, if $X$ is a random variable taking values in $\Omega$, then the law, or distribution, of $X$ is the probability distribution $p_X$ over $\Omega$ denoted by $p_X(x)$, and we write
    \[ \H(X) \coloneqq \H(p_X), \quad\text{where}\;\; p_X(x) \coloneqq \Pr[X=x]. \]    
\end{definition}

\begin{definition}[Conditional entropy and mutual information]
    For jointly distributed random variables $X$ and $Y$, the law of $(X,Y)$ is the joint distribution $p_{X,Y}$, and we define the joint entropy and the conditional entropy of $X$ given $Y$ by
    \[ \H(X,Y) \coloneqq \H(p_{X,Y}) 
    \quad\text{and}\quad 
    \H(X\mid Y) \coloneqq \H(X,Y) - \H(Y). \]
    The mutual information between $X$ and $Y$ is defined by
    \[ \rmI(X;Y) \coloneqq \H(X)+\H(Y)-\H(X,Y). \]
\end{definition}

\subsection{(Stoquastic) separable values of square Hermitian matrices}
\label{subsec:separable-values}

Let $M$ be a square Hermitian matrix of dimension $d_1d_2\cdots d_k$. 
Let $\Sep(d_1,\cdots,d_k)$ and $\Sep^+(d_1,\cdots,d_k)$ denote, respectively, the set of separable states and the set of separable \emph{non-negative} states on $\calD\rbra[\big]{\prod_{j\in[k]} d_j}$:
\begin{align*}
    \Sep(d_1,\cdots,d_k) &\coloneqq \conv\set{\ketbra{\psi_1}{\psi_1}\totimes\cdots\totimes\ketbra{\psi_k}{\psi_k}}{\ket{\psi_j}\in\calS(d_j) \text{ for every } j\in[k] },\\
    \Sep^+(d_1,\cdots,d_k) &\coloneqq \conv\set{\ketbra{\psi_1}{\psi_1}\totimes\cdots\totimes\ketbra{\psi_k}{\psi_k}}{ \ket{\psi_j}\in\calS_+(d_j) \text{ for every } j\in[k] },
\end{align*}
where $\conv(S)$ denotes the convex hull of a set $S$. 

Following~\cite[Equation (5)]{HM13}, we define the corresponding support functions, called the \emph{separable value} and the \emph{stoquastic separable value}, by
\begin{subequations}
\label{eq:separable-values}
\begin{align}
    \hsep{d_1}{\cdots,d_k}(M) &\coloneqq \max_{\ket{\psi_1}\in\calS(d_1),\cdots,\ket{\psi_k}\in\calS(d_k)} \Tr\rbra*{M \rbra{\ketbra{\psi_1}{\psi_1}\totimes \cdots \totimes \ketbra{\psi_k}{\psi_k}} },\\
    \hsepPlus{d_1}{\cdots,d_k}(M) &\coloneqq \max_{\ket{\psi_1}\in\calS_+(d_1),\cdots,\ket{\psi_k}\in\calS_+(d_k)} \Tr\rbra*{M \rbra{\ketbra{\psi_1}{\psi_1}\totimes \cdots \totimes \ketbra{\psi_k}{\psi_k}} }.
\end{align}
\end{subequations}

By definition, it follows that $\hsepPlus{d_1}{\cdots,d_k}(M) \leq \hsep{d_1}{\cdots,d_k}(M)$.
When all proofs have the same dimension, namely $d_1=\cdots=d_k=d$, we adopt the notation 
\[\hsepM{k}{d}(M) \coloneqq \hsep{d}{\cdots,d}(M) \quad\text{and}\quad \hsepPlusM{k}{d}(M) \coloneqq \hsepPlus{d}{\cdots,d}(M)\] 
for convenience.
From \Cref{eq:separable-values}, it is straightforward to see the following identity, which also immediately applies to $\hsepPlus{d_1}{\cdots,d_k}$: 

\begin{proposition}
    \label{prop:hsep-shift}
    For all real numbers $a,b \geq 0$, it holds that
    \[ \hsep{d_1}{\cdots,d_k}(aM+bI) = a \cdot \hsep{d_1}{\cdots,d_k}(M) + b.\]
\end{proposition}

\subsection{Stoquastic algorithmic toolkit}

The following definition is implicitly given in~\cite[Definition 4]{BBT06}:
\begin{definition}[Stoquastic verification circuit]
    \label{def:stoq-verifier}
    We say that an $n$-qubit quantum circuit $Q$ is a \emph{stoquastic verification circuit} if $Q$ acts on the registers $\sfW$, $\sfA_0$, and $\sfA_+$, uses only Toffoli, CNOT, and X gates before the final measurement, and then measures the designated output qubit in the register $\sfO$ in the Hadamard basis. The circuit $Q$ accepts if and only if the final measurement yields the outcome $+$. The witness register $\sfW$ contains an $m_w$-qubit state, and the ancillary registers $\sfA_0$ and $\sfA_+$ are initialized to $\ket{0}^{\otimes m_0}$ and $\ket{+}^{\otimes m_+}$, respectively, where $n=m_w+m_0+m_+$. The remaining qubits are contained in the register $\sfR$, which has $n-1$ qubits.     
\end{definition}

The SWAP test was originally proposed for pure states in~\cite{BCWdW01} and was later shown to extend to mixed states~\cite{KMY09}: 

\begin{lemma}[SWAP test for mixed states, adapted from~{\cite[Proposition 9]{KMY09}}]
\label{lemma:swap-test}
    Let $\rho_0$ and $\rho_1$ be two $m$-qubit quantum states, which may in general be \emph{mixed}. Then there exists an explicit $(2m+1)$-qubit stoquastic circuit that accepts with probability $\rbra*{1+\Tr(\rho_0\rho_1)}/2$. This circuit implementation uses a single copy of each $\rho_0$ and $\rho_1$, together with $m$ Toffoli gates and $2m$ CNOT gates.\footnote{A Fredkin gate, namely, a controlled-swap gate on two qubits, can be implemented exactly using two CNOT gates and a Toffoli gate; see, e.g., \cite[Exercise 4.25]{NC10}. \label{footnote:Fredkin}} 
\end{lemma}

We also need the reversible branch-overlap test, which we refer to simply as the \emph{branch-overlap test}, as stated in \Cref{lemma:branch-overlap-test}. This test is implicit in~\cite[Theorem 22]{Liu21}, with an explicit circuit implementation given in~\cite[Figure 1]{Liu21}, and was originally inspired by the SWAP test~\cite{BCWdW01}. 

\begin{lemma}[Reversible branch-overlap test, adapted from~{\cite[Theorem 22]{Liu21}}]
    \label{lemma:branch-overlap-test}
    Let $R_0$ and $R_1$ be polynomial-size classical reversible circuits consisting only of Toffoli, CNOT, and X gates. Assume that both $R_0$ and $R_1$ act on the same $m(n)$-qubit full input state 
    \[\ket{\Psi_\full} \coloneqq \ket{\Psi}\otimes \ket{0}^{\otimes m_0} \otimes \ket{+}^{\otimes r},\] 
    where $\ket{\Psi}$ is the non-negative input state, and $m_0(n)$ and $r(n)$ are polynomially bounded integer-valued functions. For $b\in\binset$, define $\ket{R_b(\Psi)} \coloneqq R_b \ket{\Psi_\full}$. Then there exists an explicit $(m+1)$-qubit stoquastic circuit that accepts with probability $\rbra*{1+\innerprod{R_0(\Psi)}{R_1(\Psi)}}/2$. This circuit implementation uses a single query to each of $R_0$ and $R_1$, together with one additional $\ket{+}$ ancillary qubit and a single $X$ gate. 
\end{lemma}

%%%%%%%%%%%%%%%%%%%%%%%%%%%%%%%%%%%%%%%%%%%%%%%%%%%%%%%%%%%
%%%%%%%%%%%%%%%%%%%%%%%%%%%%%%%%%%%%%%%%%%%%%%%%%%%%%%%%%%%
%%%%%%%%%%%%%%%%%%%%%%%%%%%%%%%%%%%%%%%%%%%%%%%%%%%%%%%%%%%

\section{Parallel repetition in stoquastic Merlin--Arthur proof systems}
\label{sec:parallel-repetition-StoqMAk}

We begin by defining \StoqMAk{} via \emph{stoquastic verification circuits} (\Cref{def:stoq-verifier}), which naturally extends~\cite[Definition 4]{BBT06} to the setting of \emph{unentangled} witness states: 

\begin{definition}[\StoqMAk{}]
\label{def:StoqMAk}
A promise problem $\calI = (\calI_{\yes}, \calI_{\no})$ is in $\StoqMA_\ell(k,c,s)$ if there exists a stoquastic verifier $V_\calI$ such that for every input $x \in \calI$, the verifier is specified by an $n$-qubit stoquastic verification circuit $V_{\calI}(x)$ that acts on the $k\ell$-qubit witness register $\sfW$ together with the ancillary registers $\sfA_0$ and $\sfA_+$, which are initialized to $\ket{0}^{\otimes m_0}$ and $\ket{+}^{\otimes m_+}$, respectively, where $n=k\ell+m_0+m_+$, and produces a Hadamard-basis measurement outcome on the single-qubit register $\sfO$ together with the remaining qubits in the register $\sfR$, such that the following hold for efficiently computable functions $c(n)$, $s(n)$, and $\ell(n)$: 
\begin{description}
    \item[Completeness. ] If $x\in \calI_{\yes}$, then there exists a collection of unentangled states $\ket{\psi_1}, \cdots, \ket{\psi_k}$,  each on at most $\ell(n)$ qubits and with non-negative amplitudes such that 
    \[\Pr\sbra*{V_{\calI}(x) \text{ accepts } \ket{\psi_1}\otimes\cdots\otimes\ket{\psi_k}} \geq c(n).\] 
    \item[Soundness. ] If $x \in \calI_{\no}$, then for every collection of unentangled states $\ket{\psi_1}, \cdots, \ket{\psi_k}$, each on at most $\ell(n)$ qubits and with non-negative amplitudes, we have 
    \[\Pr\sbra*{V_{\calI}(x) \text{ accepts } \ket{\psi_1}\otimes\cdots\otimes\ket{\psi_k}} \leq s(n).\]
\end{description}
Furthermore, $c(n)$ and $s(n)$ satisfy $1/2 \leq s(n) < c(n) \leq 1$ and $c(n)-s(n) \geq 1/\poly(n)$. For convenience, we write $\StoqMA_{\ell}\ssbra{k,c,\Delta}$ to emphasize the promise gap $\Delta(n) \coloneqq c(n)-s(n)$, abbreviate the stoquastic verification circuit $V_{\calI}(x)$ as $V_x$, and use the subscript $1-\epsilon$ in $\StoqMAk_{1-\epsilon}$ to denote the case of completeness error at most $\epsilon$.
\end{definition}

\begin{remark}[Optimal witness of any \StoqMAk{} verifier is non-negative]
It is worth noting that the restriction in \Cref{def:StoqMAk} to non-negative witness states \emph{does not change} the computational power of the class, because optimal stoquastic separable values can be attained by non-negative states (\Cref{prop:opt-stoqSepVal-nonNeg}). Although this phenomenon appears similar to that in the \StoqMA{} case, the underlying reason is different: in the \StoqMA{} case, the conclusion follows from the Perron-Frobenius theorem (e.g.,~\cite[Remark 2.2]{Liu21}). 
\end{remark}

\paragraph{Parallel repetition for \StoqMAk{}.}
The main result of this section is that \StoqMAk{} admits \emph{parallel-repetition}-type phenomenon: 

\begin{theorem}[Parallel repetition for \StoqMAk{}]
    \label{thm:StoqMAk-w-parallel-repetition-loss}
    For any efficiently computable functions $a(n)$ and $b(n)$ with $0 \leq b(n) < a(n) \leq 1$, and any efficiently computable positive integer-valued functions $k(n)$ and $r(n)$ bounded by $\poly(n)$, the following inclusion holds:
    \[ \StoqMA_{\ell}\ab(k,\frac{1}{2}+\frac{a}{2},\frac{1}{2}+\frac{b}{2}) \subseteq \StoqMA_{r\ell}\ab(k, \frac{1}{2}+\frac{(1+a)^r}{2^{r+1}}, \frac{1}{2}+\frac{(1+b)^r}{2^{r+1}}). \]
\end{theorem}

It is noteworthy that the proof technique underlying \Cref{thm:StoqMAk-w-parallel-repetition-loss} extends to the more general setting in which the parallel repetitions are associated with \emph{distinct} \StoqMAk{} verifiers, and hence with different promise problems. More specifically, given two promise problems $\calI^{(1)}\coloneqq\rbra[\big]{\calI_\yes^{(1)},\calI_\no^{(1)}}$ and $\calI^{(2)}\coloneqq\rbra[\big]{\calI_\yes^{(2)},\calI_\no^{(2)}}$, their \emph{direct product} is defined by:
\[ (\calI_\yes,\calI_\no) \coloneqq \calI^{(1)} \times \calI^{(2)}, \quad \text{where } \calI_\yes\coloneqq \calI_\yes^{(1)} \times \calI_\yes^{(2)} \text{ and } \calI_\no\coloneqq \calI_\no^{(1)} \times \calI_\no^{(2)}. \]

Furthermore, we refer to the corresponding closure property of \StoqMAk{} in the parameter regime of \Cref{thm:StoqMAk-w-parallel-repetition-loss} as \emph{direct product with parameter loss}, as stated in \Cref{thm:StoqMAk-closed-under-direct-product-w-loss}. Despite the parameter loss, a direct calculation yields the following error-reduction consequence for \StoqMAk{} when the completeness error is \emph{negligible}:

\begin{corollary}[Error reduction for $\StoqMAk_{1-\negl}$]
    \label{corr:error-reduction-StoqMAk-negl}
    Let $\delta(n)$ and $b(n)$ be efficiently computable non-negative functions with $0 \leq b(n) < 1-1/\!\poly(n)$ and $b(n) < 1-\delta(n)$, and suppose that $\delta(n)$ is negligible, and let $k(n)$ and $r(n)$ be efficiently computable positive integer-valued functions that are polynomially bounded, where $r(n)$ is chosen as the number of repetitions. Then parallel repetition gives the following inclusion:
    \[ \StoqMA_\ell\ab( k,\frac{1}{2}+\frac{1-\delta}{2},\frac{1}{2}+\frac{b}{2} ) \subseteq \StoqMA_{r\ell}\ab( k,\frac{1}{2}+\frac{1-\delta'}{2}, \frac{1}{2}+\frac{\delta''}{2} ). \] 
    Here, the resulting threshold parameters are 
    \[\delta'(n) \coloneqq 1-\ab(1-\frac{\delta(n)}{2})^{r(n)} \quad\text{and}\quad \delta''(n)\coloneqq \ab(\frac{1+b(n)}{2})^{r(n)}.\]
    Furthermore, $\delta'(n)$ is negligible for every polynomially bounded $r(n)$, while $\delta''(n)$ is negligible whenever $r(n)$ is chosen sufficiently large.
    In particular, to guarantee the soundness error bound $\delta''(n) \leq 2^{-n}$, it suffices to take
    \[r(n) \geq \ceil*{ \frac{n}{\log\ab(2/(1+b(n)))} }.\] 
\end{corollary}

\begin{remark}[Parameter loss in parallel repetition for \StoqMAk{}]
    \label{remark:StoqMAk-parallel-repetition-parameter-loss}
    There is parameter loss in the regime underlying \Cref{thm:StoqMAk-w-parallel-repetition-loss}, whereas some subclasses of \StoqMAk{} admit direct product \emph{without} parameter loss, as is already implicit in~\cite[Section 5]{Liu21}. This closure property holds for the subclass of \StoqMAk{} that we call \ProdStoqMAk{}. The definition of \ProdStoqMAk{} is identical to that of \Cref{def:StoqMAk}, except that for a \ProdStoqMAk{} verifier $V_x$, the maximum acceptance probability $\omega(V_x)$ has the product form $\omega(V_x) = \rbra[\big]{1+\hsepPlus{2^\ell}{\cdots,2^\ell}(M_x)}/2$, where $M_x = \bra{\bar{0}}\bra{\bar{+}} V_x^\dagger X_\sfO V_x \ket{\bar{0}} \ket{\bar{+}}$ and $M_x = \bigotimes_{j\in[k]} M_x^{(j)}$. Here, each factor $M_x^{(j)}$ is Hermitian and entrywise non-negative and corresponds to the $j$-th component of the witness state for $j\in[k]$. Consequently, \ProdStoqMAk{} admits parallel repetition without parameter loss. The detailed statement and proof can be found in an earlier version of our work~\cite[Section 3.4]{LW26}.
\end{remark}

\paragraph{Stoquastic separable values are multiplicative.} 
To prove the soundness parts of \Cref{thm:StoqMAk-w-parallel-repetition-loss}, we establish the following \emph{multiplicative} property of stoquastic separable values under tensor products, whose two cases follow by straightforward induction from \Cref{lemma:stoqSepVal-multiplicative-PSD} and \Cref{lemma:stoqSepVal-multiplicative-product}, respectively: 
\begin{theorem}[Tensor-product multiplicativity of stoquastic separable values]
    \label{thm:stoqSepVal-multiplicative}
    For any entrywise non-negative Hermitian square matrices $M_1,\cdots,M_r$ of dimensions $\prod_{i\in[k]} d^{(1)}_i$, $\cdots$, $\prod_{i\in[k]} d^{(r)}_i$, respectively, the stoquastic separable value is \emph{multiplicative under tensor products}, 
    \begin{equation}
        \hsepPlus{\prod_{i\in[r]} d^{(i)}_1}{\cdots,\prod_{j\in[r]} d^{(j)}_k}\rbra*{M_1\totimes\cdots\totimes M_r} = \prod_{j\in[r]} \hsepPlus{d^{(j)}_1}{\cdots,d^{(j)}_k }(M_j),
    \end{equation}
    provided that one of the following conditions holds: 
    \begin{enumerate}[label={\upshape(\arabic*)}]
        \item All matrices $M_1,\cdots,M_r$ are \emph{positive semi-definite}.
        \item For each $j\in[r]$, there exist Hermitian square matrices $M_1^{(j)}, \cdots, M_k^{(j)}$ of dimensions $d_1^{(j)}, \cdots, d_k^{(j)}$, respectively, such that $M_j$ admits the \emph{product decompositions} 
        \[M_j = M_1^{(j)} \totimes \cdots \totimes M_k^{(j)},\]
        where each factor $M_i^{(j)}$ is Hermitian and entrywise non-negative for all $i\in[k]$ and $j\in[r]$.
    \end{enumerate}
\end{theorem}

\vspace{1em}
In the remainder of this section, we first present the multiplicativity of stoquastic separable values in~\Cref{subsec:stoqSeq-multiplicative}, and then provide two \StoqMAk{}-complete problems arising from the class definition in \Cref{subsec:StoqMAk-complete-probs}. Lastly, we establish the closure property underlying \Cref{thm:StoqMAk-w-parallel-repetition-loss} in \Cref{subsec:StoqMAk-closed-direct-product}. 

\subsection{The multiplicativity of the stoquastic separable values}
\label{subsec:stoqSeq-multiplicative}

We begin by showing that the separable value and the stoquastic separable value, both defined in \Cref{eq:separable-values}, coincide when $M$ is an \emph{entrywise non-negative} matrix:
\begin{proposition}[Optimal stoquastic separable value achieved by non-negative states]
    \label{prop:opt-stoqSepVal-nonNeg}
    For every entrywise non-negative square Hermitian matrix $M$ of dimension $\prod_{j\in[k]} d_j$, it holds that
    \[ \hsep{d_1}{\dots,d_k}(M) = \hsepPlus{d_1}{\cdots,d_k}(M). \]
\end{proposition}

\begin{proof}
    For simplicity, we present the proof only for the case $k=2$, and the argument extends straightforwardly to general $k$. 
    Let 
    \[\ket{\psi_1^\star} \coloneqq \sum_{i\in[d_1]} a_i\ket{i} \quad\text{and}\quad \ket{\psi_2^\star} \coloneqq \sum_{j\in[d_2]} b_j \ket{j}\] 
    be two pure states that achieve the optimum stoquastic separable value. Accordingly, define the corresponding non-negative states 
    \[\ket{\psi_1^+} \coloneqq \sum_i \abs{a_i}\ket{i} \quad\text{and}\quad \ket{\psi_2^+} \coloneqq \sum_j \abs{b_j} \ket{j}.\] 
    Since $M$ is entrywise non-negative and $\hsep{d_1}{d_2}(M)$ is non-negative, we obtain: 
    \begin{subequations}
    \label{eq:pacc-bound}
    \begin{align}
        \hsep{d_1}{d_2}(M) &= \bra{\psi_1^\star}\bra{\psi_2^\star} M \ket{\psi_1^\star}\ket{\psi_2^\star}\\
        &= \sum_{i,j,i',j'} \Real\rbra*{a_i^*b_j^*a_{i'}b_{j'}} \bra{i}\bra{j}M\ket{i'}\ket{j'}\\
        &\leq \sum_{i,j,i',j'} \abs*{a_i^*b_j^*a_{i'}b_{j'}} \bra{i}\bra{j}M\ket{i'}\ket{j'}\\
        &\leq \sum_{i,j,i',j'} \abs{a_i^*}\abs{b_j^*}\abs{a_{i'}}\abs{b_{j'}} \bra{i}\bra{j}M\ket{i'}\ket{j'}\\
        &= \bra{\psi_1^+}\bra{\psi_2^+} M \ket{\psi_1^+}\ket{\psi_2^+}\\ 
        &\leq \hsepPlus{d_1}{d_2}(M) \leq \hsep{d_1}{d_2}(M).
    \end{align}
    \end{subequations}
    Here, the fourth line follows from the triangle inequality. 
    Therefore, the non-negative states $\ket{\psi_1^+}$ and $\ket{\psi_2^+}$ also achieve the desired optimum stoquastic separable value. 
\end{proof}

\paragraph{Stoquastic separable values of positive semi-definite matrices.} We next establish the main technical lemma of this subsection: 

\begin{lemma}[Stoquastic separable values of positive semi-definite matrices are multiplicative]
    \label{lemma:stoqSepVal-multiplicative-PSD}
    For any entrywise non-negative positive semi-definite matrices $M$ and $M'$ of dimensions $\prod_{i\in[k]} d_i$ and $\prod_{j\in[k]} d'_j$, respectively, the stoquastic separable value is multiplicative under tensor products: 
    \[ \hsepPlus{d_1 d'_1}{\cdots, d_k d'_k}\rbra*{M\totimes M'} = \hsepPlus{d_1}{\cdots,d_k }(M) \cdot \hsepPlus{d'_1}{\cdots,d'_k }(M'). \]
\end{lemma}

\begin{proof}
    For simplicity, we prove only the case where $k=2$, and the same argument directly extends to general $k$. Let the operator $M$ act on quantum registers $\sfA$ and $\sfB$, and let the operator $M'$ act on quantum registers $\sfA'$ and $\sfB'$. Following \Cref{prop:opt-stoqSepVal-nonNeg}, it suffices to consider only non-negative quantum states. 

    The lower bound is straightforward. This is because there exist non-negative states $\ket{\psi_1}\in\calS_+\rbra{d_1}$, $\ket{\psi_2}\in\calS_+\rbra{d_2}$, $\ket{\psi'_1}\in\calS_+\rbra{d'_1}$, and $\ket{\psi'_2}\in\calS_+\rbra{d'_2}$ such that the following holds: 
    \begin{align*}
        \Tr\rbra*{M_{\sfA\sfB} \rbra*{ \ketbra{\psi_1}{\psi_1}_{\sfA} \totimes \ketbra{\psi_2}{\psi_2}_{\sfB}} } &= \hsep{d_1}{d_2}(M),\\
        \Tr\rbra*{M'_{\sfA'\sfB'} \rbra*{ \ketbra{\psi'_1}{\psi'_1}_{\sfA'} \totimes \ketbra{\psi'_2}{\psi'_2}_{\sfB'}} } &= \hsep{d'_1}{d'_2}(M').
    \end{align*}

    Then, there exist non-negative states $\ket{\phi_1}_{\sfA\sfA'}\coloneq\ket{\psi_1}_{\sfA}\totimes\ket{\psi'_1}_{\sfA'}\in\calS_+(d_1d'_1)$ and $\ket{\phi_2}_{\sfB\sfB'}\coloneq\ket{\psi_2}_{\sfB}\totimes\ket{\psi'_2}_{\sfB'}\in\calS_+(d_2d'_2)$, which imply the desired lower bound: 
    \begin{align*}
         \hsepPlus{d_1d'_1}{d_2d'_2}(M\totimes M') 
         &\geq \Tr\rbra*{\rbra*{M_{\sfA\sfB}\totimes M'_{\sfA'\sfB'}}\rbra*{\ketbra{\phi_1}{\phi_1}_{\sfA\sfA'}\totimes\ketbra{\phi_2}{\phi_2}_{\sfB\sfB'}}}\\
         &= \hsepPlus{d_1}{d_2}(M) \cdot \hsepPlus{d'_1}{d'_2}(M').
    \end{align*}

    To establish the upper bound, let $\ket{\phi_1}_{\sfA\sfA'}\in\calS_+(d_1d'_1)$ and $\ket{\phi_2}_{\sfB\sfB'}\in\calS_+(d_2d'_2)$ be non-negative states that attain the maximum separable value of $M\totimes M'$. Decompose these states in the computational basis of their first register, $\sfA$ and $\sfB$, respectively:
    \begin{equation}
        \label{eq:hsep-PSD-decomposition}
        \ket{\phi_1}_{\sfA\sfA'} = \sum_{i=1}^{d_1} \ket{i}_{\sfA}\totimes\ket{\eta_i}_{\sfA'} \quad\text{and}\quad \ket{\phi_2}_{\sfB\sfB'} = \sum_{j=1}^{d_2} \ket{j}_{\sfB}\totimes\ket{\vartheta_j}_{\sfB'}.
    \end{equation}
    A direct calculation implies the following normalization equalities: 
    \begin{equation}
        \label{eq:hsep-PSD-normalization}
        \sum_{i=1}^{d_1} \innerprod{\eta_i}{\eta_i}=1 \quad\text{and}\quad \sum_{j=1}^{d_2} \innerprod{\vartheta_j}{\vartheta_j}=1.
    \end{equation}

    Now we expand $\hsepPlus{d_1d'_1}{d_2d'_2}(M\totimes M')$ using the decompositions in \Cref{eq:hsep-PSD-decomposition}: 
    \begin{subequations}
        \label{eq:hsep-PSD-UB}
        \begin{align}
            \hsepPlus{d_1d'_1}{d_2d'_2}(M\totimes M')  &= \bra{\phi_1}_{\sfA\sfA'}\bra{\phi_2}_{\sfB\sfB'} \rbra*{M_{\sfA\sfB}\totimes M'_{\sfA'\sfB'}} \ket{\phi_1}_{\sfA\sfA'}\ket{\phi_2}_{\sfB\sfB'}\\
            &= \sum_{i,j,\hat{i},\hat{j}} \bra{i}_{\sfA}\!\bra{\eta_i}_{\sfA'}\!\bra{j}_{\sfB}\!\bra{\vartheta_j}_{\sfB'} \rbra*{M_{\sfA\sfB}\totimes M'_{\sfA'\sfB'}} \ket{\hat{i}}_{\sfA} \!\ket{\eta_{\hat{i}}}_{\sfA'} \!\ket{\hat{j}}_{\sfB} \!\ket{\vartheta_{\hat{j}}}_{\sfB'}\\
            &\leq \sum_{i,j,\hat{i},\hat{j}} \bra{i}_{\sfA}\!\bra{j}_{\sfB}M_{\sfA\sfB}\ket{\hat{i}}_{\sfA}\!\ket{\hat{j}}_{\sfB} \cdot \abs*{\bra{\eta_i}_{\sfA'}\!\bra{\vartheta_j}_{\sfB'} M'_{\sfA'\sfB'} \ket{\eta_{\hat{i}}}_{\sfA'}\!\ket{\vartheta_{\hat{j}}}_{\sfB'}}.
        \end{align}
    \end{subequations}
    Here, the last line follows from the fact that $M$ is entrywise non-negative. 
    
    Note that the entrywise non-negativity of $M'$ implies that all entries of $M'$ are real. Since $M'$ is positive semi-definite, it follows that $M'$ is real symmetric. Consequently, $M'$ can be viewed as a Gram matrix, and we may therefore apply the Cauchy–Schwarz inequality to the positive semi-definite bilinear form induced by $M'$. In particular, for every $i\in[d_1]$, $j\in[d_2]$, $\hat{i}\in[d_1]$, and $\hat{j}\in[d_2]$, the following holds: 
    \begin{subequations}
        \label{eq:hsep-PSD-UBterm}
        \begin{align}
            &\abs*{\bra{\eta_i}_{\sfA'}\!\bra{\vartheta_j}_{\sfB'} \!M'_{\sfA'\sfB'}\! \ket{\eta_{\hat{i}}}_{\sfA'}\!\ket{\vartheta_{\hat{j}}}_{\sfB'}}\\
            \leq~& \sqrt{\bra{\eta_i}_{\sfA'}\!\bra{\vartheta_j}_{\sfB'} \!M'_{\sfA'\sfB'}\! \ket{\eta_i}_{\sfA'}\!\ket{\vartheta_j}_{\sfB'}} \cdot \sqrt{\bra{\eta_{\hat{i}}}_{\sfA'}\!\bra{\vartheta_{\hat{j}}}_{\sfB'} \!M'_{\sfA'\sfB'}\! \ket{\eta_{\hat{i}}}_{\sfA'}\!\ket{\vartheta_{\hat{j}}}_{\sfB'}}\\
            \leq~& \hsep{d'_1}{d'_2}(M') \cdot \sqrt{\innerprod{\eta_i}{\eta_i}} \cdot \sqrt{\innerprod{\vartheta_j}{\vartheta_j}} \cdot \sqrt{\innerprod{\eta_{\hat{i}}}{\eta_{\hat{i}}}} \cdot \sqrt{\innerprod{\vartheta_{\hat{j}}}{\vartheta_{\hat{j}}}}
        \end{align}
    \end{subequations}
    Plugging \Cref{eq:hsep-PSD-UBterm} into \Cref{eq:hsep-PSD-UB}, we obtain the desired upper bound:
    \begin{align*}
        &\hsepPlus{d_1d'_1}{d_2d'_2}(M\totimes M')\\
        \leq~& \hsepPlus{d'_1}{d'_2}(M') \cdot \sum_{i,j,\hat{i},\hat{j}} \bra{i}_{\sfA}\!\bra{j}_{\sfB}M_{\sfA\sfB}\ket{\hat{i}}_{\sfA}\!\ket{\hat{j}}_{\sfB} \cdot \sqrt{\innerprod{\eta_i}{\eta_i}} \cdot \sqrt{\innerprod{\vartheta_j}{\vartheta_j}} \cdot \sqrt{\innerprod{\eta_{\hat{i}}}{\eta_{\hat{i}}}} \cdot \sqrt{\innerprod{\vartheta_{\hat{j}}}{\vartheta_{\hat{j}}}}\\
        =~& \hsepPlus{d'_1}{d'_2}(M') \cdot  \bra{\xi}_{\sfA} \bra{\varsigma}_{\sfB} M_{\sfA\sfB} \ket{\xi}_{\sfA} \ket{\varsigma}_{\sfB}\\
        \leq~& \hsepPlus{d'_1}{d'_2}(M') \cdot \hsepPlus{d_1}{d_2}(M).
    \end{align*}

    Here, $\ket{\xi}=\sum_{i=1}^{d_1} \sqrt{\innerprod{\eta_i}{\eta_i}} \ket{i}$ and $\ket{\varsigma}=\sum_{j=1}^{d_2} \sqrt{\innerprod{\vartheta_j}{\vartheta_j}} \ket{j}$ are non-negative states, where the normalization is ensured by \Cref{eq:hsep-PSD-normalization}. 
\end{proof}

However, this multiplicative property only holds for Hermitian matrices in \emph{product} form, and the proof is deferred to \Cref{subsec:omitted-parallel-repetition-StoqMAk}: 

\begin{restatable}[Stoquastic separable values of product Hermitian matrices are multiplicative]{lemma}{stoqSepValMultiplicativeProduct}
    \label{lemma:stoqSepVal-multiplicative-product}
    Let $M$ and $M'$ be entrywise non-negative product Hermitian matrices of dimensions $\prod_{i\in[k]} d_i$ and $\prod_{i\in[k]} d'_i$, respectively, and suppose that these matrices admit the product decompositions 
    \begin{equation}
        \label{eq:sep-vals-product-form}
        M = M_1 \totimes \cdots \totimes M_k \quad\text{and}\quad M' = M'_1 \totimes \cdots \totimes M'_k,
    \end{equation}
    where each factor $M_i$ and $M'_i$ is Hermitian and entrywise non-negative for $i\in[k]$. Then the stoquastic separable value is multiplicative under tensor products, namely
    \[ \hsepPlus{d_1 d'_1}{\cdots, d_k d'_k}\rbra*{M \totimes M'} = \hsepPlus{d_1}{\cdots,d_k }(M) \cdot \hsepPlus{d'_1}{\cdots,d'_k }(M'). \]
\end{restatable}

\subsection{Two \StoqMAk{}-complete problems from the class definition}
\label{subsec:StoqMAk-complete-probs}

We introduce two \StoqMAk{}-complete problems that arise directly from the class definition, and each admits an interpretation in terms of property testing over different objects. Our second complete problem ($\SepRCD_k$) is directly analogous to the \StoqMA{}-complete problem \textsc{Reversible Circuit Distinguishability} (\RCD{}) introduced in~\cite[Section 4]{Liu21}. 

\subsubsection{Separable Stoquastic Close Image to Uniform}
We start with an observation due to~\cite[Theorem 9]{Kobayashi03}: deciding whether $\Tr(\rho_0\rho_1)$ is at least $2/3$ or at most $1/3$ for \emph{single-qubit} states $\rho_0$ and $\rho_1$, whose purifications can be prepared efficiently, is \BQP{}-complete. This observation suggests that the acceptance probability of quantum computation can be interpreted as a closeness testing problem. Moreover, when neither $\rho_0$ nor $\rho_1$ admits such a measurement, we can instead implement the measurement induced by the SWAP test for mixed states~\cite{BCWdW01,KMY09}, which accepts with probability $(1+\Tr(\rho_0\rho_1))/2$. 

Since the SWAP test can be viewed as a stoquastic verifier, as implicitly observed in~\cite[Section 4]{Liu21}, we introduce our first \StoqMAk{}-complete problem $\SepStoqCIU_k$:
\footnote{This \StoqMA{}-complete variant, \textsc{Stoquastic Close Image to Uniform} (\StoqCIU{}), is obtained by removing the tensor-product promise; that is, in the promise, replacing the maximization over $\ket{\psi_1},\ket{\psi_2}\in\calS_+\rbra*{2^\ell}$ with the maximization over $\ket{\psi}\in\calS_+\rbra*{2^{2\ell}}$.}

\begin{definition}[$k$-\textsc{Separable Stoquastic Close Image to Uniform}, $\SepStoqCIU_k$]
    \label{def:SepStoqCIU}
    Let $R$ be a polynomial-size classical reversible circuit consisting only of Toffoli, CNOT, and X gates. Suppose that $R$ acts on an $m(n)$-qubit input state of the form 
    \[\ket{\psi_1}\otimes \cdots \otimes \ket{\psi_k} \otimes\ket{0}^{\otimes m_0}\otimes\ket{+}^{\otimes m_+},\] 
    where each $\ket{\psi_j}$ is an $\ell(n)$-qubit non-negative state for every $j\in[k]$, and $k\ell+m_0+m_+=m$. Assume that $R$ has $r(n)$ designated output qubits, where $1 \leq r(n) \leq m(n)$ and $r=r_0+r_+$. Let $\Phi_R(\cdot)$ denote the quantum channel obtained by applying $R$ and then tracing out all non-output qubits.\footnote{Notably, the underlying quantum channel belongs to a subclass of \emph{completely positive trace-preserving maps} introduced and studied in~\cite{JS22}, which preserve non-negative amplitudes in both input and output states.} The promise problem $\SepStoqCIU_k[\alpha(n),\beta(n)]$ is to decide whether one of the following holds:
    \begin{itemize}
        \item \emph{\textbf{Yes}:} $\max\limits_{\ket{\psi_1},\cdots,\ket{\psi_k}\in\calS_+(2^{\ell})} \F^2\rbra*{\Phi_R\rbra*{\bigotimes_{j\in[k]}\ketbra{\psi_j}{\psi_j}}, \ketbra{0}{0}^{\otimes r_0} \totimes \ketbra{+}{+}^{\otimes r_+}} \geq \alpha(n)$. 
        \item \emph{\textbf{No}:}  $\max\limits_{\ket{\psi_1},\cdots,\ket{\psi_k}\in\calS_+(2^{\ell})} \F^2\rbra*{\Phi_R\rbra*{\bigotimes_{j\in[k]}\ketbra{\psi_j}{\psi_j}}, \ketbra{0}{0}^{\otimes r_0} \totimes \ketbra{+}{+}^{\otimes r_+}} \leq \beta(n)$. 
    \end{itemize}
    When $k=2$, we simply write this problem as $\SepStoqCIU[\alpha(n),\beta(n)]$. 
\end{definition}

We then establish our first \StoqMAk{}-complete problem: 

\begin{lemma}[$\SepStoqCIU_k$ is \StoqMAk{}-complete]
    \label{lemma:SepStoqCIU-StoqMAk-complete}
    For all efficiently computable functions $\alpha(n)$ and $\beta(n)$, and every efficiently computable positive integer-valued function $k(n)\leq \poly(n)$, such that $0 \leq \beta(n) < \alpha(n) \leq 1$ and $\alpha(n)-\beta(n)\geq 1/\!\poly(n)$, the following holds: 
    \begin{enumerate}[label={\upshape(\arabic*)}]
        \item \label{thmitem:SepStoqCIU-in-StoqMAk} $\SepStoqCIU_k[\alpha,\beta]$ is in $\StoqMA\rbra*{k,\frac{1+\alpha}{2},\frac{1+\beta}{2}}$. 
        \item \label{thmitem:SepStoqCIU-StoqMAkHard} $\SepStoqCIU_k[\alpha,\beta]$ is $\StoqMA\rbra*{k,\alpha,\beta}$-hard. 
    \end{enumerate}
\end{lemma}

\begin{proof}
    We prove only the case where $k=2$ for simplicity, and the same argument directly extends to general $k$. 
    To establish \Cref{thmitem:SepStoqCIU-in-StoqMAk}, we simply consider the SWAP test~\cite{BCWdW01} with two quantum states 
    \[\rho_0\coloneqq \Phi_R(\ketbra{\psi_1}{\psi_1}\otimes \ketbra{\psi_2}{\psi_2}) 
    \quad\text{and}\quad
    \rho_1\coloneqq \ketbra{0}{0}^{\otimes r_0} \otimes \ketbra{+}{+}^{\otimes r_+},\] 
    which directly constitutes a \StoqMAtwo{} verifier given the witness state $\ket{\psi_1}\otimes\ket{\psi_2}$. Noting that $\rho_1$ is a pure state, it follows that $\F^2(\rho_0,\rho_1) = \Tr(\rho_0\rho_1)$. Consequently, the acceptance probability of the SWAP test for mixed states~\cite[Proposition 9]{KMY09} implies the desired threshold bound. 

    To prove \Cref{thmitem:SepStoqCIU-StoqMAkHard}, we begin with a stoquastic verifier $V_\calI$ corresponding to the promise problem $\calI\in\StoqMA_{\ell}(2)$. For any $x\in\calI$, we fix the stoquastic verifier to be $V_x \coloneqq V_\calI(x)$. 
    For any non-negative product state witness $\ket{\psi_1}\totimes \ket{\psi_2}$, the acceptance probability of $V_x$ is 
    \begin{align*}
        \Pr[V_x \text{ accepts} \ket{\psi_1}\totimes\ket{\psi_2}] 
        &= \norm*{\ketbra{+}{+}_{\sfO} V_x \rbra*{\ket{\psi_1}\ket{\psi_2}\ket{\bar{+}}\ket{\bar{0}}}}_2^2\\
        &= \Tr\rbra*{ \ketbra{+}{+}_{\sfO} \Phi_R(\ketbra{\psi_1}{\psi_1}\totimes\ketbra{\psi_2}{\psi_2}) }\\
        &= \F^2(\ketbra{+}{+},\Phi_R(\ketbra{\psi_1}{\psi_1}\totimes\ketbra{\psi_2}{\psi_2})).
    \end{align*}
    Here, the third line follows from the identity in \Cref{eq:one-pure-state}. 
    Therefore, we complete the proof by taking $\Phi_R$ to be the  channel defined by
    \[\forall \rho\in \calD_+\rbra[\big]{2^{2\ell}}, \quad \Phi_R(\rho) \coloneqq \Tr_{\overline{\sfO}} \rbra[\big]{ R \rbra*{\rho \otimes \ketbra{\bar{0}}{\bar{0}} \otimes \ketbra{\bar{+}}{\bar{+}}} R^\dagger },\] 
    where $R=V_x$ is the required classical reversible circuit.
\end{proof}

\subsubsection{Separable Reversible Circuit Distinguishability}
\label{subsubsec:SepRCD}

Next, we present our second \StoqMAk{}-complete problem $\SepRCD_k$, which is a direct generalization of~\cite[Section 4]{Liu21}: 

\begin{definition}[$k$-\textsc{Separable Reversible Circuit Distinguishability}, $\SepRCD_k$] 
    \label{def:SepRCD}
    Let $R_0$ and $R_1$ be polynomial-size classical reversible circuits consisting only of Toffoli, CNOT, and X gates. Suppose that both $R_0$ and $R_1$ act on an $m(n)$-qubit input state of the form $\ket{\psi_1}\otimes\cdots\otimes\ket{\psi_k}\otimes \ket{0}^{\otimes m_0} \otimes \ket{+}^{\otimes r}$, where each $\ket{\psi_j}$ is an $\ell(n)$-qubit non-negative state for every $j\in[k]$ and $k\ell+m_0+r=m$. With a slight abuse of notation, let $\ket{R_b}$ denote the resulting state, defined by $\ket{R_b} \coloneqq R_b \rbra*{\ket{\psi_1}\otimes\cdots\otimes\ket{\psi_k}\otimes \ket{\bar{0}} \otimes \ket{\bar{+}}}$, for each $b\in\binset$. The promise problem $\SepRCD_k[\alpha(n),\beta(n)]$ is to decide whether one of the following holds:
    \begin{itemize}
        \item \emph{\textbf{Yes}:} There exists a non-negative state $\ket{\psi_1}\otimes\cdots\otimes\ket{\psi_k}$ such that $\innerprod{R_0}{R_1} \geq \alpha(n)$. 
        \item \emph{\textbf{No}:} For all non-negative states $\ket{\psi_1}\otimes\cdots\otimes\ket{\psi_k}$, we have $\innerprod{R_0}{R_1} \leq \beta(n)$. 
    \end{itemize}
    When $k=2$, we simply write this problem as $\SepRCD[\alpha(n),\beta(n)]$, and we denote the version with detailed input-size parameters by $\SepRCD_{\alpha,\beta}(\ell,m_0,r)$. 
\end{definition}

It is worth noting that, when $k=1$, the promise problem $\SepRCD_1$ coincides with the \StoqMA{}-complete problem \RCD{} introduced in~\cite[Section 4.1]{Liu21}. 

\begin{restatable}[$\SepRCD_k$ is \StoqMAk{}-complete]{lemma}{SepRCDisComplteForStoqMA}
    \label{lemma:SepRCD-StoqMAk-complete}
    For all efficiently computable functions $\alpha(n)$ and $\beta(n)$, and every efficiently computable positive integer-valued function $k(n)\leq \poly(n)$, such that  $0 \leq \beta(n) < \alpha(n) \leq 1$ and $\alpha(n)-\beta(n)\geq 1/\!\poly(n)$, we have
    \[\SepRCD_k[\alpha,\beta] \text{ is } \StoqMA\rbra*{k,\frac{1+\alpha}{2},\frac{1+\beta}{2}}\text{-complete}. \]
    In particular, for any stoquastic verification circuit $V_x$, the hard instance is given by the reversible circuit pair $R_0 = V_x^\dagger X_{\sfO} V_x \coloneqq \Gamma_x$ and $R_1 = I$.
\end{restatable}

The proof follows the same strategy as the proof of~\cite[Theorem 4.1]{Liu21}, and the details are given in \Cref{subsec:omitted-parallel-repetition-StoqMAk}.

\subsection{\StoqMAk{} is closed under direct product}
\label{subsec:StoqMAk-closed-direct-product}

\begin{theorem}[\StoqMAk{} is closed under direct product]
    \label{thm:StoqMAk-closed-under-direct-product-w-loss}
    Let $\big\{\calI^{(j)}\big\}_{j\in[r]}$ be a collection of promise problems such that $\calI^{(j)} \in \StoqMA_\ell\big(k,\frac{1}{2}+\frac{a_j}{2}, \frac{1}{2}+\frac{b_j}{2}\big)$, and let $V^{(j)}$ denote the corresponding stoquastic verifier. 
    Then there exists an explicit stoquastic verifier $V'$ for $\calI' \coloneqq \calI^{(1)} \times \cdots \times \calI^{(r)}$, constructed from $V^{(1)},\cdots,V^{(r)}$, such that 
    \[\calI' \in \StoqMA_{r\ell}\bigg(k, \frac{1}{2}+\frac{1}{2^{r+1}} \cdot \prod_{j\in[r]} (1+a_j), \frac{1}{2}+\frac{1}{2^{r+1}} \cdot \prod_{j\in[r]} (1+b_j)\bigg).\]
\end{theorem}

\begin{proof}
    For each promise problem $\calI^{(j)}\in\StoqMA\rbra*{k, \frac{1+a_j}{2}, \frac{1+b_j}{2}}$, where $j\in[r]$, let $V^{(j)}_{x_j}$ be a corresponding stoquastic verifier with $x_j \in \calI^{(j)}$. Following the proof that $\SepStoqCIU_k \in \StoqMAk$ (\Cref{lemma:SepStoqCIU-StoqMAk-complete}\ref{thmitem:SepStoqCIU-in-StoqMAk}), we construct a new stoquastic verifier $V'_{x'}$, where $x'\coloneqq (x_1,\cdots,x_r)$, for the direct product of the \StoqMAk{} promise problems $\calI^{(1)},\cdots,\calI^{(r)}$, as illustrated in \Cref{fig:weak-conjunction-StoqMAtwo-verifiers} for the case $k=2$ and $r=3$. More precisely, for each $j\in[r]$, we execute the verification circuit $V^{(j)}_{x_j}$ on the registers $\rbra[\big]{\sfW^{(j)}, \sfA_0^{(j)}, \sfA_+^{(j)}}$, where $\sfW^{(j)}$, $\sfA_0^{(j)}$, and $\sfA_+^{(j)}$ contain $k\ell$ qubits, $m_0^{(j)}$ qubits, and $m_+^{(j)}$ qubits, respectively, in parallel and without performing the final Hadamard-basis measurement, and then apply a SWAP test between the designated output qubits of these verification circuits and the state $\ket{+}^{\otimes r}$, taking the output qubit of the SWAP test as the output qubit of $V'_{x'}$. 

\begin{figure}[ht!]
\centering
\begin{quantikz}[row sep={0.8em}, column sep=2em]
    \lstick{$\ket{+}$} & & \qw & \qw & \ctrl{1} & \meterD{\ket{+}} \\
    \lstick{$\ket{+}^{\otimes 3}$} & & \qwbundle{3} & \qw & \gate[4]{\hspace{1.5em} \mathrm{SWAP} \hspace{1.5em}} & \qwbundle{3} \\
    \lstick{$\ket{\psi_1^{(1)}}\totimes\ket{\psi_2^{(1)}}_{\sfW^{(1)}}$} & \qwbundle{2\ell} & \gate[3]{\quad V_{x_1}^{(1)} \quad} & \permute{1,4,5,2,6,7,3,8,9} & \qw & \qw \\
    \lstick{$\ket{\bar{0}}_{\sfA_0^{(1)}}$} & \qwbundle{m_0^{(1)}} & & & & \qw \\
    \lstick{$\ket{\bar{+}}_{\sfA_+^{(1)}}$} & \qwbundle{m_+^{(1)}} & & & & \qw \\
    \lstick{$\ket{\psi_1^{(2)}}\totimes\ket{\psi_2^{(2)}}_{\sfW^{(2)}}$} & \qwbundle{2\ell} & \gate[3]{\quad V_{x_2}^{(2)} \quad} & & \qwbundle{m_0^{(1)}+2\ell-1} & \qw \\
    \lstick{$\ket{\bar{0}}_{\sfA_0^{(2)}}$} & \qwbundle{m_0^{(2)}} & & & \qwbundle{m_+^{(1)}}  & \qw \\
    \lstick{$\ket{\bar{+}}_{\sfA_+^{(2)}}$} & \qwbundle{m_+^{(2)}} & & & \qwbundle{m_0^{(2)}+2\ell-1}  & \qw \\
    \lstick{$\ket{\psi_1^{(3)}}\totimes\ket{\psi_2^{(3)}}_{\sfW^{(3)}}$} & \qwbundle{2\ell} & \gate[3]{\quad V_{x_3}^{(3)}\quad} & & \qwbundle{m_+^{(2)}} & \qw \\
    \lstick{$\ket{\bar{0}}_{\sfA_0^{(3)}}$} & \qwbundle{m_0^{(3)}} & & & \qwbundle{m_0^{(3)}+2\ell-1} & \qw \\
    \lstick{$\ket{\bar{+}}_{\sfA_+^{(3)}}$} & \qwbundle{m_+^{(3)}} & & & \qwbundle{m_+^{(3)}} & \qw
\end{quantikz}
\caption{Weak conjunction of stoquastic verifiers $V_{x_1}^{(1)}$, $V_{x_2}^{(2)}$, and $V_{x_3}^{(3)}$. }
\label{fig:weak-conjunction-StoqMAtwo-verifiers}
\end{figure}
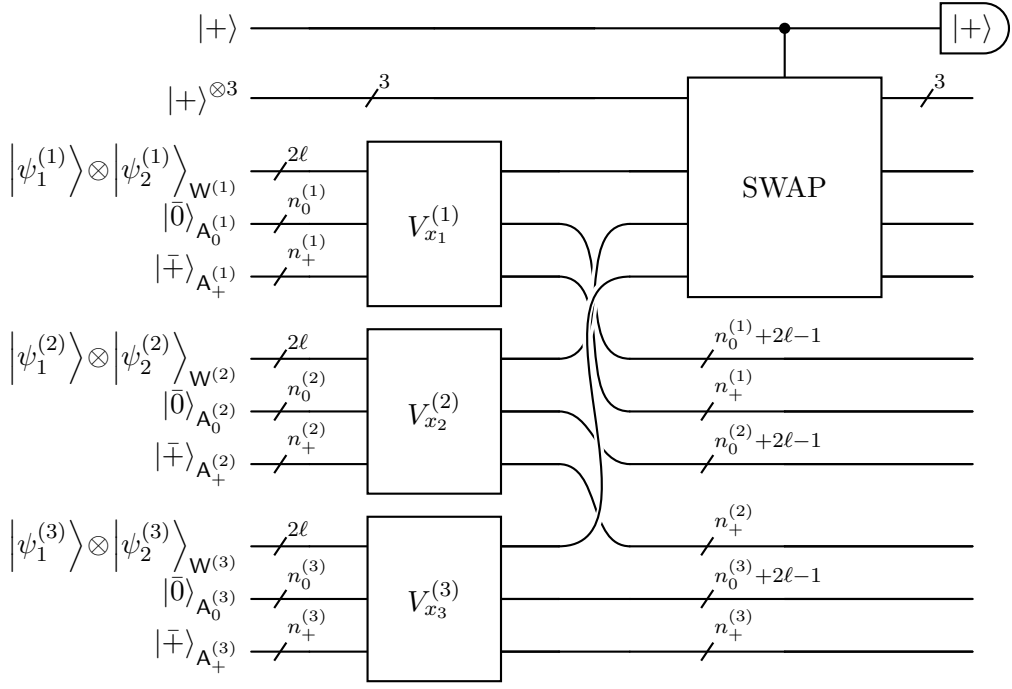

    We now analyze the maximum acceptance probability of $V'_{x'}$:
    \begin{itemize}
        \item For \emph{yes} instances, the witness state has an additional tensor-product structure: 
        \begin{equation}
            \label{eq:witness-product-form}
            \forall i \in [k], \quad \ket{\psi'_i} = \ket{\psi_i^{(1)}} \otimes \cdots \otimes \ket{\psi_i^{(r)}}.
        \end{equation}
        Consequently, by the same argument as in the proof of \Cref{lemma:SepStoqCIU-StoqMAk-complete}\ref{thmitem:SepStoqCIU-in-StoqMAk}, we obtain
        \begin{align*}
            \Pr\sbra*{V'_{x'} \text{ accepts } \ket{\psi'_1}\totimes\cdots\totimes \ket{\psi'_k}} &= \frac{1}{2} + \frac{1}{2} \prod_{j\in[r]} \Pr\sbra*{V^{(j)}_{x_j} \text{ accepts } \ket{\psi_1^{(j)}}\totimes\cdots\totimes\ket{\psi_k^{(j)}}}\\ 
            &\geq \frac{1}{2} + \frac{1}{2} \cdot \prod_{j\in[r]} \rbra*{\frac{1}{2} + \frac{a_j}{2}}.
        \end{align*}
        
        \item For \emph{no} instances, the quantum states $\ket{\psi'_1},\cdots,\ket{\psi'_k}$ may be highly entangled. We therefore complete the soundness analysis using the multiplicative property of stoquastic separable values (\Cref{lemma:stoqSepVal-multiplicative-PSD}). To this end, for each $j\in[r]$, we define the positive semi-definite matrix $M^{(j)}_{x_j}$ corresponding to the verification circuit $V_{x_j}^{(j)}$:
        \[M^{(j)}_{x_j} \coloneqq \bra{\bar{0}}_{\sfA_0^{(j)}}\bra{\bar{+}}_{\sfA_+^{(j)}} {V_{x_j}^{(j)}}^\dagger \ketbra{+}{+}_{\sfO^{(j)}} V_{x_j}^{(j)} \ket{\bar{0}}_{\sfA_0^{(j)}}\ket{\bar{+}}_{\sfA_+^{(j)}}.\]

        Let $\omega(V_x)$ denote the maximum acceptance probability of a \StoqMAk{} verifier $V_{x}$, namely, 
        \[\omega(V_x) \coloneqq \max_{\ket{\psi_1},\cdots,\ket{\psi_k}} \Pr\sbra*{V_x\text{ accepts }\ket{\psi_1}\otimes\cdots\otimes\ket{\psi_k}}.\] 
        Then, the maximum acceptance probability of $V'_{x'}$ can be expressed as:
        \begin{align*}
            \omega\rbra*{V'_{x'}}
            &= \frac{1}{2} + \frac{1}{2} \hsepPlus{2^{r\ell}}{\cdots,2^{r\ell}}\rbra*{ M^{(1)}_{x_1} \otimes\cdots\otimes M^{(r)}_{x_r} }\\
            &\leq \frac{1}{2} + \frac{1}{2} \prod_{j\in[r]} \hsepPlus{2^\ell}{\cdots,2^\ell} \rbra*{M^{(j)}_{x_j}}\\
            &= \frac{1}{2} + \frac{1}{2} \prod_{j\in[r]} \omega\rbra[\big]{V^{(j)}_{x_j}}\\
            &\leq \frac{1}{2} + \frac{1}{2} \cdot \prod_{j\in[r]} \rbra*{\frac{1}{2} + \frac{b_j}{2}}.
        \end{align*}
    \end{itemize}
    This completes the proof, where the second line follows from \Cref{lemma:stoqSepVal-multiplicative-PSD}. 
\end{proof}

%%%%%%%%%%%%%%%%%%%%%%%%%%%%%%%%%%%%%%%%%%%%%%%%%%%%%%%%%%%
%%%%%%%%%%%%%%%%%%%%%%%%%%%%%%%%%%%%%%%%%%%%%%%%%%%%%%%%%%%
%%%%%%%%%%%%%%%%%%%%%%%%%%%%%%%%%%%%%%%%%%%%%%%%%%%%%%%%%%%

\section{Robustness of \StoqMAtwo{}}

In this section, we study the robustness of the definition of \StoqMAtwo{} by establishing two basic properties analogous to those of \QMAtwo{}, particularly \emph{prover compression}~\cite{HM13} and \emph{witness symmetrization}~\cite{ABDFS09}. Notably, while \StoqMAtwo{} admits only parallel repetition (\Cref{thm:StoqMAk-w-parallel-repetition-loss}), rather than full error reduction as in \QMAtwo{}, our results imply that \StoqMAtwo{} is \emph{well-defined} with respect to these properties, \emph{up to a quadratic gap loss}.

\vspace{1em}
When there are more than two provers, we establish \emph{prover compression} for \StoqMAk{} by stoquastizing the well-known product test. This technique allows one to ``compress'' any $\StoqMAk{}$ protocol to a $\StoqMAtwo{}$ protocol while preserving the \emph{constant} gap:
\begin{theorem}[Prover compression for \StoqMAk{} up to a quadratic gap loss, informal version of \Cref{thm:stoq-k-to-2}]
\label{thm:StoqMAk-in-StoqMAtwo}
Let $k(n)$ be a polynomially bounded integer-valued function satisfying that $k(n) > 2$. Let $\epsilon(n)$ and $\Delta(n)$ be functions such that $0 \leq \epsilon(n) < 1/2$ and $\Delta(n) > 0$. Then the following inclusion holds:  
    \[\StoqMA_{\ell}\ssbra*{k,1-\epsilon,\Delta} \subseteq \StoqMA_{k\ell}\ssbra*{2,1-O(\epsilon \Delta),\Omega(\Delta^2)}\]
\end{theorem}

When the completeness error $\epsilon(n)$ is \emph{negligible}, combining prover compression with parallel repetition for \StoqMAtwo{} (\Cref{thm:StoqMAk-w-parallel-repetition-loss}, more specifically \Cref{corr:error-reduction-StoqMAk-negl}) yields prover compression \emph{without} any gap loss, and thus achieves a promise gap exponentially close to $1/2$:\footnote{One subtlety behind \Cref{thm:StoqMAk-in-StoqMAtwo} concerns the dyadic truncation error, as explained in \Cref{footnote:dyadic-truncation-error}, which may prevent perfect completeness being preserved in the resulting \StoqMAtwo{} verification circuit. Nevertheless, together with parallel repetition for \StoqMAtwo{} (\Cref{thm:StoqMAk-w-parallel-repetition-loss}), we can simply choose $\lambda=1/2$, carry out the remaining construction, and amplify the promise gap later.}
\begin{corollary}[Prover compression for $\StoqMAk_{1-\negl}$]
\label{corr:StoqMAk-eq-StoqMAtwo-small-yesError}
Let $k(n)$ be a polynomially bounded integer-valued function satisfying $k(n) > 2$. 
Let $\epsilon(n)$ be a negligible non-negative function, and let $\Delta(n)$ be a positive function that is at least $1/\poly(n)$. Then the following bounds hold:
\begin{enumerate}[label={\upshape(\arabic*)}]
    \item There exist polynomially bounded integer-valued functions $r(n)$ and $q(n)$ such that 
    \begin{equation}
        \label{eq:prover-compression-negl-general}
        \StoqMA_\ell \ssbra*{k, 1-\epsilon, \Delta} \subseteq \StoqMA_{rk\ell} \ssbra*{2, 1-\epsilon', \frac{1}{2}-2^{-q}}.
    \end{equation}
    Here, $\epsilon'(n)$ is also a negligible non-negative function. 
    In particular, by setting $r_\star = \ceil*{1/\Delta}$, letting $\epsilon_\star \leq \epsilon/(2\Delta)$ and $\ell_\star(n) = O(k\ell/\Delta)$, it holds that
    \begin{equation}
        \label{eq:prover-compression-negl-linearGap}
        \StoqMA_\ell\ssbra{k,1-\epsilon,\Delta} \subseteq \StoqMA_{\ell_\star}\ssbra*{2,1-\epsilon_\star,\Omega\rbra{\Delta}}.
    \end{equation}
    \item \label{thmitem:prover-compression-negl-lengthPreserved}Without further parallel repetition, it holds that
    \begin{equation}
        \label{eq:prover-compression-negl-lengthPreserved}
        \StoqMA_\ell\ssbra{k,1-\epsilon,\Delta} \subseteq \StoqMA_{k\ell}\ssbra*{2,1-\epsilon/2,\Omega\rbra{\Delta^2}}.
    \end{equation}
\end{enumerate}
\end{corollary}

\begin{proof}
    For simplicity, we only establish \Cref{eq:prover-compression-negl-linearGap,eq:prover-compression-negl-lengthPreserved}, while \Cref{eq:prover-compression-negl-general} follows from similar reasoning. To this end, apply prover compression (\Cref{thm:stoq-k-to-2}) with the exactly dyadic choice $\lambda=1/2$, which increases the proof length from $\ell$ to $k\ell$. Writing $c=1-\epsilon$, $s=c-\Delta$, and $\gamma=c_{\mathrm{prod}}/4=1/12$, both cases of \Cref{eq:prove-compression-bound} are minimized at $x=1-s=\epsilon+\Delta$, giving the compressed verifier completeness and soundness thresholds 
    \[ \hat{c} = 1 - \frac{\epsilon}{2} \quad\text{and}\quad \hat{s} = 1 - \frac{\gamma}{2} (\epsilon+\Delta)^2 \geq 1 - 2\gamma\Delta^2,\] 
    which yields the promise gap $\hat{\Delta}=\hat{c}-\hat{s}= \rbra*{\gamma(\epsilon+\Delta)^2-\epsilon}/2$.
    Since $\epsilon$ is negligible and $\Delta\geq 1/\poly(n)$, we obtain $\hat{\Delta} \geq \gamma\Delta^2/4$ for all sufficiently large $n$.
    Without parallel repetition, using $\epsilon=o(\Delta^2)$, we have $\hat{\Delta} = \rbra[\big]{ \frac{\gamma}{2}+o(1) } \Delta^2 = \Theta(\Delta^2)$, which proves \Cref{eq:prover-compression-negl-lengthPreserved}.

    Now applying parallel repetition for \StoqMAk{} (\Cref{thm:StoqMAk-w-parallel-repetition-loss}) with $r_{\star} = \ceil*{1/\Delta}$, which increases the proof length further to $\ell_\star = r_\star k\ell$, we obtain the resulting gap
    \[ \Delta_\star = \frac{\hat{c}^{r_\star} - \hat{s}^{r_\star}}{2}
    \geq \frac{\hat{\Delta}}{2} \cdot r_\star \cdot  \hat{s}^{r_{\star} -1}
    \geq \frac{\gamma\Delta^2}{8} \cdot \frac{1}{\Delta} \cdot \rbra{1-2\gamma\Delta} \geq \frac{11\Delta}{1152}.\]
    Here, the first inequality follows from $x^k-y^k=(x-y)\sum_{j=0}^{k-1}x^{k-1-j}y^j$ for all integers $k\geq 2$, and the second inequality uses Bernoulli's inequality, $(1-a)^r \geq 1-ra$ for all $a\in[0,1]$. 
    Finally, the resulting completeness error $\epsilon_\star = \rbra*{ 1 - \hat{c}^{r_\star} }/2 \leq r_\star \epsilon /4 \leq \epsilon / (2\Delta)$. 
\end{proof}

\vspace{1em}
In addition to prover compression, we also prove \emph{witness symmetrization} for \StoqMAk{}, which can be viewed as a closure property under \emph{symmetrization} of the product witness states. More specifically, we define a symmetrized version of \StoqMAk{}, denoted by \SymStoqMAk{}, in which all $k$ factors (``proofs'') of the witness state are promised to be \emph{identical}, as stated in \Cref{def:SymStoqMAk}. This symmetric variant is analogous to the class $\textsf{SymQMA}(k)$ introduced in~\cite[Section 4.4]{ABDFS09}.

\begin{definition}[\SymStoqMAk{}]
\label{def:SymStoqMAk}
A promise problem $\calI = (\calI_{\yes}, \calI_{\no}) \in$ $\SymStoqMA_\ell(k,c,s)$ if there exists a stoquastic verifier $V_\calI$ such that for every input $x \in \calI$, the verifier is specified by an $n$-qubit stoquastic verification circuit $V_{\calI}(x)$, as in \Cref{def:StoqMAk}, such that the following hold for efficiently computable functions $c(n)$, $s(n)$, and $\ell(n)$: 
\begin{description}
    \item[Completeness. ] If $x\in \calI_{\yes}$, then there exists an $\ell(n)$-qubit non-negative state $\ket{\psi}$ such that 
    \[\Pr\sbra*{V_{\calI}(x) \text{ accepts } \ket{\psi}^{\otimes k}} \geq c(n).\] 
    \item[Soundness. ] If $x \in \calI_{\no}$, then for every $\ell(n)$-qubit non-negative state $\ket{\psi}$, we have
    \[\Pr\sbra*{V_{\calI}(x) \text{ accepts } \ket{\psi}^{\otimes k}} \leq s(n).\]
\end{description}
Furthermore, $c(n)$ and $s(n)$ satisfy $1/2 \leq s(n) < c(n) \leq 1$ and $c(n)-s(n) \geq 1/\poly(n)$. For convenience, we write the stoquastic verification circuit $V_{\calI}(x)$ as $V_x$. 
\end{definition}

To establish the equivalence between \StoqMAk{} and \SymStoqMAk{} up to a quadratic gap loss, we first show that \SymStoqMAk{} can simulate \StoqMAk{} using a new \emph{length-efficient} symmetrization, and then prove that $\StoqMAk{}$ can simulate \SymStoqMAk{} by stoquastizing the projector onto the symmetric subspace: 

\begin{theorem}[Witness symmetrization for \StoqMAk{} up to quadratic gap loss] 
\label{thm:StoqMAk-eq-SymStoqMAk}
Let $k(n)$ be a polynomially bounded integer-valued function satisfying that $k(n) \geq 2$. Let $\epsilon(n)$ and $\Delta(n)$ be functions such that $0 \leq \epsilon(n) < 1/2$ and $\Delta(n) > 0$. Then the following inclusions hold: 
\begin{enumerate}[label={\upshape(\arabic*)}]
    \item \label{thmitem:StoqMA-in-SymStoqMA-ssbra} $\displaystyle\StoqMA_\ell\ssbra*{k,1-\epsilon,\Delta} \subseteq \SymStoqMA_{O(\ell\log k+\log^2 k)}\ssbra*{k,1-\epsilon-\Delta/2,\Delta/2}.$
    \item \label{thmitem:SymStoqMA-in-StoqMA-ssbra} $\displaystyle\SymStoqMA_\ell\ssbra*{k,1-\epsilon,\Delta} \subseteq \StoqMA_\ell\ssbra*{k, 1-O(\epsilon\Delta),\Omega(\Delta^2)}.$
\end{enumerate}
\end{theorem}

\vspace{1em}
In the remainder of this section, we establish \Cref{thm:StoqMAk-in-StoqMAtwo} in \Cref{subsec:StoqMAk-vs-StoqMAtwo}: we introduce the stoquastic product test in \Cref{subsubsec:stoq-product-test}, and then use it to prove prover compression for \StoqMAk{} in \Cref{subsubsec:prover-compression}. Next, we prove \Cref{thm:StoqMAk-eq-SymStoqMAk} in \Cref{subsec:SymStoqMAk-vs-StoqMAk}: we first present our length-efficient symmetrization of \StoqMAk{} by \SymStoqMAk{} in \Cref{subsubsec:StoqMAk-in-SymStoqMAk}, then introduce the stoquastic equality test in \Cref{subsubsec:stoq-equality-test}, finally use it to show that \StoqMAk{} can simulate \SymStoqMAk{} in \Cref{subsubsec:SymStoqMAk-in-StoqMAk}. 

\subsection{\StoqMAk{} vs.~\StoqMAtwo{}}
\label{subsec:StoqMAk-vs-StoqMAtwo}

We begin by establishing a stoquastic version of the product test (\Cref{thm:prodtest-constant}), and then use this procedure as a key ingredient in proving prover compression for \StoqMAk{}.

\subsubsection{The stoquastic product test}
\label{subsubsec:stoq-product-test}

\begin{lemma}[Stoquastic implementation of the product test]
\label{lem:stoq-prod}
For every pair of pure states $\ket{\rho},\ket{\sigma}\in (\mathbb{C}^{2^\ell})^{\otimes k}$, there is a stoquastic verification circuit that takes the state $\ket{\rho}\otimes\ket{\sigma}$ as input and accepts with probability
\begin{equation}\label{eq:stoq-prod-test}
    \frac12+\frac12\,P_{\mathrm{prod}}(\rho,\sigma), 
    \quad\text{where } P_{\mathrm{prod}}(\rho,\sigma)
     \coloneqq \frac{1}{2^k}\sum_{S\subseteq [k]} \Tr(\rho_S\sigma_S).
\end{equation}
In particular, if $\ket{\rho}=\ket{\sigma}$ is a product state across the $k$ tensor factors, then the circuit accepts with probability $1$.
\end{lemma}

\begin{proof}
Let $\rho,\sigma$ be states over $\sfA=\sfA_1\cdots \sfA_k$ and $\sfB=\sfB_1\cdots \sfB_k$, respectively, where each $\sfA_i$ and $\sfB_i$ is an $\ell$-qubit register.  For every subset $S\subseteq [k]$, let
\[
    F_S \coloneqq \prod_{i\in S} \mathrm{SWAP}_{\sfA_i\sfB_i},
\]
where the factors commute because they act on disjoint register pairs. By the Harrow--Montanaro product test identity, as stated in the proof of~\cite[Lemma 2]{HM13},
\begin{equation}\label{eq:product-HM-identity}
    \frac{1}{2^k}\sum_{S\subseteq [k]} \Tr\bigl((\rho\otimes\sigma)F_S\bigr)
    =\frac{1}{2^k}\sum_{S\subseteq [k]} \Tr(\rho_S\sigma_S)
    =P_{\mathrm{prod}}(\rho,\sigma),    
\end{equation}
$P_{\mathrm{prod}}(\rho,\sigma)$ is the probability that the standard product test accepts $\rho$ and $\sigma$. 

\LinesNotNumbered
\begin{algorithm}[ht!]
    \UseProtocolCounter
    \caption{Stoquastic implementation of the product test.}
    \label{protocol:stoq-product-test}
    \SetEndCharOfAlgoLine{.}
    \setlength{\parskip}{5pt}
    \SetKwInOut{Registers}{Registers}
    \SetKwInOut{Parameter}{Parameters}

    \Registers{
    \begin{minipage}[t]{0.7\linewidth}
        \hangindent=3.2em
        \hangafter=1
        $\sfJ$: the $k$-qubit branch-index register, initialized to $\ket{+}^{\otimes k}$\; 
        $\sfA,\sfB$: two $k\ell$-qubit witness registers\;
        $\sfZ$: the ancillary register, initialized to $\ket{\bar{0}}$.
    \end{minipage}
    }

    \textbf{1.} For every $s\in\binset^k$, let $S(s) \coloneqq \set{i\in[k]}{s_i=1}$, where $s_i$ denotes the $i$-th bit of $s$. Define the reversible permutation
    $F_s \coloneqq \prod_{i\in S(s)} \mathrm{SWAP}_{\sfA_i,\sfB_i}$\;
    \textbf{2.} For every $i\in[k]$ and every qubit position $t\in[\ell]$, apply a controlled-SWAP between the $t$-th qubit of $\sfA_i$ and the $t$-th qubit of $\sfB_i$, controlled on the branch bit $\sfJ_i$.
    Equivalently, one can construct a reversible circuit $\Gamma$ on $(\sfJ,\sfA,\sfB,\sfZ)$ satisfying\footnote{Controlled-SWAP gates can be implemented using only CNOT and Toffoli gates; see \Cref{footnote:Fredkin} for details.}
    \[ \Gamma\ket{s}_{\sfJ}\ket{x}_{\sfA}{ \ket{y}_{\sfB} }\ket{\bar{0}}_{\sfZ}
    = \ket{s}_{\sfJ} F_s\rbra*{ \ket{x}_{\sfA} \ket{y}_{\sfB} } \ket{\bar{0}}_{\sfZ}. \]
    \textbf{3.} Apply the branch-overlap test (\Cref{lemma:branch-overlap-test}) to the reversible circuit pair $(\Gamma,I)$ on the initialized input state $\ket{+}_{\sfJ}^{\otimes k}\ket{\rho}_{\sfA}\ket{\sigma}_{\sfB}\ket{\bar{0}}_{\sfZ}$. 
\end{algorithm}

Next, we construct a stoquastic verification circuit that accepts with probability $\frac{1}{2}+\frac{1}{2} P_{\mathrm{prod}}(\rho,\sigma)$, as presented in \Cref{protocol:stoq-product-test}. By \Cref{lemma:branch-overlap-test}, for each branch $S$, the corresponding acceptance probability is given by 
\[
    \frac12+\frac12\,\bra{\rho}\bra{\sigma}F_S\ket{\rho}\ket{\sigma}
    =\frac12+\frac12\,\Tr\bigl((\rho\otimes\sigma)F_S\bigr).
\]
Consequently, averaging uniformly over $S$ yields \Cref{eq:stoq-prod-test}.

If $\ket{\rho}=\bigotimes_{i=1}^k \ket{\rho_i}$ is product and $\ket{\sigma}=\ket{\rho}$, then every
reduced density matrix $\rho_S$ is pure, so $\Tr(\rho_S^2)=1$ for all $S$.  Hence
$P_{\mathrm{prod}}(\rho,\rho)=1$, and the circuit accepts with probability $1$.
\end{proof}

\begin{theorem}[Stoquastic product test]
\label{thm:prodtest-constant}
There exists an absolute constant $c_{\mathrm{prod}}>0$ such that for every pure $k$-partite state
$\ket{\rho}$, if
\[
    1-\eta(\rho)
     \coloneqq \max\cbra*{ \abs*{ \innerprod{\rho}{\rho_1\otimes\cdots\otimes\rho_k} }^2:
    \ket{\rho_i}\in \bbC^{2^\ell} }.
\]
then
\[ P_{\mathrm{prod}}(\rho,\rho) \leq 1-c_{\mathrm{prod}}\,\eta(\rho)
    \quad\text{and}\quad
    P_{\mathrm{prod}}(\rho,\sigma) \leq 1-\frac{c_{\mathrm{prod}}}{2}\,\eta(\rho).
\]
Consequently, the stoquastic product test from \Cref{lem:stoq-prod} accepts with probability at most
\[
    1-\frac{c_{\mathrm{prod}} \eta(\rho)}{4}.
\]
In addition, one may take $c_{\mathrm{prod}}=1/3$.
\end{theorem}

\begin{proof}
Soleimanifar--Wright's analysis of the product test
\cite[Theorem~8]{SW22} gives
\[
    P_{\mathrm{prod}}(\rho,\rho)
    \le
    \begin{cases}
        1-\eta+\eta^2, & \eta\le \frac12,\\
        1-\frac23\eta+\frac13\eta^2, & \eta\ge \frac12,
    \end{cases}
\]
which implies the uniform bound
\[
    P_{\mathrm{prod}}(\rho,\rho)\le 1-\frac13\,\eta(\rho).
\]
Thus any fixed $c_{\mathrm{prod}}\in (0,1/3]$ is admissible, and in particular one may take
$c_{\mathrm{prod}}=1/3$.

By the Cauchy--Schwarz inequality and the AM-GM inequality, it follows that
\[
    \Tr(\rho_S\sigma_S)
    \le \sqrt{\Tr(\rho_S^2)}\,\sqrt{\Tr(\sigma_S^2)}
    \le \frac{\Tr(\rho_S^2)+\Tr(\sigma_S^2)}{2}
\]
for every $S\subseteq [k]$.  Consequently, averaging over $S$ gives
\[
    P_{\mathrm{prod}}(\rho,\sigma)
    \le \frac12\Bigl(P_{\mathrm{prod}}(\rho,\rho)+P_{\mathrm{prod}}(\sigma,\sigma)\Bigr)
    \le 1-\frac{c_{\mathrm{prod}}}{2}\,\eta(\rho),
\]
using the trivial bound $P_{\mathrm{prod}}(\sigma,\sigma)\le 1$.  The final claim follows from the
acceptance formula in \Cref{lem:stoq-prod}.
\end{proof}

\subsubsection{Prover compression for \StoqMAk{}}
\label{subsubsec:prover-compression}
An immediate consequence of the stoquastic product test (\Cref{thm:prodtest-constant}) is prover compression for \StoqMAk{}, which reduces the number of provers from $k$ to two: 

\begin{theorem}[Prover compression for \StoqMAk{} up to a quadratic gap loss]
\label{thm:stoq-k-to-2}
Let $k(n)$ be a polynomially bounded integer-valued function satisfying $k(n) > 2$. Let the completeness error $\epsilon(n)$ and the promise gap $\Delta(n)$ be functions such that $0 \leq \epsilon(n) < 1/2$ and $\Delta(n) > 0$. Then the following inclusion holds:
\[ \StoqMA_\ell \ssbra*{k,1-\epsilon,\Delta} \subseteq \StoqMA_{k\ell}\ssbra*{2,1-\lambda\epsilon, \frac{\gamma\Delta^2}{2}}. \]
Here, $\gamma \coloneqq c_{\mathrm{prod}}/{4}$, $\lambda \coloneqq \gamma\Delta$, and $c_\mathrm{prod}$ is the constant from~\Cref{thm:prodtest-constant}, which may be taken as $1/3$. In particular, the resulting gap shows that our construction preserves a constant gap $\Delta$. 
\end{theorem}

To establish the soundness, we also need the following simple proposition: 
\begin{proposition}[Product-state approximation preserves non-negativity]
\label{prop:nonnegative-product-approx}
Let $\ket{\xi}\in \calS_+(d_1\cdots d_k)$ be a non-negative state, and let $\ket{\Phi}=\ket{\phi_1}\otimes\cdots\otimes\ket{\phi_k}$ be a product state. For each $i\in[k]$, let $\ket{\phi_i^+}$ be obtained from
$\ket{\phi_i}$ by replacing every computational-basis amplitude by its absolute value, and set $\ket{\Phi^+}\coloneqq
    \ket{\phi_1^+}\otimes\cdots\otimes\ket{\phi_k^+}$. 
Then
\[
    \td(\ket{\xi},\ket{\Phi^+})
    \le
    \td(\ket{\xi},\ket{\Phi}).
\]
In particular, if $\ket{\Phi}=\ket{\psi}^{\otimes k}$, then
$\ket{\Phi^+}=(\ket{\psi^+})^{\otimes k}$.
\end{proposition}

\begin{proof}
Write $\ket{\xi}=\sum_z \xi_z\ket{z}$ with $\xi_z\ge 0$, and write
$\ket{\Phi}=\sum_z \Phi_z\ket{z}$. Then
\[
    \abs*{\innerprod{\xi}{\Phi}}
    =
    \abs*{\sum_z \xi_z \Phi_z}
    \le
    \sum_z \xi_z \abs{\Phi_z}
    =
    \innerprod{\xi}{\Phi^+}, 
\]
which implies the desired bound by \Cref{eq:pure-pure-Fuchs-vanGraaf}.
\end{proof}

We then continue to prove \Cref{thm:stoq-k-to-2}:
\begin{proof}[Proof of \Cref{thm:stoq-k-to-2}]
Let $V$ be the given $ \StoqMA_\ell \ssbra*{k,1-\epsilon,\Delta}$ verifier acting on witness registers $\sfA_1,\dots,\sfA_k$. Let $c(n) \coloneqq 1-\epsilon(n)$ and $s(n) \coloneqq c(n)-\Delta(n)$ be the completeness and soundness parameters of $V$, respectively. 
We now construct a new stoquastic verification circuit $W$ with witness registers $\sfA=\sfA_1\cdots \sfA_k$ and $\sfB=\sfB_1\cdots \sfB_k$. The witness state has the form $\ket{\Upsilon_1}_\sfA\otimes\ket{\Upsilon_2}_{\sfB}$, where $\ket{\Upsilon_1}$ and $\ket{\Upsilon_2}$ are called the first and second proof states, respectively, and each consists of $k\ell$ qubits. The verification circuit $W$ proceeds as follows: 
\begin{itemize}
    \item With probability $1-\lambda$, run the product-test branch from
    \Cref{protocol:stoq-product-test} on the two registers $\sfA$ and $\sfB$.
    \item With probability $\lambda$, ignore $\sfB$ and run the original verification circuit $V$ on the $k$ blocks $\sfA_1,\dots,\sfA_k$ of the first proof state. 
\end{itemize}
Since our protocol is a convex combination $(1-\lambda,\lambda)$ of two stoquastic verification circuits, the resulting verification circuit remains stoquastic.\footnote{The only subtlety is that the parameter $\lambda\in(0,1)$ should be a \emph{dyadic} number representable with $\poly(n)$ bits. We ignore the truncation error incurred by making $\lambda$ dyadic, since this error is \emph{exponentially small} and does not significantly affect the completeness and soundness parameters.\label{footnote:dyadic-truncation-error}}

\parheading{Completeness.} For \emph{yes} instances, there exist pure states $\ket{\psi_1},\dots,\ket{\psi_k}$ such that
$V$ accepts $\ket{\psi_1}\otimes\cdots\otimes\ket{\psi_k}$ with probability at least $c$.  Let
\[
    \ket{\Psi} \coloneqq \ket{\psi_1}\otimes\cdots\otimes\ket{\psi_k}.
\]

The honest two-proof witness for $W$ is $\ket{\Psi}_\sfA\otimes\ket{\Psi}_\sfB$. The product-test
branch accepts it with probability $1$, and the branch running $V$ accepts with probability at least $c$. Hence the completeness of $W$ is $c'\coloneqq 1-\lambda(1-c)$.

\parheading{Soundness.} Fix an arbitrary product witness $\ket{\rho}_\sfA\otimes\ket{\sigma}_\sfB$. 
Let $\eta(\rho)$ be defined as in \Cref{thm:prodtest-constant}, and set $x \coloneqq \sqrt{\eta(\rho)}$. 
Then there exists a product state
\[
    \ket{\Phi}=\ket{\phi_1}\otimes\cdots\otimes\ket{\phi_k}
\]
such that $\td(\ket{\rho},\ket{\Phi})\le x$. By \Cref{prop:opt-stoqSepVal-nonNeg} and \Cref{prop:nonnegative-product-approx}, it suffices to consider a non-negative product witness $\ket{\rho}_{\sfA}\otimes\ket{\sigma}_{\sfB}$, and the product state $\ket{\Phi}$ satisfying $\td(\ket{\rho},\ket{\Phi})\le x$ may also be chosen non-negative.
By the soundness $s$ of $V$, the branch running $V$ accepts $\ket{\Phi}$ with probability at most $s$. By the measurement bound for trace distance (\Cref{lemma:meas-bound-td}), together with $\td(\ket{\rho},\ket{\Phi})\le x$, the same branch accepts $\ket{\rho}$ with probability at most $s+x$. 
Combining this reasoning with the trivial upper bound $1$, its acceptance probability is at most $\min\cbra{1,s+x}$. 
By~\Cref{thm:prodtest-constant}, the product-test branch accepts
$\ket{\rho}_\sfA\otimes\ket{\sigma}_\sfB$ with probability at most $1-\gamma x^2$. Therefore,
\[
    \Pr[W \text{ accepts } \ket{\rho}\ket{\sigma}]
    \le (1-\lambda)(1-\gamma x^2)+\lambda\min\{1,\,s+x\}.
\]
Subtracting this from $c'$ yields
\begin{equation}
    \label{eq:prove-compression-bound}
    c'-\Pr[W \text{ accepts } \ket{\rho}\ket{\sigma}]
    \ge \lambda\bigl(c-\min\{1,\,s+x\}\bigr)+(1-\lambda)\gamma x^2.
\end{equation}
It remains to simplify the lower bound in \Cref{eq:prove-compression-bound}:
\begin{itemize}
    \item If $s+x\le 1$, then \Cref{eq:prove-compression-bound} gives
\[
    c'-\Pr[W \text{ accepts } \ket{\rho}\ket{\sigma}]
    \ge \lambda\Delta-\lambda x+(1-\lambda)\gamma x^2
    \eqqcolon f(x).
\]
This is a convex quadratic whose unique minimizer is
\[
    x_* \coloneqq \frac{\lambda}{2(1-\lambda)\gamma}=\frac{\Delta}{2(1-\lambda)}.
\]
Because $\lambda=\gamma\Delta\le \gamma<1/2$, we have $x_*\le \Delta\le 1-s$, so the minimizer lies
within the regime $s+x\le 1$.  Therefore, it follows that
\[
    f(x)\ge f(x_*)
    =\lambda\Delta-\frac{\lambda^2}{4(1-\lambda)\gamma}
    =\gamma\Delta^2\left(1-\frac{1}{4(1-\lambda)}\right)
    \ge \frac{\gamma\Delta^2}{2}.
\]

\item If $s+x\ge 1$, then \Cref{eq:prove-compression-bound} yields
\[
    c'-\Pr[W \text{ accepts } \ket{\rho}\ket{\sigma}]
    \ge -\lambda(1-c)+(1-\lambda)\gamma x^2
    \eqqcolon g(x).
\]
The function $g$ is increasing for $x\ge 0$, and the boundary between the two regimes is
$x_0=1-s\ge \Delta\ge x_*$.  Since $g(x_0)=f(x_0)$, we obtain
\[
    g(x)\ge g(x_0)=f(x_0)\ge f(x_*)\ge \frac{\gamma\Delta^2}{2}.
\]
\end{itemize}

Therefore, in all cases the acceptance probability is at most $c'-\frac{\gamma\Delta^2}{2}$, finishing the proof.
\end{proof}

\subsection{\SymStoqMAk{} vs.~\StoqMAk{}}
\label{subsec:SymStoqMAk-vs-StoqMAk}
We now establish the equivalence between \StoqMAk{} and \SymStoqMAk{} up to a quadratic gap loss. We first show that \SymStoqMAk{} can simulate \StoqMAk{} in \Cref{subsubsec:StoqMAk-in-SymStoqMAk}, and then, after introducing the stoquastic equality test in \Cref{subsubsec:stoq-equality-test}, prove that \StoqMAk{} can simulate \SymStoqMAk{} in \Cref{subsubsec:SymStoqMAk-in-StoqMAk}. 

\subsubsection{\texorpdfstring{$\StoqMAk \subseteq \SymStoqMAk$}{}}
\label{subsubsec:StoqMAk-in-SymStoqMAk}

To prove that $\StoqMAk\subseteq\SymStoqMAk$, as stated in \Cref{thm:StoqMAk-eq-SymStoqMAk}\ref{thmitem:StoqMA-in-SymStoqMA-ssbra}, the simple \emph{diagonal selection} trick in~\cite[Lemma 4.8]{ABDFS09} would suffice: each prover in $\SymStoqMAk$ sends the concatenation of the $k$ witness states $\ket{\psi_1}\ket{\psi_2}\cdots\ket{\psi_k}$ in the $\StoqMAk$ protocol. Then the $\SymStoqMAk$ verifier will simulate the $\StoqMAk$ verifier by taking only the $i$-th witness from the $i$-th prover of $\SymStoqMAk$. It follows that
\begin{equation}\label{eq:naive-SymStoqMAk-simulation}
    \StoqMA_\ell\ssbra*{k,c,s}\subseteq \SymStoqMA_{\ell k}\ssbra*{k,c,s}.
\end{equation}
The downside of the above trick is that the proof length blows up by a factor of $k$. 

\paragraph{A length-efficient symmetrization.} We provide a new simulation that only blows up the proof length by a multiplicative factor $O(\log k)$ and a new additive $O(\log^2 k)$ factor. 
The new idea is that each symmetric prover can send a bundle of $r = O(\log k)$ copies of the coherent superposition of the $k$ proofs from the $\StoqMAk$ provers, $\frac{1}{\sqrt{k}}\sum_{j\in[k]} \ket{j}\ket{\psi_j}$, instead of the wasteful concatenation of the $k$ proofs. 
Then by a \emph{coupon collector} principle, a typical branch of the $k$ symmetric proofs, now consisting of $O( k \log k )$ random proofs from the set of $k$ proofs of $\StoqMAk$, contains all $k$ different proofs with high probability. Therefore the $\SymStoqMAk$ verifier will be able to simulate the old $\StoqMAk$ verifier in a typical branch.
We start by formalizing our specific coupon collector lemma in the language of perfect matching. 
\begin{theorem}[Hall's theorem~\cite{Hall35}]
\label{thm:hall}
A bipartite graph with left side $L$ and right side $R$ has a matching of size $|L|$ if and only if
\[
    \abs*{N(S)}\ge \abs*{S}
    \qquad\text{for every } S\subseteq L,
\]
where $N(S)\subseteq R$ denotes the neighborhood of $S$ in $R$.
\end{theorem}

\begin{lemma}[Random matching bound]\label{lemma:randomized-matching-bound}
Let $B_{N,d}$ be the random bipartite multi-graph with left side $L$ and right side $R$, where $L= R=[N]$ and each left vertex independently chooses $d$ right vertices uniformly at random, with replacement. If $d=\lceil 12\ln N\rceil$, then $B_{N,d}$ has a matching covering all left vertices with probability at least $3/4$.
\end{lemma}

\begin{proof}
By Hall's theorem (\Cref{thm:hall}), it suffices to rule out a set $S\subseteq L$ with
$|N(S)|<|S|$.  Fix $s=|S|$.  If Hall fails for this $S$, then there is a set
$T\subseteq R$ of size $s-1$ such that all $ds$ choices made by vertices in
$S$ land inside $T$.  Therefore
\[
\Pr[\exists S: |S|=s,\ |N(S)|<s] \leq
\binom Ns\binom N{s-1}\left(\frac{s-1}{N}\right)^{ds}.
\]
Summing over $s=1,\ldots,N$ and using $d=\lceil12\ln N\rceil$, this total is
less than $1/4$.  Hence Hall's condition holds with probability at least $3/4$.
\end{proof}

With the above matching bound at our disposal, we prove our promised length-efficient upgrade to~\Cref{eq:naive-SymStoqMAk-simulation}: 
\begin{theorem}[Length-efficient symmetrization]\label{thm:StoqMAk-in-SymStoqMAk}
Let $k(n)$ be a polynomially bounded integer-valued function satisfying that $k(n) \geq 2$. Let $\epsilon(n)$ and $\Delta(n)$ be functions such that $0 \leq \epsilon(n) < 1/2$ and $\Delta(n) > 0$, and set $c(n) \coloneqq 1-\epsilon(n)$. 
Let $r \coloneqq \ceil*{12\ln k}$ and $m \coloneqq r\rbra*{\ell+\ceil*{\log k}}$, which are polynomially bounded integer-valued functions. Then the following inclusion holds:
\[
    \StoqMA_{\ell}\ssbra*{k,c,\Delta}
    \subseteq
    \SymStoqMA_{m}\ssbra*{k,c-\Delta/2,\Delta/2} .
\]
Equivalently, the new verifier has completeness at least $c-\Delta/2$ and
soundness at most $c-\Delta$.  Moreover, the new witness length is 
\[
    m(n) =O\rbra*{(\ell(n)+\log k(n))\log k(n)}\leq \poly(n).
\]
\end{theorem}

\begin{proof}
Fix an input $x$, and let $V$ be the original $k$-witness verifier. For \emph{yes} instances, there are non-negative $\ell$-qubit states $\ket{\psi_1},\ldots,\ket{\psi_k}$ such that
\begin{equation}\label{eq:old-completeness}
    \Pr\sbra*{V \text{ accepts }\ket{\psi_1}\totimes\cdots\totimes\ket{\psi_k}}
    \ge c .
\end{equation}
For \emph{no} instances, every product of $k$ non-negative $\ell$-qubit states is
accepted with probability at most $c-\Delta$.

\parheading{Honest symmetric proof.}
One slot consists of a label register and an $\ell$-qubit data register. The honest one-copy symmetric proof is
\[
    \ket{\Phi} \coloneqq 
    \rbra[\bigg]{
        \frac{1}{\sqrt{k}}\sum_{j\in[k]}\ket{j}\ket{\psi_j}
    }^{\otimes r} .
\]
Merlin sends $\ket{\Phi}^{\otimes k}$.  This is a valid non-negative symmetric witness of length $m$.

\parheading{New verifier.} For each copy $i\in[k]$, let $\sfL_{i,h}$ and $\sfD_{i,h}$ be the $h$-th label register and data register, respectively, of this copy. We expand $\ket{\Phi}^{\otimes k}$ with respect to the registers $\sfL\coloneqq (\sfL_{1,1},\cdots,\sfL_{k,r})$:
\begin{equation}
    \ket{\Phi}^{\otimes k} = \frac{1}{k^{kr/2}} \sum_{J\in[k]^{k\times r}}\bigotimes_{i\in[k]} \bigotimes_{h\in[r]} \ket{j_{i,h}}_{\sfL_{i,h}}\ket{\psi_{j_{i,h}}}_{\sfD_{i,h}}.
\end{equation}
Since different index-label branches $J = (j_{i,h})_{i\in[k],h\in[r]} \in [k]^{k \times r}$ \emph{do not} interfere before the verifier's final measurement, it suffices to analyze each index-label branch $J$ separately. For a fixed index-label branch, the verifier deterministically constructs,  whenever possible, a mapping scheme from the available labeled data registers to the roles of the original verifier. 

We now describe the construction for a fixed index-label $J$. Define a bipartite graph $G_J = (L\sqcup R, E_J)$, where the left side $L = [k]$ corresponds to the witness copies, and the right side $R = [k]$ corresponds to the roles of the original verifier. Place an edge between $i\in L$ and $a\in R$ if and only if the $i$-th witness copy, namely $\rbra*{j_{i,h}}_{h\in[r]}$, contains at least one slot labeled by $a$, that is, $j_{i,h_{\star}}=a$ for some $h_\star\in[r]$. If $G_J$ has a perfect matching, then the verifier chooses the lexicographically first perfect matching as the mapping scheme and routes the corresponding data registers to $V$. Otherwise, it runs a fixed dummy verifier that accepts with probability $c-\Delta$. The resulting stoquastic verification circuit for the length-efficient symmetrization is given in \Cref{protocol:stoq-length-efficient-symmetrization}.

\LinesNotNumbered
\begin{algorithm}[ht!]
    \UseProtocolCounter
    \caption{Stoquastic length-efficient symmetrization.}
    \label{protocol:stoq-length-efficient-symmetrization}
    \SetEndCharOfAlgoLine{.}
    \setlength{\parskip}{5pt}
    \SetKwInOut{Registers}{Registers}
    \SetKwInOut{Parameter}{Parameters}
    \SetKwIF{If}{ElseIf}{Else}{If}{:}{elif}{Else:}{}%

    \Parameter{$r\coloneqq \ceil*{12\ln k}$ and $m\coloneqq r\rbra*{\ell+\ceil*{\log k}}$.}
    \Registers{
    \begin{minipage}[t]{0.75\linewidth}
        \hangindent=3.2em
        \hangafter=1
        $\sfW_1,\ldots,\sfW_k$: the $k$ symmetric witness registers, where each $\sfW_i$ consists of $r$ label-data register pairs $(\sfL_{i,h},\sfD_{i,h})_{h\in[r]}$\;
        $\sfA_0$: the ancillary register initialized to $\ket{0}^{\otimes m_0}$\;
        $\sfA_+$: the ancillary register initialized to $\ket{+}^{\otimes m_+}$\;
        $\sfF$: a single-qubit register initialized to $\ket{0}.$
    \end{minipage}
    }

    \textbf{1.} Consider a fixed index-label branch $J=(j_{i,h})_{i\in[k],h\in[r]}$. 
    Define the bipartite graph $G_J=(L\sqcup R,E_J)$ as above. Let $P_\match$ be the predicate such that $P_\match(J)=1$ if and only if $G_J$ admits a perfect matching. There is a reversible circuit $\Gamma_\match$ such that 
    \[\Gamma_\match \ket{J}_\sfL\ket{0}_\sfF\ket{\bar{0}}_
    {\sfA_0} = \ket{J}_\sfL\ket{P_\match(J)}_\sfF\ket{\bar{0}}_
    {\sfA_0}.\]
    \textbf{2.} \If{ \textnormal{$G_J$ has a perfect matching, or equivalently, $P_\match(J)=1$}}{
        \smallskip
        \textbf{2.1} Write the lexicographically first perfect matching as $\mu_J:R\to L$. For each role $a\in[k]$, let 
        $h_J(a)\coloneqq \min\set{h\in[r]}{j_{\mu_J(a),h}=a}$.
        Route the ordered data registers $\sfD_{\mu_J(1),h_J(1)},\ldots,\sfD_{\mu_J(k),h_J(k)}$ to the $k$ input registers of $V$, and let $\Gamma_\route$ denote the corresponding reversible routing circuit\;
        \textbf{2.2} Apply the reversible circuit 
        $\Gamma_{V,J} \coloneqq \Gamma_\route^\dagger V^\dagger X_{\sfO}V\Gamma_\route$, with $\sfO$ denoting the output qubit register of $V$\;
    }\Else{
        Apply a fixed dummy reversible circuit $\Gamma_\dum$ such that the branch-overlap test (\Cref{lemma:branch-overlap-test}) applied to the circuit pair $(\Gamma_\dum,I)$ accepts with probability $c-\Delta$\;
    }
    \textbf{3.} Combine the branch operations into one reversible circuit $\Gamma_{\rm sim}$, and apply the branch-overlap test (\Cref{lemma:branch-overlap-test}) to the circuit pair $(\Gamma_{\rm sim},I)$
    on the initialized state of the witness registers together with $\sfA_0$, $\sfA_+$ and $\sfF$.
\end{algorithm}

\parheading{Completeness.}
Under the honest proof, each of the $k$ proof copies independently chooses $r=\ceil*{12\ln k}$ uniform labels from $[k]$. Therefore the matching bound (\Cref{lemma:randomized-matching-bound}) gives over a random index-label branch $J$,
\begin{equation}\label{eq:matching-prob}
    \Pr\sbra*{G_J\text{ has a perfect matching}}
    \ge \frac34 .
\end{equation}
On the matching branch, the routed tuple is exactly $\ket{\psi_1}\totimes\cdots\totimes\ket{\psi_k}$, so \Cref{eq:old-completeness} implies acceptance probability at least $c$.  On the no-matching branch, the dummy branch accepts with probability $c-\Delta$. Hence, it follows that
\begin{align*}
    \Pr\sbra*{\text{new verifier accepts}}
    &\ge
    \Pr\sbra*{\text{matching}}\,c
    +
    \Pr\sbra*{\text{no matching}}\,(c-\Delta) \\
    &= c-\Pr\sbra*{\text{no matching}}\Delta \\
    &\ge c-\frac{\Delta}{4}
    \ge c-\frac{\Delta}{2} .
\end{align*}
This proves the claimed completeness threshold.

\parheading{Soundness.}
Assume $x$ is a \emph{no} instance, and let $\ket{\Phi}$ be any non-negative one-copy state used by the symmetric prover.  The full witness is $\ket{\Phi}^{\otimes k}$. 
Over any index-label branch $J$, if $G_J$ has no perfect matching, the verifier accepts with probability $c-\Delta$ in~\Cref{protocol:stoq-length-efficient-symmetrization}.  If it has a perfect matching, the verifier selects one data block from each of $k$ distinct proof copies.  
After discarding unused registers, the old verifier receives a product state
\[
    \rho_1\otimes\cdots\otimes\rho_k .
\]
Here, each $\rho_i$ is obtained by
slicing a non-negative pure state
$\bigotimes_{h\in[r]} \ket{\psi_{j_{i,h}}}_{\sfD_{i,h}}$
over the discarded registers. Thus
$\rho_1\otimes\cdots\otimes\rho_k$ is a convex combination of non-negative product pure
states.
Consequently, by convexity and \Cref{prop:opt-stoqSepVal-nonNeg}, the original soundness bound also holds for arbitrary non-negative product mixed states. Thus the matching branch accepts with probability at most $c-\Delta$.
Every fixed index-label branch therefore contributes acceptance probability at most $c-\Delta$.  Averaging over label tables preserves the same bound. The new soundness is at most $c-\Delta$.

Combining completeness and soundness, the new promise gap is
\[
    \rbra*{c-\Delta/2}-\rbra*{c-\Delta}=\Delta/2,
\]
as claimed.
\end{proof}

\subsubsection{The stoquastic equality test}
\label{subsubsec:stoq-equality-test}
Let $\calH$ be the single-proof Hilbert space and let $U_\pi$ denote the unitary that permutes the $k$ witness registers according to $\pi\in S_k$. The projector onto the fully symmetric subspace satisfies
\[
    \Pi_{\sym}=\frac{1}{k!}\sum_{\pi\in S_k} U_\pi.
\]

The only minor subtlety is that the ideal construction would use a uniform superposition branch over the $k!$ permutations, but such a state cannot generally be prepared exactly from the allowed $\ket{0}$ and $\ket{+}$ ancillary qubits, since $k!$ need not be \emph{dyadic}. We avoid this issue by introducing a $q$-qubit dyadic branch register, initialized to $\ket{+}^{\otimes q}$, and mapping its $2^q$ basis branches to permutations as evenly as possible. 

\begin{lemma}[Stoquastic projector onto the symmetric subspace]
\label{lem:stoq-sym-branch-dyadic}
Let $b\geq 0$ be an integer, let $K\coloneqq k!$, set
$q\coloneqq \ceil*{\log K}+b$ and $N\coloneqq 2^q$, and write
$r\coloneqq N \bmod K$. For every pure state $\ket{\Phi}\in\calH^{\otimes k}$, the stoquastic verification circuit $V_q$ in \Cref{protocol:stoq-symmetric-space-projector-dyadic} satisfies
\[ \abs*{\Pr\sbra*{V_q \text{ accepts } \ket{\Phi}}
    - \rbra*{\frac{1}{2}+
    \frac12\bra{\Phi}\Pi_{\sym}\ket{\Phi}}}
    \leq \zeta_{\dyad}(q). \]
Here, the dyadic truncation error $\zeta_{\dyad}(q)$ is bounded by
\begin{equation*}
    \zeta_{\dyad}(q)
    \coloneqq \frac{r(K-r)}{KN}
    \leq \frac{K}{4N}
    \leq 2^{-b-2}.
\end{equation*}
Moreover, every tensor power $\ket{\psi}^{\otimes k}$ is accepted by $V_q$ with probability $1$.
\end{lemma}

\begin{proof}
Fix an ordering $S_k=\cbra*{\pi_0,\ldots,\pi_{K-1}}$.  Let
$a\coloneqq \floor*{N/K}$, so $N=aK+r$ with $0\leq r<K$. Define the balanced
map $f_q \colon \cbra*{0,\ldots,N-1}\to\cbra*{0,\ldots,K-1}$ by
\begin{equation}
    \label{eq:balanced-map}
    f_q(j)\coloneqq
    \begin{cases}
        \floor*{j/(a+1)},
        & 0\leq j<r(a+1),\\[1mm]
        r+\floor*{\rbra*{j-r(a+1)}/a},
        & r(a+1)\leq j<N .
    \end{cases}
\end{equation}
Equivalently, the balanced map $f_q$ assigns exactly $a+1$ dyadic branches to each of the first $r$ permutations and exactly $a$ dyadic branches to each of the remaining $K-r$ permutations. 
We now implement this balanced map by the stoquastic verification circuit provided in \Cref{protocol:stoq-symmetric-space-projector-dyadic}. 

\LinesNotNumbered
\begin{algorithm}[ht!]
    \UseProtocolCounter
    \caption{Stoquastic implementation of the symmetric-subspace projector.}
    \label{protocol:stoq-symmetric-space-projector-dyadic}
    \SetEndCharOfAlgoLine{.}
    \setlength{\parskip}{5pt}
    \SetKwInOut{Parameters}{Parameters}
    \SetKwInOut{Registers}{Registers}

    \Parameters{Extra precision bits $b$, $q=\ceil*{\log(k!)}+b$, and $N=2^q$.}
    \Registers{
    \begin{minipage}[t]{0.76\linewidth}
        $\sfA=\sfA_1\cdots\sfA_k$: witness register, storing $\ket{\Phi}_{\sfA}$\;
        $\sfJ$: $q$-qubit branch register, initialized to $\ket{+}_{\sfJ}^{\otimes q}$\;
        $\sfZ$: workspace for $\Gamma_q$, initialized to $\ket{\bar{0}}_{\sfZ}$.
    \end{minipage}}

    \textbf{1.} Define a reversible circuit $\Gamma_q$ that, on dyadic branch $j\in\cbra{0,\ldots,2^q-1}$, computes $t=f_q(j)$ and applies the permutation $\pi_t$ to the registers $\sfA_1,\ldots,\sfA_k$, uncomputing all workspace afterward. Equivalently, it holds that: 
    \[ \Gamma_q\ket{j}_{\sfJ}\ket{x_1}_{\sfA_1}\cdots\ket{x_k}_{\sfA_k}\ket{\bar{0}}_{\sfZ} =
        \ket{j}_{\sfJ}
        \ket{x_{\pi_{f_q(j)}^{-1}(1)}}_{\sfA_1}\cdots
        \ket{x_{\pi_{f_q(j)}^{-1}(k)}}_{\sfA_k}
        \ket{\bar{0}}_{\sfZ}.
    \]
    \textbf{2.} Apply the branch-overlap test
    (\Cref{lemma:branch-overlap-test}) to the circuit pair $(\Gamma_q,I)$ on the initialized state $\ket{+}_{\sfJ}^{\otimes q}\ket{\Phi}_{\sfA}\ket{\bar{0}}_{\sfZ}$.
\end{algorithm}

Let $p_t\coloneqq \abs*{f_q^{-1}(t)}/N$. Following \Cref{prop:opt-stoqSepVal-nonNeg}, without loss of generality, we can assume that $\ket{\Phi}$ is a non-negative state. As a consequence of \Cref{lemma:branch-overlap-test} and \Cref{protocol:stoq-symmetric-space-projector-dyadic}, the acceptance probability of $V_q$ is
\begin{equation}
    \label{eq:stoq-equality-test-dyadic}
    \Pr\sbra*{V_q \text{ accepts } \ket{\Phi}}
    = \frac12+\frac{1}{2N}\sum_{j=0}^{N-1} \bra{\Phi}U_{\pi_{f_q(j)}}\ket{\Phi} 
    = \frac12+\frac12\sum_{t=0}^{K-1}p_t \bra{\Phi}U_{\pi_t}\ket{\Phi}.
\end{equation}

On the other hand, the ideal uniform average over permutations would give
\begin{equation}
    \label{eq:stoq-equality-test-errorless}
    \frac{1}{2}+\frac{1}{2}\bra{\Phi}\Pi_{\sym}\ket{\Phi}
    = \frac{1}{2}+\frac{1}{2K}\sum_{t=0}^{K-1} \bra{\Phi}U_{\pi_t}\ket{\Phi} .
\end{equation}

The balanced map minimizes the total-variation distance from the uniform
measure on $S_k$ among all maps from $N$ dyadic branches to the $K$ permutations. Consequently, combining \Cref{eq:stoq-equality-test-dyadic,eq:stoq-equality-test-errorless} yields the following dyadic truncation-error bound: 
\begin{align*}
    \abs*{\Pr\sbra*{V_q \text{ accepts } \ket{\Phi}}
    - \rbra*{\frac{1}{2}+\frac{1}{2}\bra{\Phi}\Pi_{\sym}\ket{\Phi}}} &= \abs*{ \frac{1}{2} \sum_{t=0}^{K-1} \rbra*{p_t - \frac{1}{K}} \bra{\Phi}U_{\pi_t}\ket{\Phi} }\\
    &\leq \frac{1}{2} \sum_{t=0}^{K-1} \abs*{p_t - \frac{1}{K}} \cdot 1\\
    &=\frac{1}{2}\rbra*{
        r\cdot\frac{K-r}{KN}+(K-r)\cdot\frac{r}{KN}}\\
    &= \frac{r(K-r)}{KN} \eqqcolon \zeta_{\dyad}(q).
\end{align*}
Here, the second line uses the triangle inequality and the bound $\bra{\Phi}U_{\pi_t}\ket{\Phi} \leq 1$, and the third line follows from \Cref{eq:balanced-map}. 
Since $r(K-r)\leq K^2/4$ and $N\geq 2^bK$, we have
\[\zeta_{\dyad}(q)\leq K/(4N)\leq 2^{-b-2}.\] 

Finally, if $\ket{\Phi}=\ket{\psi}^{\otimes k}$, then
$U_{\pi_t}\ket{\Phi}=\ket{\Phi}$ for every $t$, so the acceptance probability is exactly $1$, as desired.
\end{proof}

\begin{lemma}[Stoquastic equality test]
\label{lem:sym-close}
Let $\ket{\Phi}=\ket{\phi_1}\otimes\cdots\otimes\ket{\phi_k}$ be a pure product state in $\mathcal{H}^{\otimes k}$.
Then there exists a one-register pure state $\ket{\psi}\in\mathcal{H}$ such that
\[
    \td \rbra[\big]{ \ket{\Phi},\ket{\psi}^{\otimes k} }
    \le 2\sqrt{1-\bra\Phi \Pi_{\sym}\ket\Phi}.
\]
\end{lemma}

\begin{proof}
Let $s_{\sym} \coloneqq \bra\Phi \Pi_{\sym}\ket\Phi.$ 
If $s_{\sym}=0$, the conclusion is trivial.  Assume henceforth that $s_{\sym}>0$ and
set
\[
    \ket{S} \coloneqq \frac{\Pi_{\sym}\ket{\Phi}}{\sqrt{s_{\sym}}}.
\]
Then $\ket{S}$ is a normalized symmetric state and
$\abs{\langle S\ket\Phi}^2=s_{\sym}.$
Since $\ket{\Phi}$ is itself a product state, the maximum product overlap of $\ket{S}$ is at least
$\sqrt{s_{\sym}}$.  For symmetric pure states, a closest product state may be chosen
symmetric~\cite[Lemma~1]{HKW+09}; equivalently, there exists a one-register state $\ket{\psi}$ such that
$\abs*{\innerprod*{S}{\psi^{\otimes k}}}^2\ge s_{\sym}.$
Therefore, the identity in \Cref{eq:pure-pure-Fuchs-vanGraaf} yields 
\[
    \td(\ket{S},\ket{\psi}^{\otimes k})
    \le \sqrt{1-s_{\sym}}.
\]
Also,
\[
    \td(\ket{\Phi},\ket{S})
    =\sqrt{1-s_{\sym}}.
\]
The claim now follows from the triangle inequality:
\[
    \td \rbra[\big]{ \ket{\Phi},\ket{\psi}^{\otimes k} }
    \le \td (\ket{\Phi},\ket{S})+\td(\ket{S},\ket{\psi}^{\otimes k})
    \le 2\sqrt{1-s_{\sym}}.\qedhere
\]
\end{proof}

\subsubsection{\texorpdfstring{$\SymStoqMAk \subseteq \StoqMAk$}{}}
\label{subsubsec:SymStoqMAk-in-StoqMAk}

A direct consequence of the stoquastic equality test (\Cref{lem:stoq-sym-branch-dyadic}) is the other direction of witness symmetrization, which establishes \Cref{thm:StoqMAk-eq-SymStoqMAk}\ref{thmitem:SymStoqMA-in-StoqMA-ssbra}:

\begin{theorem}
\label{thm:sym-to-stoq-k}
Let $k(n)$ be a polynomially bounded integer-valued function satisfying $k(n) \geq 2$. Let the completeness error $\epsilon(n)$ and the promise gap $\Delta(n)$ be functions such that $0 \leq \epsilon(n) < 1/2$ and $\Delta(n) > 0$. Then the following inclusion holds:
\[
    \SymStoqMA_\ell\ssbra*{k,1-\epsilon,\Delta} \subseteq \StoqMA_\ell\ssbra*{k, 1-\lambda\epsilon,\frac{\Delta^2}{16}}.
\]
Here, $\lambda \coloneqq \Delta/8$. 
In particular, the resulting gap shows that the construction preserves both constant and inverse-polynomial gaps. 
\end{theorem}

\begin{proof}
Let $V$ be the given $\SymStoqMA_\ell\ssbra*{k,1-\epsilon,\Delta}$ verifier acting on witness registers $\sfA_1,\cdots,\sfA_k$. Let $c(n) \coloneqq 1-\epsilon(n)$ and $s(n) \coloneqq c(n)-\Delta(n)$ be the completeness and soundness parameters of $V$, respectively.
We now construct a new stoquastic verification circuit $W$, using the same witness registers $\sfA_1,\cdots,\sfA_k$ as follows:
\begin{itemize}
    \item With probability $1-\lambda$, run the symmetry branch from \Cref{protocol:stoq-symmetric-space-projector-dyadic} with parameters $b \coloneqq \ceil*{\log\rbra*{16/\Delta^2}}$, $q \coloneqq \ceil*{\log(k!)+b}$, and $\xi \coloneqq \xi_\dyad(q)$. 
    \item With probability $\lambda$, run the original verification circuit $V$.
\end{itemize}
Since our protocol is a convex combination $(1-\lambda,\lambda)$ of two stoquastic verification circuits, the resulting verification circuit remains stoquastic.\footnote{For truncation errors caused by choosing $\lambda$, see \Cref{footnote:dyadic-truncation-error}.}

\parheading{Completeness.} For \emph{yes} instances, the honest witness is of the form $\ket{\psi}^{\otimes k}$.  The
symmetry branch accepts with probability $1$, and the original branch accepts with probability at least $c=1-\epsilon$.  Hence the completeness is $c'=1-\lambda(1-c).$

\parheading{Soundness.} Fix an arbitrary product witness
\[
    \ket{\Phi}=\ket{\phi_1}\otimes\cdots\otimes\ket{\phi_k}.
\]
Let $\eta \coloneqq 1-\bra\Phi\Pi_{\sym}\ket\Phi$. 
By \Cref{lem:stoq-sym-branch-dyadic}, the symmetry branch accepts with probability at most $1-\eta/2 + \zeta$, the chosen parameters ensure that 
\begin{equation}
    \label{eq:dyadic-symmetric-projector-bound}
    \xi \leq 2^{-b-2} \leq \Delta^2/64.
\end{equation}

By \Cref{lem:sym-close}, there exists a one-register state $\ket{\psi}$ such that
\begin{equation}
    \label{eq:witness-sym-T-bound}
     \td \rbra[\big]{ \ket{\Phi},\ket{\psi}^{\otimes k} } \le 2\sqrt{\eta}.
\end{equation}
By \Cref{prop:opt-stoqSepVal-nonNeg} and \Cref{prop:nonnegative-product-approx}, it suffices to consider a non-negative product witness $\ket{\Phi}$, and the state $\ket{\psi}$ in \Cref{eq:witness-sym-T-bound} may also be chosen non-negative; hence the soundness of $V$ applies to $\ket{\psi}^{\otimes k}$, and the measurement bound for trace distance (\Cref{lemma:meas-bound-td}) implies that $V$ accepts $\ket{\Phi}$ with probability at most $s+2\sqrt{\eta}$.
Combining this reasoning with the trivial upper bound $1$, its acceptance probability becomes at most $\min\cbra{1,s+2\sqrt{\eta}}$. 
Therefore, 
\[
    \Pr[W \text{ accepts } \ket{\Phi}]
    \le (1-\lambda)\left(1-\frac{\eta}{2}+\zeta\right)+\lambda\min\set{1,\,s+2\sqrt{\eta}}.
\]
Subtracting this from $c'$ gives
\begin{equation}
    \label{eq:witness-sym-bound}
    c'-\Pr[W \text{ accepts } \ket{\Phi}]
    \ge \lambda \rbra*{ c-\min\set{1,\,s+2\sqrt{\eta}} } +\frac{1-\lambda}{2}\eta -(1-\lambda)\zeta.
\end{equation}
Set $x \coloneqq \sqrt{\eta}$. Since $x\in[0,1]$, it remains to simplify the lower bound in \Cref{eq:witness-sym-bound}:
\begin{itemize}
    \item If $s+2x\le 1$, then \Cref{eq:witness-sym-bound} yields
\[
    c'-\Pr[W \text{ accepts } \ket{\Phi}]
    \ge \underbrace{\lambda\Delta-2\lambda x+\frac{1-\lambda}{2}x^2}_{\coloneqq f(x)} -(1-\lambda)\zeta.
\]
This is a convex quadratic whose unique minimizer is
\[
    x_* \coloneqq \frac{2\lambda}{1-\lambda} \le 4\lambda=\frac{\Delta}{2}\le \frac{1-s}{2}. 
\]
Here, the first inequality holds because $\lambda=\Delta/8\le 1/8$.
Consequently, the minimizer lies within the regime $s+2x\le 1$, and it follows that
\begin{equation}
    \label{eq:sym-dyadic-optimizer-bound}
    f(x)\ge f(x_*) =\lambda\Delta-\frac{2\lambda^2}{1-\lambda}.
\end{equation}

Plugging $1/(1-\lambda) \leq 8/7$ and \Cref{eq:dyadic-symmetric-projector-bound} into \Cref{eq:sym-dyadic-optimizer-bound}, we obtain
\[ f(x)-(1-\lambda)\zeta \geq f(x)-\zeta \geq \frac{\Delta^2}{8} - \frac{\Delta^2}{28} - \frac{\Delta^2}{64} = \frac{33\Delta^2}{448} \geq \frac{\Delta^2}{16}. \]
Therefore, the first case gives $c'-\Pr[W \text{ accepts } \ket{\Phi}] \geq \Delta^2/16$. 

    \item If $s+2x\ge 1$, then \Cref{eq:witness-sym-bound} implies
\[
    c'-\Pr[W \text{ accepts } \ket{\Phi}]
    \ge \underbrace{-\lambda(1-c)+\frac{1-\lambda}{2}x^2}_{g(x)} - (1-\lambda)\zeta.
\]
The function $g$ is increasing for $x\ge 0$, and the boundary between the two regimes is
$x_0=(1-s)/2$.  Since $x_0\ge x_*$, $g(x_0)=f(x_0)$, and $(1-\lambda)\zeta \leq \zeta$, we obtain
\[
    g(x)-(1-\lambda)\zeta \ge g(x_0)-\zeta=f(x_0)-\zeta\ge f(x_*)-\zeta\ge \frac{\Delta^2}{16}. 
\]
\end{itemize}
Hence, in all cases the acceptance probability is at most $c'-\frac{\Delta^2}{16}$, finishing the proof.
\end{proof}

%%%%%%%%%%%%%%%%%%%%%%%%%%%%%%%%%%%%%%%%%%%%%%%%%%%%%%%%%%%
%%%%%%%%%%%%%%%%%%%%%%%%%%%%%%%%%%%%%%%%%%%%%%%%%%%%%%%%%%%
%%%%%%%%%%%%%%%%%%%%%%%%%%%%%%%%%%%%%%%%%%%%%%%%%%%%%%%%%%%

\section{The power of \texorpdfstring{$\StoqMAtwo$}{} via distribution testing}
\label{sec:stoqma-power}

In this section, we study the power of \StoqMAtwo{} with short proofs. We first establish that, with $O(\sqrt n)$ unentangled proofs of length $O(\log{n})$, a stoquastic verifier can certify \NP{} languages, with a \emph{constant} promise gap:
\begin{theorem} 
\label{thm:power-stoqmak-np}
    For any constant $\epsilon>0$, the following inclusion holds:
    \[ \NP\subseteq \StoqMA_{O(\log n)} \ssbra*{O(\sqrt n), 1-\epsilon,\frac12-2\epsilon}. \]
\end{theorem}

In view of~\Cref{thm:sym-to-stoq-k,thm:StoqMAk-w-parallel-repetition-loss}, it suffices to prove the following symmetric version:
\begin{restatable}{theorem}{SymStoqABDFS}
\label{thm:power-SymStoqMAk-np}
    For any constant $\epsilon>0$, there exists a constant $\Delta=\Omega(1)$ for which the following inclusion holds:
    \[ \NP\subseteq \SymStoqMA_{O(\log n)} \ssbra*{O(\sqrt n), 1-\epsilon,\Delta}. \]
\end{restatable} 

The symmetric protocol in \Cref{thm:power-SymStoqMAk-np} is, in spirit, analogous to the famous $\QMA_{\log} (\sqrt n)$ protocol of Aaronson, Beigi, Drucker, Fefferman, and Shor~\cite{ABDFS09}, as well as Chen--Drucker's improved protocol~\cite{CD10}. Those protocols, however, use non-stoquastic operations and therefore are not directly applicable here. 
We design a stoquastic protocol for Dinur's PCP~\cite{Dinur07} by leveraging Paninski's uniformity test for probability distributions. One extra technical benefit is that Paninski's uniformity test is not intertwined with other tests, making our soundness analysis more modular.

\begin{remark}[Stoquastizing the uniformity test loses perfect completeness]
    \label{remark:stoq-uniformity-test-loss}
    Notably, there is a tradeoff between the number of provers $K$ and the completeness $1-\epsilon$ in \Cref{thm:power-SymStoqMAk-np}: the required value of $K$ contains a multiplicative factor of $\log\frac1\epsilon$, which comes from the completeness loss of the \emph{distributional} uniformity test and is known to be necessary~\cite{DGPP18}. Consequently, our stoquastic version of the uniformity test cannot achieve perfect completeness, but it can achieve completeness as high as $1-2^{-\polylog(n)}$, at the cost of \emph{slightly} increasing the number of provers from $O(\sqrt n)$ to $\widetilde O(\sqrt n)$.
\end{remark}

Following \Cref{thm:stoq-k-to-2} and \Cref{remark:stoq-uniformity-test-loss}, we obtain another corollary: for a stoquastic verifier, two proofs of length $\widetilde O (\sqrt n)$ suffice to certify $\NP$ languages. Applying parallel repetition for \StoqMAtwo{} from~\Cref{thm:StoqMAk-w-parallel-repetition-loss}, which gives error reduction when completeness error $\epsilon$ is negligible, then yields
\begin{restatable}{theorem}{StoqABDFS} 
\label{thm:power-stoqma-np}
    For every $\epsilon(n)=\Omega\rbra*{2^{-\polylog(n)}}$, the following inclusion holds:
    \[ \NP\subseteq \StoqMA_{\widetilde{O}(\sqrt n)}\ssbra*{2,1-\epsilon,\frac{1}{2}-2\epsilon}. \]
\end{restatable}

\vspace{1em}
The same intuition from distribution testing also suggests a two-prover \StoqMAtwo{} protocol for $\NP$ languages with logarithmic-size witnesses and an inverse-polynomial gap. This result can be viewed as a $\StoqMAtwo$ counterpart of the Blier--Tapp protocol~\cite{BT07}: 
\begin{restatable}{theorem}{NPinStoqMAlogTwo} 
\label{thm:power-stoqma-log-np}
$\displaystyle \NP\subseteq \StoqMA_{O(\log n)}\ssbra*{2,1-\Theta\left(\frac{1}{n^2}\right), \,\Theta\left(\frac{1}{n^2}\right)}.$
\end{restatable}

Scaling this up to polynomial-size proofs using a sufficiently structured PCP for $\NEXP$, say the one from~\cite[Theorem 7.4]{JW23}, one obtains that 
\[\NEXP\subseteq\StoqMA\ssbra*{2,1-2^{-\poly(n)},2^{-\poly(n)}}.\]

The reverse inclusion is straightforward. Following the standard convention, we call $\StoqMAtwo$ with an exponentially small gap $\PreciseStoqMAtwo$. Therefore, together with the easy inclusion $\PreciseQMA(2)\subseteq \NEXP$ (see, e.g.,~\cite{Pereszlenyi12}), we obtain:
\begin{corollary}
    $\PreciseStoqMAtwo = \NEXP.$
\end{corollary}

\vspace{1em}
To establish the above theorems, we record Dinur's PCP as our starting point.
\begin{definition}[Dinur's Gap Constraint Graph Problem, \GapCG{}, adapted from~{\cite[Definition 1.5]{Dinur07}}]
Fix constants $d \in \bbN$, $Q \in \bbN$, and $\eta > 0$.  An instance of $(1,\eta)\text{-}\GapCG_{d,Q}$ consists of a $d$-regular graph $G=(V,E)$, an alphabet $\Sigma$ of size $Q$, and for every undirected edge $(u,v)\in E$ a binary relation $R_{uv} \subseteq \Sigma \times \Sigma$.
The promise is that exactly one of the following holds:
\begin{itemize}
    \item \emph{\textbf{Yes}:} There exists a labeling $\iota:V\to\Sigma$ satisfying every edge constraint, i.e.
    \[
        (\iota(u),\iota(v)) \in R_{uv}, \qquad \forall (u,v)\in E;
    \]
    \item \emph{\textbf{No}:} Every labeling $\iota:V\to\Sigma$ violates at least an $\eta$-fraction of the undirected edges.
\end{itemize}
\end{definition}

\begin{theorem}[Dinur's PCP Theorem, {adapted from~\cite[Theorem 1.7]{Dinur07}}]
\label{thm:dinur-gapcg-form}
There exist absolute constants $d,Q\in\mathbb N$ and $\eta>0$, together with a $\widetilde{O}(m)$-time 
reduction from $\SAT$ instances of size $m$ to $(1,\eta)\text{-}\GapCG_{d,Q}$ instances whose underlying constraint graph $G=(V,E)$ has size $\widetilde{O}(m)$, meaning that $|V|=\widetilde{O}(m)$ and $|E|=\widetilde{O}(m)$. 
\end{theorem}

\vspace{1em}
In the remainder of this section, we first give a $\SymStoqMA_{O(\log n)}(\sqrt n)$ protocol for $(1,\eta)$-\GapCG{} in \Cref{subsec:NP-in-SymStoqMAsqrtN}, which proves \Cref{thm:power-SymStoqMAk-np}. We then apply the same intuition to obtain a more direct $\StoqMA_{O(\log n)}(2)$ protocol with an inverse-polynomial promise gap in \Cref{subsec:stoqma-log-two-prover}, which proves \Cref{thm:power-stoqma-log-np}.

\subsection{A \texorpdfstring{$\SymStoqMA(\sqrt n)$}{} protocol for \texorpdfstring{$(1,\eta)$}{}-\GapCG{}}
\label{subsec:NP-in-SymStoqMAsqrtN}
We begin by designing a $\SymStoqMA(O(\sqrt n))$ protocol that certifies $\NP$ languages with a constant promise gap:
\SymStoqABDFS*

We will need the following classical theorem of Paninski~\cite{Paninski08}, in the clean formulation recorded in~\cite[Theorem 5.1]{Clement20}, to analyze the product weights induced by the witnesses:
\begin{theorem}[Uniformity test for distributions, adapted from~{\cite[Theorem 5.1]{Clement20}}]
\label{thm:paninski}
Given sample access to an unknown distribution $q$ over a finite set $\Omega$ of size $n$. For every fixed $\delta>0$ and $\epsilon>0$, there is an efficiently computable predicate $\calT$ on
\[
    K= O\bigg(\frac{\sqrt n}{\delta^2}\log\frac{1}{\epsilon}\bigg)
\]
i.i.d. samples from $q$, such that
\begin{itemize}[itemsep=0em]
    \item If $q$ is the uniform distribution, $\calT$ accepts w.p. at least $1-\epsilon$;
    \item If $q$ is $\delta$-far from uniform in total variation distance, $\calT$ accepts w.p. at most $\epsilon$.
\end{itemize}
Furthermore, $\calT$ depends only on the collision pattern among the samples.
\end{theorem}

Given a \GapCG{} instance $\frakI$, we identify $V\times\Sigma$ with a subset of the computational basis on $\ceil*{ \log{|V|} } + \ceil*{ \log{|\Sigma|} } = O(\log |V|)$ qubits.
Suppose the input is a \emph{yes} instance with a satisfying labeling $\iota:V\to\Sigma$, and let $|V|=n$.  The honest one-copy witness is the subset state
\[ \ket{\phi_\rmL} = \frac{1}{\sqrt{|V|}}\sum_{v\in V} \ket{v,\iota(v)}.\]
Fix a constant $C>0$, to be specified later, and let $K \coloneqq \lceil C\sqrt n\rceil.$ We ask for $K$ copies of this state, namely $\ket{\Phi} \coloneqq \ket{\phi_\rmL}^{\otimes K}$, so that the witnesses encode the satisfying labeling. 

\paragraph{Branch-local protocol for vertex-label branches.} Let $\sfW_i$ be the register containing the $i$-th witness factor, and let $\sfW \coloneqq (\sfW_1, \cdots, \sfW_K)$. An arbitrary non-negative one-copy state $\ket{\psi}$ can be expressed as
\[ \ket{\psi} = \sum_{(v,a)\in V\times \Sigma} \alpha_{v,a} \ket{v,a}, \quad\text{where } \;\; \alpha_{v,a}\geq 0, \forall (v,a)\in V\times \Sigma\;\; \text{and}\;\; \sum_{(v,a)\in V\times \Sigma} \alpha^2_{v,a} = 1. \]
Thus, a computational-basis branch of the symmetric witness $\ket{\Psi} \coloneqq \ket{\psi}^{\otimes K}$ is labeled by 
\[ x = ((v_1,a_1),\ldots,(v_K,a_K))\in (V\times\Sigma)^K. \]

We refer to $(v_i,a_i)$ as the \emph{vertex-label factor branch} in the $i$-th tensor factor, or simply as the \emph{vertex-label branch} when the tensor-product structure is clear. The full tuple $x$ is a \emph{$K$-vertex-label branch}. These branch labels are not observed outcomes: the actual stoquastic verifier \emph{never measures} $\sfW_1,\cdots,\sfW_K$ before the final Hadamard-basis measurement. The terminology serves only to identify the computational-basis branches on which the classical reversible circuit is coherently evaluated. 

One can verify directly that the branch weight of the $K$-vertex-label branch $x$ is
\begin{equation}
    \label{eq:branch-weight-product}
    \abs*{ \prod\nolimits_{i\in[K]} \alpha_{v_i,a_i} }^2 = \prod\nolimits_{i\in[K]} \alpha_{v_i,a_i}^2 = \prod\nolimits_{i\in[K]} p(v_i,a_i), \quad\text{where}\;\; p(v,a)\coloneqq \alpha_{v,a}^2.
\end{equation}
Here, $p(v,a)$ denotes the one-copy squared-amplitude distribution. Following \Cref{eq:branch-weight-product}, the squared-amplitude weights on $K$-vertex-label branches are exactly the product weights $p^{\otimes K}$. This product-weight identity is the only sense in which distribution testing enters the proof. We also define the probability distribution $q$ over $V$ as the vertex marginal of the one-copy distribution $p$, equivalently as the distribution that would be obtained by measuring $\ket{\psi}$ in the computational basis and keeping only the vertex register: 
\begin{equation}
    \label{eq:one-copy-vertex-marginal}
    \forall v\in V, \quad  q(v) \coloneqq \sum_{a\in\Sigma} \alpha_{v,a}^2 . 
\end{equation}

We are now ready to state the \emph{branch-local} protocol, presented in \Cref{protocol:branchLocal-NP-StoqMAkLog}, which describes the behavior of the $\SymStoqMA(K)$ verifier within a fixed $K$-vertex-label branch, including branch-local versions of the uniformity and consistency tests. 
\LinesNotNumbered
\begin{algorithm}[ht!]
    \UseProtocolCounter
    \caption{Branch-local protocol for $\NP$ on a fixed $K$-vertex-label branch}
	\label{protocol:branchLocal-NP-StoqMAkLog}
    \SetEndCharOfAlgoLine{.}
    \setlength{\parskip}{5pt}
    \SetKwFor{While}{}{:}{}
    \SetKwInOut{Parameter}{Parameters}

    \textbf{1.} Consider a $K$-vertex-label branch of $\ket{\Psi} = \ket{\psi}^{\otimes K}$, denoted by 
    \[ \rbra*{ (v_1,a_1), \dots, (v_{K},a_{K}) } \in \rbra*{V\times\Sigma}^K. \] 
        
    \textbf{2.} \While{\textnormal{Apply both the uniformity and consistency test}}{

    \medskip
    
    \textbf{2.1} \textbf{Uniformity test:} 
    Let $\calT$ be the predicate from Paninski's uniformity tester in~\Cref{thm:paninski}, distinguishing the uniform distribution $\mu_V$ over $V$ from distributions $q$ such that $\TV(q,\mu_V)\ge \delta$, for some fixed distance parameter $\delta>0$ with error parameter $\epsilon$.
    Apply the deterministic collision-based predicate $\calT$ to the vertex components $v_1,\dots,v_{K}$. 

    \medskip
    \textbf{2.2} \textbf{Consistency test:} \emph{Accept} if neither of the following occurs
        \begin{enumerate}[label=(\roman*),itemsep=-2pt]
            \item Exist  $i\neq j$ such that $v_i=v_j$ but $a_i\neq a_j$;
            \item Exist $i\neq j$ such that $(v_i,v_j)\in E$ and
            $(a_i,a_j)\notin R_{v_iv_j}$.
        \end{enumerate}
    }
    \textbf{3.} Accept if and only if both tests accept.
\end{algorithm}

\paragraph{Implementing the branch-local protocol.} 
To implement \Cref{protocol:branchLocal-NP-StoqMAkLog} as a stoquastic verification circuit, we consider the predicate $A_\unif$ underlying Paninski's uniformity tester (\Cref{thm:paninski}) and write the acceptance condition of the consistency test as the predicate $A_\cons$. Since both predicates are deterministic, we obtain a deterministic predicate 
\[ A_\branch \coloneqq A_\unif \wedge A_\cons,\]
which corresponds to the acceptance condition of \Cref{protocol:branchLocal-NP-StoqMAkLog}. 

Let $\sfA_0$ and $\sfF$ denote the register containing $\ket{0}$-ancillary qubits and a single-qubit register initialized to $\ket{0}$. Using the standard reversible-computation convention, we obtain a classical reversible circuit $\Gamma_\branch$ satisfying that
\[ \forall x\in\rbra*{V\times \Sigma}^K, \quad 
\Gamma_\branch \ket{x}_\sfW \ket{\bar{0}}_{\sfA_0}\ket{0}_\sfF = \ket{x}_\sfW \ket{\bar{0}}_{\sfA_0}\ket{A_\branch(x)}_\sfF. \]

Finally, we obtain the resulting stoquastic verification circuit by applying the branch-overlap test (\Cref{lemma:branch-overlap-test}) to the circuit pair $(\Gamma_\branch,I)$ on the input state $\ket{\Psi}_{\sfW}\ket{\bar{0}}_{\sfA_0}\ket{0}_\sfF$, where $\ket{\Psi}=\ket{\Phi}$ for \emph{yes} instances. 

\subsubsection{Analysis of \texorpdfstring{\Cref{protocol:branchLocal-NP-StoqMAkLog}}{}}
Now we analyze our branch-local verifier. Only the one-copy squared-amplitude distribution $p$ over $V \times \Sigma$, as defined in \Cref{eq:branch-weight-product}, matters, because the $K$-factor branch weights factor as $p^{\otimes K}$, which corresponds, from the distribution testing perspective, to $K$ i.i.d. samples drawn from $p$. The vertex marginal $q$ over $V$ also has been defined in \Cref{eq:one-copy-vertex-marginal}.

\paragraph{Parameter choices.}
There are several running parameters in our analysis,
\begin{itemize}[itemsep=0em]
    \item[]  $\epsilon$, is the target error parameter in \Cref{protocol:branchLocal-NP-StoqMAkLog}, which is positive and can be arbitrarily small; 
    \item[] $\eta$, is the fixed constant from Dinur's PCP~\Cref{thm:dinur-gapcg-form}, representing the minimum fraction of violating edges for any labeling in a \emph{no} instance; 
    \item[] $\delta$, is the distance parameter used in \Cref{protocol:branchLocal-NP-StoqMAkLog}, which we set to 
    $\delta = \eta/48;$
    \item[] $\kappa$, an auxiliary parameter, which we set to $\kappa= \eta/64.$
\end{itemize}

\paragraph{Completeness.}
For \emph{yes} instances, the induced vertex distribution $q$ is uniform, and the consistency test is passed trivially with probability $1$. Overall, the verifier accepts with probability at least $1-\epsilon$ for large enough $K=O(\sqrt n)$ by~\Cref{thm:paninski}, where $\epsilon>0$ can be an arbitrarily small constant, or even as small as $2^{-\polylog(n)}$ if we allow $K=\widetilde O(\sqrt n)$. Hence,
\begin{lemma}[Completeness of \Cref{protocol:branchLocal-NP-StoqMAkLog}]
\label{lem:completeness}
    The completeness of the classical branch predicate is at least $1-\epsilon$, and consequently, the completeness of the branch-local verifier $\calV$ is at least $1-\epsilon/2$.
\end{lemma}

\paragraph{Soundness.}
We now prove the soundness:
\begin{lemma}[Soundness of \Cref{protocol:branchLocal-NP-StoqMAkLog}]
    The soundness of the classical branch predicate is at most $\epsilon$, and consequently, the completeness of the branch-local verifier $\calV$ is at most $1/2+\epsilon/2$.
\end{lemma}
If $q$ is $\delta$-far from the uniform distribution $\mu_V$, the uniformity test rejects with probability $1-\epsilon$. 
Hence, assume from now on that $\TV(q,\mu_V)<\delta.$
Define the following sets
\[
    L \coloneqq \set{v\in V}{q(v)<\tfrac{1}{2n}},
    \qquad
    H \coloneqq \set{v\in V}{q(v)>\tfrac{2}{n}},
    \qquad
    G \coloneqq V\setminus(L\cup H).
\]
Then because every $v\in L$ contributes more than $\frac{1}{2n}$ to the total deficit from $\mu_V$, one has
        $\abs{L}<2\delta n;$
and analogously because every $v\in H$ contributes more than $\frac{1}{n}$ to the total excess over $\mu_V$, one has
        $\abs{H}<\delta n.$
Therefore,
    $\abs{V\setminus G}<3\delta n.$

For each vertex $v$ with $q(v)>0$, define the conditional label distribution $r_v:\Sigma\to \bbR$
and the \emph{plurality} label $\iota^*:V\to\Sigma$, that
\[
    r_v(a) \coloneqq \frac{p(v,a)}{q(v)},\qquad \iota^*(v) \coloneqq \arg\max_{a\in\Sigma} r_v(a).
\]
For vertices with $q(v)=0$, choose $\iota^*(v)$ arbitrarily.  Define the local \emph{ambiguity}
\[
    b_v \coloneqq 1-r_v(\iota^*(v)).
\]
Thus $b_v=0$ iff the label at $v$ is deterministic. 
We split the soundness analysis into two cases.
\begin{enumerate}[label=(\alph*)]
    \item\label{enu:ambig-case} There is substantial ambiguity on good vertices. Then the consistency test is likely to see a same vertex sampled twice with different labels by the birthday paradox, and thus rejects. 
    \item\label{enu:determ-case} The label is almost deterministic on a good fraction of vertices. Then the consistency test is likely to see adjacent vertices with unsatisfying labels by the birthday paradox, and thus rejects.
\end{enumerate}

The following generalized birthday paradox, whose proof is deferred to \Cref{subsec:omitted-birthday}, captures the birthday-paradox estimate required by the two cases:
\begin{restatable}[A generalized birthday paradox]{lemma}{birthdayParadox}
\label{lem:birthday}
Let $\Omega$ be a finite set of size $n$,  $\mu$ be a probability distribution on $\Omega$, and $B\subseteq \Omega\times\Omega$ be a symmetric relation representing the ``birthday collisions.''
For some subset $\Omega_0\subseteq \Omega$, if for some parameters $\alpha,\beta>0$ and some integer 
$D>0$, the following are true
\begin{enumerate}[label=\upshape(\arabic*), itemsep=0em]
    \item $\mu(x)\le \alpha/n$ for every $x\in\Omega_0$;
    \item for every $x\in\Omega_0$, the number of $y\in\Omega_0$ with $(x,y)\in B$ is at most
    $D$;
    \item the bad-pair weight on $\Omega_0$ satisfies
    \[
        \lambda \coloneqq \sum_{(x,y)\in B\cap(\Omega_0\times\Omega_0)} \mu(x)\mu(y) \ge \frac{\beta}{n}.
    \]
\end{enumerate}
Then for any constant $\epsilon>0$, there is a large enough $C$ such that the
following holds for all sufficiently large $n$:
Let $X_1,\dots,X_K$ be i.i.d. samples from $\mu$ with $K=C\sqrt n$,
\[
    \Pr\Bigl[\exists i<j \text{ such that } (X_i,X_j)\in B\cap(\Omega_0\times\Omega_0)\Bigr]
    \ge 1-\epsilon.
\]
\end{restatable}

With this generalized birthday paradox at our disposal, we prove the promised soundness of our protocol by considering the aforementioned two cases:

\subparheading{Case~\ref{enu:ambig-case}:} There is substantial ambiguity among good vertices.
Suppose
\[
    \sum_{v\in G} q(v)b_v\ge \kappa
\]
for $\kappa > 0$, whose exact choice $\kappa$ is irrelevant in this case. 
To apply the birthday paradox, define the sample space
\[
    \Omega \coloneqq V\times\Sigma 
    \quad\text{and}\quad
    \mu(v,a) \coloneqq p(v,a),
\]
and define 
\[
    \Omega_0^{\ver} \coloneqq G\times\Sigma\subseteq \Omega.
\]

Let $B_{\ver}\subseteq \Omega\times\Omega$ be the symmetric bad-pair relation consisting of all ordered pairs
\[
    ((v,a),(v,a')) \quad \text{with } v\in G \text{ and } a\neq a'.
\]

For $x=(v,a)\in\Omega_0^{\ver}$, we have
\[
    \mu(x)=p(v,a)\le q(v)\le \frac{2}{n}.
\]

Moreover, every $x\in\Omega_0^{\ver}$ has at most $|\Sigma|-1$ neighbors inside $\Omega_0^{\ver}$ with respect to $B_{\ver}$. The bad-pair weight equals
\[
    \lambda_{\ver} \coloneqq \sum_{(x,y)\in B_{\ver}} \mu(x)\mu(y) 
    =\sum_{v\in G} q(v)^2 \sum_{a\neq a'} r_v(a)r_v(a') 
\]

Since $r_v(\iota^*(v))=1-b_v$ and $1-b_v\ge 1/|\Sigma|$, one has the elementary bound
\[
    \sum_{a\neq a'} r_v(a)r_v(a')  \ge 2b_v(1-b_v)\ge \frac{2}{|\Sigma|}b_v.
\]

Also, for $v\in G$ one has $q(v)^2\ge q(v)/(2n)$.  Therefore,
\begin{align*}
    \lambda_{\ver}
    \ge \sum_{v\in G} \frac{q(v)}{2n}\cdot \frac{2}{|\Sigma|}b_v 
    = \frac{1}{|\Sigma|n}\sum_{v\in G} q(v)b_v 
    \ge \frac{\kappa}{|\Sigma|n}.
\end{align*}

Applying the generalized birthday paradox (\Cref{lem:birthday}) with
\[
    \Omega_0=\Omega_0^{\ver},\qquad \alpha=2|\Sigma|,\qquad \beta=\kappa,\qquad D=|\Sigma|-1,
\]
we conclude that for $K=C\sqrt n$ with $C$ sufficiently large, the consistency test rejects with probability at least $1-\epsilon>0$.

\subparheading{Case~\ref{enu:determ-case}:} The label is unambiguous on most good vertices, meaning that
\[
    \sum_{v\in G} q(v)b_v<\kappa,
\]
Define $T \coloneqq \set{v\in G}{b_v\le 1/2}$, the set of unambiguous vertices. Then $G\setminus T=\set{v\in G}{b_v>1/2}$, and hence
\[
    q(G\setminus T)
    \le
    2\sum_{v\in G\setminus T}q(v)b_v
    \le
    2\sum_{v\in G}q(v)b_v
    <2\kappa.
\]
Because every vertex in $G$ has mass at least $1/(2n)$, this implies
$\abs{G\setminus T}<4\kappa n.$
Consequently,
\[
    \abs{V\setminus T}
    \le \abs{V\setminus G}+\abs{G\setminus T}
    < 3\delta n + 4\kappa n.
\]

Since we are in a \emph{no} instance, the plurality labeling $\iota^*$ violates at least an $\eta$-fraction of the
edges, i.e., at least $\eta\abs{E}=\eta dn/2$ edges.  Delete all edges incident to $V\setminus T$.
Since the graph has degree $d$, the number of deleted edges is at most
\[
    d\abs{V\setminus T}
    < 3\delta dn + 4\kappa dn.
\]

Now our choice of parameters $\delta = \eta/48$ and $\kappa=\eta/64$ guarantees that there remain at least $\eta dn/4$ violated edges $(u,v)\in E$ such that:
\begin{enumerate}[label=(\roman*),leftmargin=2em,itemsep=0em]
    \item both $u$ and $v$ lie in $T$;
    \item $(\iota^*(u),\iota^*(v))\notin R_{uv}$.
\end{enumerate}
Call these \emph{good violated edges}. Now define the set
\[
    \Omega_0^{\mathrm{edge}}
     \coloneqq \set{(v,\iota^*(v))}{ v \text{ is incident to some good violated edge}}
    \subseteq \Omega,
\]
and define a symmetric bad-pair relation $B_{\mathrm{edge}}\subseteq \Omega\times\Omega$ by including,
for every good violated edge $(u,v)$, the two ordered pairs
\[
    ((u,\iota^*(u)),(v,\iota^*(v)))
    \qquad \text{and} \qquad
    ((v,\iota^*(v)),(u,\iota^*(u))).
\]
Every atom in $\Omega_0^{\mathrm{edge}}$ lies in $T\subseteq G$, hence has mass at most $2/n$, and
every atom has degree at most $d$ inside this relation.

For a good violated edge, because $u,v\in T\subseteq G$, we have
\[
    q(u),q(v)\ge \frac{1}{2n}
    \quad\text{and}\quad
    r_u(\iota^*(u)),\ r_v(\iota^*(v))\ge \frac12.
\]
Consequently, it holds that
\[
    p(u,\iota^*(u))\,p(v,\iota^*(v))
    \ge \frac{1}{16n^2}.
\]
Summing over the two orientations of each good violated edge yields
\[
    \lambda_{\mathrm{edge}}
     \coloneqq  \sum_{(x,y)\in B_{\mathrm{edge}}} \mu(x)\mu(y)
    \ge \frac{\eta d}{32n}.
\] 
Applying the generalized birthday paradox (\Cref{lem:birthday}) with
\[
    \Omega_0=\Omega_0^{\mathrm{edge}},\qquad \alpha=2|\Sigma|,\qquad
    \beta=\eta d|\Sigma|/32, \qquad D=d,
\]
we conclude again that for $K=C\sqrt n$ with $C$ sufficiently large, the consistency test rejects with probability at least $1-\epsilon$. 

\subsection{A \texorpdfstring{$\StoqMA_{\log}(2)$}{} protocol with inverse-polynomial gap}
\label{subsec:stoqma-log-two-prover}

We now present the promised inverse-polynomial-gap version. In the constant-gap protocol, corresponding to its branch-local version presented in \Cref{protocol:branchLocal-NP-StoqMAkLog}, the $K=\Theta(\sqrt n)$ vertex-label factor branches are used only to ensure that the relevant vertex- and edge-detection events have \emph{constant} total branch weight. With two vertex-label factor branches, the same events have branch weight $\Theta(1/n)$, which is sufficient to obtain an \emph{inverse-polynomial} promise gap.

\NPinStoqMAlogTwo*

\begin{proof}
In view of~\Cref{thm:sym-to-stoq-k}, it suffices to prove
\begin{equation} 
\NP \subseteq \SymStoqMA_{O(\log n)}\ssbra*{2,\,1-\Theta\rbra*{\frac{1}{n}},\,\Theta\rbra*{\frac{1}{n}}}.
\end{equation}

Consider a $\GapCG$ instance as in~\Cref{thm:dinur-gapcg-form}, with $n=|V|$ vertices. Merlins send an $O(\log n)$-qubit state on the computational basis $V\times\Sigma$.  On a \emph{yes} instance, for the satisfying labeling $\iota:V\to\Sigma$, both Merlins send
\[
    \ket{\phi_\rmL}= \frac1{\sqrt n}\sum_{v\in V}\ket{v,\iota(v)}.
\]

\paragraph{The simple branch-local protocol and its stoquastic implementation.}
Similar to the analysis in the previous subsection, it suffices to first consider the simple branch-local protocol $W$, analogous to \Cref{protocol:branchLocal-NP-StoqMAkLog}, as presented in \Cref{protocol:simple-branchLocal-NP-StoqMAkLog}. 

\LinesNotNumbered
\begin{algorithm}[ht!]
    \UseProtocolCounter
    \caption{Simple branch-local protocol for $\NP$ on a fixed $2$-vertex-label branch.}
	\label{protocol:simple-branchLocal-NP-StoqMAkLog}
    \SetEndCharOfAlgoLine{.}
    \setlength{\parskip}{5pt}
    \SetKwFor{While}{}{:}{}
    \SetKwInOut{Parameter}{Parameters}

    \textbf{1.} Consider a $2$-vertex-label branch of $\ket{\Psi} = \ket{\psi} \otimes \ket{\psi}$, denoted by 
    \[ ((u,a), (v,b)) \in \rbra*{V\times\Sigma}^2. \] 
        
    \textbf{2.} \While{\textnormal{Apply either the uniformity or the consistency test, each with probability $1/2$}}{

    \medskip
    
    \textbf{2.1} \textbf{Uniformity test:} \emph{Reject} if $u=v$. 
    
    \medskip
    \textbf{2.2} \textbf{Consistency test:} \emph{Reject} if either of the following occurs:
        \begin{enumerate}[label=(\roman*),itemsep=-2pt]
            \item $u=v$ and $a\neq b$;
            \item $(u,v)\in E$ and $(a,b) \notin R_{uv}$. 
        \end{enumerate}
    }
    \textbf{3.} Accept if and only if the chosen test accepts.
\end{algorithm}

To implement \Cref{protocol:simple-branchLocal-NP-StoqMAkLog} as a stoquastic verification circuit, we write the two acceptance conditions of the uniformity and consistency tests, respectively, as deterministic predicates $A_\unif$ and $A_\cons$, and implement them by classical reversible circuits $\Gamma_\unif$ and $\Gamma_\cons$, respectively, using the standard reversible-computation convention. 

Finally, we implement the convex combination $(1/2,1/2)$ of the two tests by combining these circuits using an additional $\ket{+}$ ancillary qubit, obtaining a reversible circuit $\Gamma_\branch$ for the uniform mixture of $\Gamma_\unif$ and $\Gamma_\cons$. The resulting stoquastic verification circuit is then obtained by applying the branch-overlap test (\Cref{lemma:branch-overlap-test}) to the circuit pair $(\Gamma_\branch,I)$ with the input state $\ket{+}\ket{\Psi}\ket{\bar{0}}$, where $\ket{\Psi}$ coincides with $\ket{\phi_\rmL}\otimes\ket{\phi_\rmL}$ for \emph{yes} instances. 

\paragraph{Analysis of the branch-local verifier $W$.}
We analyze the branch-local protocol $W$ on two symmetric witnesses. From the distributional testing perspective, $(u,a), (v,b)$ are drawn from $p^{\otimes 2}$, where $p$ is a distribution on $V\times\Sigma$. 
We will use the same notation as in the proof of~\Cref{thm:power-SymStoqMAk-np}.

\parheading{Completeness of $W$.} For every \emph{yes} instance, the consistency test sees no violation, the only rejection can occur in the uniformity test, when $u=v$, whose total weight is easily seen to be $\frac{1}{2n}$, thus 
\begin{equation}
\label{eq:weak-two-sample-gap-completeness}
    \Pr[W \text{ rejects}] = \frac{1}{2n}.
\end{equation}

\parheading{Soundness of $W$.} We claim that, for every \emph{no} instance, and for a constant $\gamma>0$,
\begin{equation}
\label{eq:weak-two-sample-gap}
    \Pr[W \text{ rejects}]\ge \frac{1}{2n}+\frac{\gamma}{n}
\end{equation}

Indeed, if the vertex marginal $q$ is $\delta$-far from uniform in total variation distance, the uniformity test rejects with probability
\[
    \frac12 \sum_v q(v)^2
    = \frac{1}{2n}+\frac12\|q-\mu_V\|_2^2
    \ge \frac{1}{2n}+\frac{2\delta^2}{n},
\]
where the last step is by Cauchy-Schwarz that $ 4\TV(q,\mu_V)^2=\|q-\mu_V\|_1^2\le n\|q-\mu_V\|_2^2$. Now assume $q$ is close to uniform, we split the remaining soundness analysis into two cases.

\subparheading{Ambiguous case.} If $q$ is close to uniform but there is substantial ambiguity.
Let  $B_{\ver}\subseteq \Omega\times\Omega$ be the symmetric bad-pair relation consisting of all
ordered pairs
\[
    ((v,a),(v,a')) \quad \text{with } v\in G \text{ and } a\neq a'.
\]

Then the bad-pair weight
\[\lambda_{\ver} \coloneqq \sum_{((u,a),(v,b))\in B_{\ver}} p(u,a)p(v,b)\] 
based on the calculation in Case~\ref{enu:ambig-case} from~\Cref{thm:power-SymStoqMAk-np} is such that 
\[\lambda_{\ver}\ge \frac{\kappa}{|\Sigma|n},\]
contributing an additional $\lambda_{\ver}/2$ rejection weight within the consistency test, to the weight $1/(2n)$ from the uniformity test.  

\subparheading{Unambiguous case.} Otherwise, $q$ is close to uniform and the labels are unambiguous on a good fraction of vertices. Let the symmetric bad-pair relation $B_{\mathrm{edge}}\subseteq \Omega\times\Omega$ include the two ordered pairs
\[
    ((u,a),(v,b))
    \;\;\text{and}\;\;
    ((v,b),(u,a)), \quad \text{where}\;\; (a,b)\not\in R_{uv}.
\]

The bad-pair weight with respect to $B_{\mathrm{edge}}$,
\[\lambda_{\mathrm{edge}} \coloneqq \sum_{((u,a),(v,b))\in B_{\mathrm{edge}}} p(u,a)p(v,b),\] 
by the calculation in Case~\ref{enu:determ-case} from~\Cref{thm:power-SymStoqMAk-np} gives 
\[\lambda_{\mathrm{edge}}\ge \frac{\eta d}{32n},\]
contributing an additional $\lambda_{\mathrm{edge}}/2$ rejection weight within the consistency test, to the rejection weight of $1/(2n)$ from the uniformity test.  

Then~\Cref{eq:weak-two-sample-gap} holds by taking 
\[\gamma = \min \cbra*{2\delta^2,\frac{\kappa}{2|\Sigma|},\frac{\eta d}{64}}. \qedhere\]
\end{proof}

%%%%%%%%%%%%%%%%%%%%%%%%%%%%%%%%%%%%%%%%%%%%%%%%%%%%%%%%%%%
%%%%%%%%%%%%%%%%%%%%%%%%%%%%%%%%%%%%%%%%%%%%%%%%%%%%%%%%%%%
%%%%%%%%%%%%%%%%%%%%%%%%%%%%%%%%%%%%%%%%%%%%%%%%%%%%%%%%%%%

\section{Upper bounds for \StoqMAk{} with logarithmic-size proofs}
\label{sec:StoqMAk-log-upper-bound}

In this section, we provide nearly tight upper bounds for \StoqMAk{} with logarithmic-size proofs. We first prove that $\StoqMA_{\log}$ collapses to \BPP{}, which provides positive evidence for the long-standing open problem of whether $\StoqMA=\MA$, since $\MA_{\log}=\BPP$ (e.g.,~\cite[Theorem 3]{Regev06}), and aligns with the equivalence $\QMA_{\log}=\BQP$~\cite[Theorem 3.10]{MW05}:  

\begin{theorem}[Randomized singly exponential-time upper bound for \StoqMA{}]
    \label{thm:StoqMA-BPTIME-upper-bound}
    For all efficiently computable functions $c(n)$ and $\Delta(n)$ with $c(n) \in (1/2, 1]$, $\Delta(n) \in (0,1/2]$, and $c(n)-\Delta(n)\geq 1/2$, and for every efficiently computable positive integer-valued function $\ell(n) \leq \poly(n)$, the following inclusion holds:
    \[ \StoqMA_\ell\ssbra{1,c,\Delta} \subseteq \BPTIME\sbra*{\widetilde{O}\rbra[\bigg]{\frac{2^{3\ell}}{\Delta^2}} \cdot \poly(|V|)}. \]
    Here, $|V|$ denotes the description size of the underlying stoquastic verification circuit $V$. 

    \noindent Furthermore, up to lower-order factors, this upper bound is \emph{tight} with respect to the proof length $\ell$ under randomized ETH. 
\end{theorem}

By taking $\ell(n) = O(\log{n})$ in \Cref{thm:StoqMA-BPTIME-upper-bound}, we immediately obtain:

\begin{corollary}
    $\StoqMA_{\log} = \BPP$. 
\end{corollary}

For multiple-prover scenarios, although the best known upper bound for $\QMAtwo_{\log}$ is the same as that for general \QMAtwo{}, namely \NEXP{}, we establish an upper bound $\StoqMAtwo_{\log} \subseteq \MA$ that is better than the bound due to $\StoqMAtwo\subseteq \StoqMA$~\cite{GJ26}: 

\begin{restatable}[\StoqMAk{} with logarithmic-size proofs is in \MA{}]{theorem}{StoqMAkLogInMA}
    \label{thm:StoqMAk_log-in-MA}
    For all efficiently computable functions $c(n)$ and $\Delta(n)$ with $1/2 < c(n) \leq 1$, $\Delta(n)\geq 1/\poly(n)$, and $c(n)-\Delta(n)\geq 1/2$, and for every efficiently computable positive integer-valued function $\ell(n)$ and $k(n)$ such that $\ell(n) \leq O(\log{n})$ and $k(n) \leq \poly(n)$, the following inclusion holds:
    \[ \StoqMA_\ell\ssbra{k,c,\Delta} \subseteq \MA. \]    
\end{restatable}

In the remainder of this section, we first investigate the single-prover case in \Cref{subsec:log-size-witenss-single-prover}, and then explore the multi-prover case in \Cref{subsec:log-size-witenss-multi-prover}. Throughout this section, for convenience, we slightly abuse notation and also use $\Gamma$ to denote the exponential-size matrix induced by the reversible circuit in an \RCD{} or \SepRCD{} instance.

\subsection{\texorpdfstring{$\StoqMA_{\log} \subseteq \BPP$}{} and general rETH-optimal \texorpdfstring{\BPTIME{}}{BPTIME} upper bound}
\label{subsec:log-size-witenss-single-prover}

The optimality, with respect to the proof length $\ell$, of the \BPTIME{} upper bound underlying \Cref{thm:StoqMA-BPTIME-upper-bound} follows from \Cref{lemma:rETH-optimal}, and the proof is deferred to \Cref{subsec:omitted-rETH-opt-BPTIME-bound}.

\begin{restatable}[Proof-length optimality under randomized ETH]{lemma}{rETHoptimal}
    \label{lemma:rETH-optimal}
    Let $\ell(n)$ be an efficiently computable positive integer-valued function satisfying $\ell(n) \leq \poly(n)$. Assume that 
    \[\StoqMA_\ell\ssbra{1,1,1/2} \subseteq \BPTIME\rbra*{ 2^{o(\ell)} \cdot \poly(|V|) },\] 
    where $|V|$ denotes the description size of the underlying stoquastic verification circuit $V$. Then this randomized upper bound refutes randomized ETH.  
\end{restatable}

\subsubsection{Proof of \texorpdfstring{\Cref{thm:StoqMA-BPTIME-upper-bound}}{Theorem 6.1}}
To establish the \BPTIME{} upper bound for \StoqMA{}, the key observation is that, once a computational-basis witness state $\ket{a}$ is fixed, with $a\in\binset^\ell$, the only randomness remaining in the compressed matrix $M_{\Gamma}$ comes from the $\ket{+}$ ancillary qubits. Consequently, one random simulation of $\Gamma$ produces one sample from a column of $M_\Gamma$. Estimating all $2^\ell$ columns takes running time \emph{singly exponential} in $\ell$, which is polynomial in $n$ when $\ell(n) = O(\log{n})$. We now proceed with the detailed randomized algorithm and the corresponding analysis: 

\begin{proof}[Proof of \Cref{thm:StoqMA-BPTIME-upper-bound}]
    Following \Cref{lemma:SepRCD-StoqMAk-complete}, we consider a \RCD{} instance $(\Gamma,I)$ such that 
    \[ \innerprod{R_0(\psi)}{R_1(\psi)} = \bra{\psi}\bra{\bar{0}}\bra{\bar{+}}\Gamma\ket{\psi}\ket{\bar{0}}\ket{\bar{+}} \] 
    for every non-negative state $\ket{\psi}$. Therefore, it suffices to distinguish 
    \[ \lambda_{\max}\rbra[\Big]{\frac{I+M_\Gamma}{2}} \geq c(n) \quad\text{from}\quad \lambda_{\max}\rbra[\Big]{\frac{I+M_\Gamma}{2}} \leq s(n) \coloneqq c(n)-\Delta(n).\] 
    Here, $M_\Gamma \coloneqq \bra{0}^{\otimes m_0} \bra{+}^{\otimes r} \Gamma \ket{0}^{\otimes m_0} \ket{+}^{\otimes r}$ is a matrix of size $2^\ell \times 2^\ell$, which is polynomial in $n$ when $\ell(n) \leq O(\log{n})$. The rETH-based optimality statement follows from \Cref{lemma:rETH-optimal}.

\LinesNumbered
\begin{algorithm}[ht!]
    \SetAlgorithmName{Algorithm}{algorithm}{List of Algorithms}
    \caption{Column-sampling matrix reconstruction of $M_\Gamma$.}
    \label{algo:column-sampling-matrix-reconstruction}
    \SetEndCharOfAlgoLine{.}
    \setlength{\parskip}{5pt}
    \SetKwFor{While}{}{:}{}
    \SetKwFor{For}{For}{:}{}
    \SetKwIF{If}{ElseIf}{Else}{If}{:}{elif}{Else:}{}%
    \SetKwComment{Comment}{// }{}
    \SetKwInOut{Input}{Input}
    \SetKwInOut{Output}{Output}
    \SetKwInOut{Parameter}{Parameters}

    \Input{A reversible circuit $\Gamma$ and an entrywise accuracy parameter $\eta>0$.}
    \Output{An explicit matrix $\widetilde{M}_\Gamma$ approximating $M_\Gamma$.}
    \Parameter{The witness length $\ell$, the number of $\ket{0}$ ancillary qubits $m_0$, and the number of $\ket{+}$ ancillary qubits $r$.}

    Set $L\coloneqq \binset^{\ell}$ and $U\coloneqq \binset^r$\;
    Set $T\coloneqq \ceil*{C\log(20\abs{L}^2)/\eta^2}$ for a sufficiently large absolute constant $C$\;
    Initialize an $\abs{L}\times \abs{L}$ zero matrix $\widetilde M_\Gamma$\;
    \For{each column $a\in L$}{
        Initialize counters $N_{b,a}\coloneqq 0$ for all $b\in L$\;
        \For{each sample $j\in\cbra{1,\dots,T}$}{
            Sample $u \sim \Uniform(U)$\;
            Compute $\Gamma(a,0^{m_0},u)=(b,z,u')\in L\times \binset^{m_0}\times U$\;
            \If{$z=0^{m_0}$}{
                Update $N_{b,a}\leftarrow N_{b,a}+1$\;
            }
        }
        \For{each row $b\in L$}{
            Set $(\widetilde{M}_\Gamma)_{b,a}\coloneqq N_{b,a}/T$\;
        }
    }
    Return $\widetilde{M}_\Gamma$\;
\end{algorithm}

    To this end, we provide a column-sampling matrix reconstruction procedure for $M_\Gamma$, as presented in \algoref{algo:column-sampling-matrix-reconstruction}. By invoking \algoref{algo:column-sampling-matrix-reconstruction} with the circuit $\Gamma$ and the accuracy parameter $\eta = \Delta/(8|L|)$, we obtain $\widetilde{M}_{\Gamma}$ by estimating $M_{\Gamma}$ column by column. For every fixed pair $(a,b) \in L\times L$, define the Bernoulli random variable 
    \[ Y_{a,b}(u) \coloneqq \bbI\sbra*{ \Gamma(a,0^{m_0},u) \in \cbra{b} \times \cbra{0^{m_0}} \times U }, \quad\text{where } u \in U. \]
    It is easy to verify that the expectation of $Y_{a,b}$ corresponds to the matrix entry $\bra{b}M_{\Gamma}\ket{a}$:
    \begin{equation}
         \mathop{\bbE}_{u\sim\Uniform(U)}[Y_{a,b}(u)] = 2^{-r} \abs*{\set{u\in U}{\Gamma(a,0^{m_0},u) \in \cbra{b}\times\cbra{0^{m_0}}\times U}} =  \bra{b}M_{\Gamma}\ket{a}.
    \end{equation}

    Since the number of samples used by \algoref{algo:column-sampling-matrix-reconstruction} is $T = \ceil*{C\log(20\abs{L}^2)/\eta^2}$, for every fixed $(a,b)$, the Hoeffding bound (see, e.g.,~\cite[Theorem 4.12]{MU17}) gives
    \begin{equation}
        \label{eq:StoqMA-BPTIMEbound-column}
        \Pr\sbra*{\abs*{\bra{b}\widetilde{M}_{\Gamma}\ket{a} - \bra{b}M_{\Gamma}\ket{a}} > \eta}
        \leq 2e^{-2T\eta^2} 
        \leq \frac{1}{10\abs{L}^2}.
    \end{equation}
    By taking a union bound over all $\abs{L}^2$ entries, each of which satisfies \Cref{eq:StoqMA-BPTIMEbound-column}, we obtain
    \begin{subequations}
    \label{eq:StoqMA-BPTIMEbound-union}
    \begin{align}
        \Pr\sbra*{\max_{(a,b)\in L\times L} \abs*{\bra{b}\widetilde{M}_{\Gamma}\ket{a} - \bra{b}M_{\Gamma}\ket{a}}\leq \eta}
        &\geq 1-  \sum_{(a,b)\in L\times L} \Pr\sbra*{\abs*{\bra{b}\widetilde{M}_{\Gamma}\ket{a} - \bra{b}M_{\Gamma}\ket{a}} > \eta}\\ 
        &\geq 1 - \abs{L}^2 \cdot \frac{1}{10\abs{L}^2} = \frac{9}{10}.
    \end{align}
    \end{subequations}

    By Weyl's perturbation theorem (see, e.g.,~\cite[Corollary III.2.6]{Bhatia96}), on the good event underlying \Cref{eq:StoqMA-BPTIMEbound-union}, we obtain:
    \begin{subequations}
    \label{eq:StoqMA-BPTIMEbound-error}
    \begin{align}
        \abs*{\lambda_{\max}\rbra[\bigg]{\frac{I+\widetilde{M}_{\Gamma}}{2}} - \lambda_{\max}\rbra*{\frac{I+M_{\Gamma}}{2}}} 
        &= \frac{1}{2} \cdot \abs*{\lambda_{\max}\rbra*{\widetilde{M}_{\Gamma}} - \lambda_{\max}\rbra*{M_{\Gamma}}}\\ 
        &\leq \frac{1}{2} \cdot \norm*{\widetilde{M}_{\Gamma} - M_{\Gamma}}\\
        &\leq \frac{1}{2} \cdot \norm*{\widetilde{M}_{\Gamma} - M_{\Gamma}}_2\\
        &= \frac{1}{2} \cdot \sqrt{\sum_{(a,b)\in L\times L} \abs*{\bra{b}\widetilde{M}_{\Gamma}\ket{a} - \bra{b}M_{\Gamma}\ket{a}}^2}\\
        &\leq \frac{1}{2} \cdot \abs{L} \eta = \frac{\Delta}{16}.
    \end{align}
    \end{subequations}
    
    Consequently, conditioned on the good event underlying \Cref{eq:StoqMA-BPTIMEbound-union}, comparing $\rbra[\big]{1+\lambda_{\max}\rbra[\big]{\widetilde{M}_{\Gamma}}}/2$ with the midpoint $(c+s)/2=c-\Delta/2$ correctly decides the promise problem. 

    \parheading{Time complexity.}
    It remains to bound the running time. \algoref{algo:column-sampling-matrix-reconstruction} estimates $\abs{L}$ columns, each using $T$ samples, and every sample corresponds to a polynomial-time classical evaluation of the reversible circuit $\Gamma$. Therefore, using the facts that $\abs{L}=2^{\ell}$ and that estimating the maximum eigenvalue of an explicit $\abs{L} \times \abs{L}$ symmetric matrix takes running time $\abs{L}^3$ (see, e.g.,~\cite[Section 5.3.3]{Demmel97}), the time complexity of \algoref{algo:column-sampling-matrix-reconstruction} is
    \begin{align*}
        \max\cbra*{ \abs{L}\cdot T\cdot \poly(\abs{V}), O\rbra*{\abs{L}^3} }
        &= \max\cbra*{ \widetilde{O}\rbra*{\abs{L}\cdot \frac{1}{\eta^2}}\cdot \poly(\abs{V}), O\rbra*{\abs{L}^3} }\\
        &= \widetilde{O}\rbra*{\frac{\abs{L}^3}{\Delta^2}}\cdot \poly(\abs{V})
        =\widetilde{O}\rbra*{\frac{2^{3\ell}}{\Delta^2}}\cdot \poly(\abs{V}). \qedhere
    \end{align*}
\end{proof}

\subsection{\texorpdfstring{$\StoqMAk_{\log} \subseteq \MA$}{}}
\label{subsec:log-size-witenss-multi-prover}

We start with the two-prover case as a warm-up, and then move on to the multiple-prover case with at most polynomially many provers.

\subsubsection{A warm-up: \texorpdfstring{$\StoqMAtwo_{\log} \subseteq \MA$}{}}
Since the full witness space has dimension $2^{2\ell}=\poly(n)$, after flattening the pair of witness registers into one alphabet of size $2^{2\ell}$, Arthur can invoke \algoref{algo:column-sampling-matrix-reconstruction} exactly as in the one-prover scenario. Nevertheless, the stoquastic separable value does not coincide with the largest eigenvalue in general, the prover thus must provide a classical dyadic description of a \emph{rounded} separable witness, and Arthur evaluates that candidate product state against the reconstructed matrix. We then continue with the actual proof:

\begin{theorem}[\StoqMAtwo{} with logarithmic-size proofs is in \MA{}]
    \label{thm:StoqMA2_log-in-MA}
    For all efficiently computable functions $c(n)$ and $\Delta(n)$ with $1/2 < c(n) \leq 1$, $\Delta(n)\geq 1/\poly(n)$, and $c(n)-\Delta(n)\geq 1/2$, and for every efficiently computable positive integer-valued function $\ell(n)$ such that $\ell(n) = O(\log{n})$, the following inclusion holds:
    \[ \StoqMA_\ell\ssbra{2,c,\Delta} \subseteq \MA. \]    
\end{theorem}

\begin{proof}
    Using \Cref{lemma:SepRCD-StoqMAk-complete}, we consider a \SepRCD{} instance $(\Gamma,I)$ such that 
    \[\innerprod{R_0(\Psi)}{R_1(\Psi)} = \bra{\Psi}\bra{\bar{0}}\bra{\bar{+}} \Gamma\ket{\Psi} \ket{\bar{0}}\ket{\bar{+}}\] 
    for every non-negative state $\ket{\Psi} \coloneqq \ket{\psi_1}\otimes\ket{\psi_2}$. Write $A = B = \binset^\ell$ and $U\coloneqq \binset^r$. We define the matrix $M_\Gamma$, indexed by $(a,b),(a',b') \in A\times B$, by
    \begin{equation}
    \label{eq:two-prover-compressed-matrix}
        \bra{a',b'} M_\Gamma \ket{a,b}
        =2^{-r}\abs*{\set{u\in U}{\Gamma(a,b,0^{m_0},u)\in \cbra{a'}\times \cbra{b'}\times \cbra{0^{m_0}}\times U}}.
    \end{equation}
    Let $Q_\Gamma \coloneqq (I+M_\Gamma)/2$. Therefore, it suffices to distinguish the stoquastic separable value
    \[ \hsepPlus{2^\ell}{2^\ell}\rbra*{Q_\Gamma} \geq c(n) \quad\text{from}\quad \hsepPlus{2^\ell}{2^\ell}\rbra*{Q_\Gamma} \leq s(n) \coloneqq c(n) - \Delta(n).\] 
    Here, the entry-wise non-negative matrix $Q_\Gamma$ is of size $2^{2\ell} \times 2^{2\ell}$, which is polynomial in $n$ when $\ell(n) = O(\log{n})$, and satisfies that $0 \preceq Q_\Gamma \preceq I$. 

    \paragraph{Rounding non-negative product-state witnesses to dyadic numbers.} Fix $\delta \coloneqq \Delta/10$, and choose an integer $t$ such that $2^{-t} \leq \delta/2^{\frac{\ell}{2}}$. For every coordinate of $\ket{\psi_1}$ and $\ket{\psi_2}$, set non-negative integer-valued vectors $\bfm_1$ and $\bfm_2$ by 
    \[ m_1(a) \coloneqq \floor*{ 2^t \innerprod{a}{\psi_1} + \frac{1}{2} } \text{ for } a\in A \quad\text{and}\quad m_2(b) \coloneqq \floor*{ 2^t \innerprod{b}{\psi_2} + \frac{1}{2} } \text{ for } b\in B. \]
    
    Arthur interprets the classical witness $(\bfm_1,\bfm_2)$ as the normalized vectors
    \[ \ket{\psi'_1} \coloneqq \frac{1}{\norm{\bfm_1}_2} \sum_{a\in A} m_1(a) \ket{a} \quad\text{and}\quad \ket{\psi'_2} \coloneqq \frac{1}{\norm{\bfm_2}_2} \sum_{b\in B} m_2(b) \ket{b}. \]
    Consequently, the witness $(\bfm_1,\bfm_2)$ is a pair of integer lists describing a rounded non-negative product-state witness, all of whose amplitudes are dyadic before normalization. With the chosen parameters, define $\ket{\psi_1^\star}\coloneqq 2^{-t} \sum_{a\in A} m_1(a) \ket{a}$ and $\ket{\psi_2^\star}\coloneqq 2^{-t} \sum_{b\in B} m_2(b) \ket{b}$. Then a direct calculation shows the following rounding-error bound:
    \begin{subequations}
    \label{eq:two-prover-product-state-rounding}
    \begin{align}
        &\norm*{ \ket{\psi'_1}\totimes\ket{\psi'_2} - \ket{\psi_1}\totimes\ket{\psi_2} }_2\\
        =~& \norm*{ \rbra*{\ket{\psi'_1}\totimes\ket{\psi'_2} - \ket{\psi'_1}\totimes\ket{\psi_2}} + \rbra*{\ket{\psi'_1}\totimes\ket{\psi_2} - \ket{\psi_1}\totimes\ket{\psi_2}}}_2\\
        \leq~& \norm*{\ket{\psi'_2}-\ket{\psi_2}}_2 + \norm*{\ket{\psi'_1}-\ket{\psi_1}}_2\\
        \leq~& \sum_{j\in\cbra{1,2}} \rbra[\big]{ \norm*{\ket{\psi'_j}-\ket{\psi^\star_j}}_2 + \norm*{\ket{\psi^\star_j}-\ket{\psi_j}}_2 } \\
        =~& \sum_{j\in\cbra{1,2}} \rbra[\big]{ \abs*{\norm*{\ket{\psi^\star_j}}_2 - 1} + \norm*{\ket{\psi^\star_j}-\ket{\psi_j}}_2 } \\ 
        \leq~& 2 \sum_{j\in\cbra{1,2}} \norm*{\ket{\psi^\star_j} - \ket{\psi_j}}_2 \\ 
        \leq~& 4 \cdot 2^{\ell} \cdot 2^{-\ell/2} \cdot 2^{-t-1}\\
        \leq~& 2\delta.
    \end{align}
    \end{subequations}
    Here, the third and fourth lines follow from the triangle inequality, the fifth line uses the identity 
    \[ \norm*{\ket{\psi^\star_j}-\ket{\psi'_j}}_2 = \norm*{\ket{\psi^\star_j}-\ket{\psi^\star_j}/\norm*{\ket{\psi^\star_j}}_2}_2 = \abs*{1 - 1/\norm*{\ket{\psi^\star_j}}_2} \cdot \norm*{\ket{\psi^\star_j}}_2 = \abs*{\norm*{\ket{\psi^\star_j}}_2 - 1}, \]
    and the sixth line uses the reverse triangle inequality and the fact that $\norm*{\ket{\psi_j}}_2=1$.

    As a consequence, \Cref{eq:two-prover-product-state-rounding} yields that
    \begin{subequations}
    \label{eq:two-prover-rounding-perturbation}
    \begin{align}
        &\abs*{ \bra{\psi'_1}\bra{\psi'_2} Q_\Gamma \ket{\psi'_1}\ket{\psi'_2} - \bra{\psi_1}\bra{\psi_2} Q_\Gamma \ket{\psi_1}\ket{\psi_2} } \\
        \leq~& \td\rbra*{\ketbra{\psi'_1}{\psi'_1}\totimes\ketbra{\psi'_2}{\psi'_2}, \ketbra{\psi_1}{\psi_1}\totimes\ketbra{\psi_2}{\psi_2}}\\
        =~& \sqrt{1-\innerprod{\psi'_1}{\psi_1}^2\innerprod{\psi'_2}{\psi_2}^2}\\
        \leq~& \norm*{ \ket{\psi'_1}\totimes\ket{\psi'_2} - \ket{\psi_1}\totimes\ket{\psi_2} }_2\\
        \leq~& 2\delta = \frac{\Delta}{5}.
    \end{align}
    \end{subequations}
    Here, the second line uses $0 \preceq Q_\Gamma \preceq I$ and the variational characterization of trace distance (see, e.g.,~\cite[Lemma 9.1.7]{Wilde13}), the third line follows from \Cref{eq:pure-pure-Fuchs-vanGraaf}, and the fourth line uses the fact that $\ket{\psi_j}$ and $\ket{\psi'_j}$ are non-negative states for each $j\in\cbra{1,2}$.
    
    \paragraph{Estimating the stoquastic separable value with the prover's help.} Upon receiving the classical witness $(\bfm_1,\bfm_2)$, the verifier invokes \algoref{algo:column-sampling-matrix-reconstruction} with the circuit $\Gamma$ and the accuracy parameter $\eta\coloneqq \Delta/(2 \abs{A\times B})$, and obtains the matrix $\widetilde{M}_\Gamma$ with probability at least $9/10$ satisfying the following bound (i.e., the good event) via the same analysis as in \Cref{eq:StoqMA-BPTIMEbound-error}:
    \begin{equation}
        \label{eq:two-prover-matrix-reconstruction}
        \norm*{\widetilde{Q}_\Gamma - Q_\Gamma} = \frac{1}{2} \norm*{\widetilde{M}_\Gamma - M_\Gamma} \leq \frac{\abs{A\times B} \eta}{2} = \frac{\Delta}{4}, \quad\text{where } \widetilde{Q}_\Gamma \coloneqq \frac{I+\widetilde{M}_\Gamma}{2}. 
    \end{equation}

    Arthur then computes the quantity $\bra{\psi'_1}\bra{\psi'_2} \widetilde{Q}_\Gamma \ket{\psi'_1}\ket{\psi'_2}$ and accepts if and only if 
    \[ \bra{\psi'_1}\bra{\psi'_2} \widetilde{Q}_\Gamma \ket{\psi'_1}\ket{\psi'_2} \geq \frac{c+s}{2} = c-\frac{\Delta}{2}. \]
    Consequently, we complete the analysis using the following bounds:
    \begin{itemize}
        \item For \emph{yes} instances, on the good event, we obtain:
        \begin{align*}
            &\bra{\psi'_1}\bra{\psi'_2} \widetilde{Q}_\Gamma \ket{\psi'_1}\ket{\psi'_2} \\
            \geq~& \bra{\psi_1}\bra{\psi_2} Q_\Gamma \ket{\psi_1}\ket{\psi_2} - \abs*{ \bra{\psi'_1}\bra{\psi'_2} Q_\Gamma \ket{\psi'_1}\ket{\psi'_2} - \bra{\psi_1}\bra{\psi_2} Q_\Gamma \ket{\psi_1}\ket{\psi_2} } \\
            & \qquad - \abs*{ \bra{\psi'_1}\bra{\psi'_2} \rbra[\big]{ \widetilde{Q}_\Gamma - Q_\Gamma } \ket{\psi'_1}\ket{\psi'_2} } \\
            \geq~& \bra{\psi_1}\bra{\psi_2} Q_\Gamma \ket{\psi_1}\ket{\psi_2} - \abs*{ \bra{\psi'_1}\bra{\psi'_2} Q_\Gamma \ket{\psi'_1}\ket{\psi'_2} - \bra{\psi_1}\bra{\psi_2} Q_\Gamma \ket{\psi_1}\ket{\psi_2} } - \norm*{\widetilde{Q}_\Gamma - Q_\Gamma} \\
            \geq~& c - \frac{\Delta}{5} - \frac{\Delta}{4} > c - \frac{\Delta}{2}.
        \end{align*}
        Here, the first inequality uses the reverse triangle inequality, the second inequality follows from the definition of the operator norm, and last line follows immediately from \Cref{eq:two-prover-rounding-perturbation,eq:two-prover-matrix-reconstruction}. 
        \item For \emph{no} instances, the soundness condition guarantees that $\bra{\psi_1}\bra{\psi_2} Q_\Gamma \ket{\psi_1}\ket{\psi_2} \leq s$ for every non-negative product-state witness $\ket{\psi_1}\otimes\ket{\psi_2}$, and on the good event, it holds that:
        \begin{align*}
            \bra{\psi'_1}\bra{\psi'_2} \widetilde{Q}_\Gamma \ket{\psi'_1}\ket{\psi'_2}
            &\leq \bra{\psi'_1}\bra{\psi'_2} Q_\Gamma \ket{\psi'_1}\ket{\psi'_2} + \abs*{\bra{\psi'_1}\bra{\psi'_2} \rbra*{ \widetilde{Q}_\Gamma - Q_\Gamma } \ket{\psi'_1}\ket{\psi'_2}}\\
            &\leq \bra{\psi'_1}\bra{\psi'_2} Q_\Gamma \ket{\psi'_1}\ket{\psi'_2} + \norm*{\widetilde{Q}_\Gamma - Q_\Gamma}\\
            &\leq s + \frac{\Delta}{4} < c - \frac{\Delta}{2}. 
        \end{align*}
        Here, the first line uses the triangle inequality, the second line follows from the definition of the operator norm, and the last line follows immediately from \Cref{eq:two-prover-matrix-reconstruction}.  \qedhere
    \end{itemize}
\end{proof}

\subsubsection{Proof of \texorpdfstring{\Cref{thm:StoqMAk_log-in-MA}}{Theorem 6.3}}

The general $k$-prover proof follows the same dyadic rounding procedure used in the two-prover case (\Cref{thm:StoqMA2_log-in-MA}), but avoids reconstructing the full matrix $M_\Gamma$, whose dimension may already be \emph{exponential}. Instead, Arthur directly estimates a single overlap using the \emph{dual-access viewpoint} from distribution testing~\cite[Theorem 6]{CR14}, following an approach similar to that used to establish $\mathsf{eStoqMA} \subseteq \MA$~\cite[Section 3]{Liu21}.

\StoqMAkLogInMA*

\begin{proof}
    Following \Cref{lemma:SepRCD-StoqMAk-complete}, we consider a $\SepRCD_k$ instance $(\Gamma,I)$ such that 
    \[ \innerprod{R_0(\Psi)}{R_1(\Psi)} = \bra{\Psi}\bra{\bar{0}}\bra{\bar{+}}\Gamma\ket{\Psi}\ket{\bar{0}}\ket{\bar{+}} \] 
    for every non-negative state $\ket{\Psi} \coloneqq \ket{\psi_1}\otimes\cdots\otimes\ket{\psi_k}$. 
    Write $L \coloneqq \binset^\ell$ and $U \coloneqq \binset^r$. We can define the matrix $M_\Gamma$, indexed by $(a_1,\dots,a_k), (b_1,\dots,b_k) \in L^k$, which directly generalizes \Cref{eq:two-prover-compressed-matrix}.  
    Let $Q_\Gamma \coloneqq (I+M_\Gamma)/2$. It thus suffices to distinguish the stoquastic value 
    \[ \hsepPlusM{k}{2^\ell}(Q_\Gamma) \geq c(n) \quad\text{from}\quad \hsepPlusM{k}{2^\ell}(Q_\Gamma) \leq s(n)\coloneqq c(n) - \Delta(n). \]
    Here, the entry-wise non-negative matrix $Q_\Gamma$ is of size $2^{k\ell}\times 2^{k\ell}$, which is no longer polynomial in $n$ when $\ell(n) = O(\log{n})$, and satisfies $0 \preceq Q_\Gamma \preceq I$. 

    \paragraph{Rounding non-negative product-state witnesses to dyadic numbers.} We now extend the dyadic rounding procedure in the proof of \Cref{thm:StoqMA2_log-in-MA} for two-prover cases to $k$-prover cases. 
    Fix $\delta \coloneqq \Delta/(5k)$, and choose an integer $t$ such that $2^{-t} \leq \delta/2^{\ell/2}$. For each proof $i\in[k]$, set non-negative integer-valued vectors $\bfm_i$ coordinate-wise by 
    \[ \forall a\in L, \quad m_i(a) \coloneqq \floor[\Big]{ 2^t \innerprod{a}{\psi_i} + \frac{1}{2} }.\] 
    
    Let $\ket{\psi'_i} \coloneqq \frac{1}{\norm{\bfm_i}_2} \sum_{a\in L} m_i(a)\ket{a}$, and define $\ket{\Psi'} \coloneqq \ket{\psi'_1}\otimes\cdots\otimes\ket{\psi'_k}$. 
    Following the reasoning in \Cref{eq:two-prover-product-state-rounding}, we can similarly obtain the rounding-error bound via the triangle inequality:
    \begin{equation}
    \label{eq:general-k-tensor-rounding}
        \norm{\ket{\Psi'}-\ket{\Psi}}_2 \leq \sum_{i\in[k]} \norm{ \ket{\psi'_i} - \ket{\psi_i} }_2 \leq k\delta = \frac{\Delta}{5}. 
    \end{equation}
    As a consequence, as in \Cref{eq:two-prover-rounding-perturbation}, we similarly achieve the bound from \Cref{eq:general-k-tensor-rounding}:
    \begin{equation}
    \label{eq:general-k-perturbation}
        \abs*{ \bra{\Psi'}Q_\Gamma\ket{\Psi'} - \bra{\Psi} Q_\Gamma \ket{\Psi} } \leq \norm*{ \ket{\Psi'} - \ket{\Psi} }_2 \leq \frac{\Delta}{5}. 
    \end{equation}

    The underlying classical witness $(\bfm_1,\dots,\bfm_k)$ is of polynomial-size because $k(n) \leq \poly(n)$, $\ell(n) = O(\log{n})$, each integer $m_i(a)$ uses only $O\rbra*{ \log\rbra[\big]{2^\ell k/\Delta} }$ bits.

    \paragraph{Estimating the stoquastic separable value via dual access to the witness state.}
    Noting that $\bra{\Psi'}Q_\Gamma\ket{\Psi'} = (1+\bra{\Psi'}M_\Gamma\ket{\Psi'})/2$, it suffices to estimate the overlap $\bra{\Psi'}M_\Gamma\ket{\Psi'}$. Our approach is inspired by the framework underlying the proof of $\mathsf{eStoqMA}\subseteq \MA$~\cite[Section~3]{Liu21}, which boils down to the dual-access viewpoint from distribution testing~\cite{CR14}. 
    
    The starting point is that both query and sample access to $\ket{\Omega'}$ can be implemented efficiently from the classical witness $(\bfm_1,\cdots,\bfm_k)$, where the initialized rounded witness $\ket{\Omega'} \coloneqq \ket{\Psi'}\ket{\bar{0}}\ket{\bar{+}}$. Implementing the query access to $\ket{\Omega'}$ follows immediately from the tensor-product structure of $\ket{\Psi'}$ and the fact that $\bfm_i$ serves as a classical description of $\ket{\psi'_i}$ whose amplitudes are \emph{dyadic} before normalization.  
    To implement the sample access to $\ket{\Omega'}$, for each proof $i\in[k]$, let $D_i$ be the squared-amplitude distribution of $\ket{\psi'_i}$, with the associated random variable $X_i$. Then, 
    \begin{equation}
        \label{eq:StoqMAk-log-distribution}
         D_i(a) \coloneqq \frac{m_i(a)^2}{Z_i} = \Pr[X_i=a], \quad\text{where } Z_i \coloneqq \sum_{a\in L} m_i(a)^2. 
    \end{equation}
    Consequently, one can sample from the distribution $D_i$ exactly: 
    \begin{enumerate}[label={\upshape(\arabic*)}]
        \item Draw an integer $W_i \sim \Uniform\rbra{\sbra{Z_i}}$ using $O(\log Z_i)$ random bits; 
        \item Identify the unique $a_i\in L$ such that $\sum_{b<a_i} m_i(b)^2 < W_i \leq \sum_{b\leq a_i} m_i(b)^2$.
    \end{enumerate}
    
    Since the lists $\bfm_i$ are explicit and $Z_i$ has polynomial bit length, this sampling procedure runs in probabilistic polynomial time. By sampling all $k$ proofs independently in this way and then sampling $u \sim \Uniform(U)$, one can implement the sample access to $\ket{\psi'_i}$. 

\LinesNumbered
\begin{algorithm}[!ht]
    \UseProtocolCounter
    \caption{An \MA{} protocol for estimating short-proof stoquastic separable values.}
    \label{protocol:StoqMAk-log-in-MA}
    \SetEndCharOfAlgoLine{.}
    \setlength{\parskip}{5pt}
    \SetKwFor{While}{}{:}{}
    \SetKwFor{For}{For}{:}{}
    \SetKwIF{If}{ElseIf}{Else}{If}{:}{elif}{Else:}{}%
    \SetKwComment{Comment}{// }{}
    \SetKwInOut{Input}{Input}
    \SetKwInOut{Output}{Output}
    \SetKwInOut{Parameter}{Parameters}

    \Input{A reversible circuit $\Gamma$.}
    \Output{ACCEPT or REJECT.}
    \Parameter{The completeness parameter $c$, the promise gap parameter $\Delta$, the number of proofs $k$, the witness length $\ell$, the number of $\ket{0}$ ancillary qubits $m_0$, and the number of $\ket{+}$ ancillary qubits $r$.}

    {\color{gray}\Comment{Receive and check the legality of the classical witness encoded $\ket{\Psi'}$.}}
    Receive the classical witness $(\bfm_1,\cdots,\bfm_k)$\;
    \For{each proof $i\in\sbra{k}$}{
        Set $Z_i \coloneqq \sum_{a\in L} m_i(a)^2$\;
        \If{$Z_i=0$}{
            Return REJECT\;
        }
    }
    {\color{gray}\Comment{Estimate the inner product $\bra{\Psi'} Q_\Gamma \ket{\Psi'}$ via dual accesses to $\ket{\Psi'}$.}}
    Set $L\coloneqq \binset^{\ell}$, $U\coloneqq \binset^r$, and $N\coloneqq \ceil*{12/\Delta^2}$\;
    \For{each sample $j \in [N]$}{
        \For{each proof $i\in\sbra{k}$}{
            Sample $W_i \sim \Uniform(\sbra{Z_i})$ using $O(\log{Z_i})$ random bits, and identify the unique $a_i\in L$ with $\sum_{b<a_i} m_i(b)^2 < W_i \leq \sum_{b\le a_i} m_i(b)^2$\;
        }
        Sample $u \sim \Uniform(U)$\;
        Compute $\Gamma(a_1,\dots,a_k,0^{m_0},u) = (a'_1,\dots,a'_k,z,u')$\;
        \If{$z \neq 0^{m_0}$}{
            Set $R_j\coloneqq 0$\;
        }\Else{
            Set $R_j\coloneqq \prod_{i\in\sbra{k}} \frac{m_i(a'_i)}{m_i(a_i)}$\;
        }
        Set $Y_j\coloneqq (1+R_j)/2$\;
    }
    {\color{gray}\Comment{Return an estimate of $\bra{\Psi'} Q_\Gamma \ket{\Psi'}$ as a decision.}}
    Set $\widehat{Y} \coloneqq \frac{1}{N}\sum_{j=1}^{N} Y_j$\;
    \If{$\widehat{Y} \geq c-\Delta/2$}{
        Return ACCEPT\;
    }\Else{
        Return REJECT\;
    }
\end{algorithm}

    \parheading{The \MA{} protocol.}
    We now formalize the \MA{} protocol in \Cref{protocol:StoqMAk-log-in-MA}. Let $(a'_1,\dots,a'_k,z,u') \coloneqq \Gamma (a_1,\dots,a_k,z,u)$, a direct calculation shows that 
    \begin{subequations}
    \label{eq:ratio-random-variable}
    \begin{align}
        \bra{\Psi'} M_\Gamma \ket{\Psi'} = \bra{\Omega'} \Gamma \ket{\Omega'} &= \sum_{\substack{(a_1,\dots,a_k)\in L^k\\u\in U, ~z=0^{m_0}}} \prod_{i\in[k]} \sqrt{D_i(a_i)D_i(a'_i)} \cdot 2^{-r}\\
        &= \sum_{\substack{(a_1,\dots,a_k)\in L^k\\z=0^{m_0}}} \prod_{i\in[k]} D_i(a_i) \cdot \frac{\sqrt{D_i(a'_i)}}{\sqrt{D_i(a_i)}} \\
        &= \E_{(A_1,\dots,A_k)\sim D_1\times\cdots\times D_k} R(A_1,\dots,A_k), \\
        \text{where} \quad R(A_1,\dots,A_k) &\coloneqq 
        \begin{cases}
            \prod_{i\in[k]} \frac{m_i(A'_i)}{m_i(A_i)}, & \text{ if } z=0^{m_0}\\
            0, & \text{otherwise}
        \end{cases}.
    \end{align}
    \end{subequations}
    Here, the second line uses the fact that $\abs{U}=2^r$, and the last line follows from \Cref{eq:StoqMAk-log-distribution}, and $R \coloneqq R(A_1,\dots,A_k)$ denotes the ratio random variable. 
    Consider the random variable $Y \coloneqq (1+R)/2$, and then \Cref{eq:ratio-random-variable} immediately implies that $\bbE[Y] = \bra{\Psi'}Q_\Gamma\ket{\Psi'}$. 

    Next, we bound the second moment of $R$. Since $\Gamma$ is a permutation, the map $(a_1,\dots,a_k,u) \mapsto (a'_1,\dots,a'_k,u')$ is injective on the subset of tuples satisfying $z=0^{m_0}$. Consequently, 
    \begin{equation}
    \label{eq:StoqMAk-log-2nd-moment}
        \E[R^2] = \sum_{\substack{(a_1,\dots,a_k)\in L^k\\z=0^{m_0}}} \prod_{i\in[k]} \frac{m_i(a_i)^2}{Z_i} \cdot \prod_{j\in[k]} \frac{m_j(a'_j)^2}{m_j(a_j)^2}
        = \sum_{\substack{(a_1,\dots,a_k)\in L^k\\z=0^{m_0}}} \prod_{i\in[k]} \frac{m_i(a'_i)^2}{Z_i}
        \leq 1.
    \end{equation}

    Therefore, we obtain the following variance bounds from \Cref{eq:StoqMAk-log-2nd-moment}:
    \begin{equation}
        \label{eq:StoqMAk-log-variance-bounds}
        \Var[R] = \E[R^2]-\E[R]^2 \leq 1 \quad\text{and}\quad \Var[Y] = \frac{\Var[R]}{4} \leq \frac{1}{4}.
    \end{equation}
    
    Let $\widehat{Y} \coloneqq \frac{1}{N}\sum_{j=1}^N Y_j$ be the empirical mean of the $N$ independent samples produced by \Cref{protocol:StoqMAk-log-in-MA}. By Chebyshev's inequality (e.g.,~\cite[Theorem 3.6]{MU17}), we obtain
    \begin{equation}
        \label{eq:StoqMAk-log-estimation-bound}
        \Pr\sbra*{ \abs*{\widehat{Y} - \E\sbra*{Y}} \geq \frac{\Delta}{4} } \leq \frac{16 \Var\sbra[\big]{\widehat{Y}}}{\Delta^2}
        = \frac{16 \Var\sbra{Y}}{N \Delta^2}
        \leq \frac{16}{48} = \frac{1}{3}.
    \end{equation}
    Here, the last inequality uses \Cref{eq:StoqMAk-log-variance-bounds} and the fact that $N=\ceil*{12/\Delta^2}$. 
    
    We now verify the completeness and soundness condition: 
    \begin{itemize}
        \item For \emph{yes} instances, there exists a non-negative product state $\ket{\Psi}$ such that $\bra{\Psi} Q_\Gamma \ket{\Psi} \geq c$. \Cref{eq:StoqMAk-log-estimation-bound} yields the following bound, which holds with probability at least $2/3$, 
        \begin{align*}
            \widehat{Y} \geq \E[Y] - \frac{\Delta}{4} 
            &= \bra{\Psi'} Q_\Gamma \ket{\Psi'} - \frac{\Delta}{4} \\
            &\geq \bra{\Psi} Q_\Gamma \ket{\Psi} - \abs*{ \bra{\Psi'} Q_\Gamma \ket{\Psi'} - \bra{\Psi} Q_\Gamma \ket{\Psi} } - \frac{\Delta}{4} \\
            &\geq c - \frac{\Delta}{5} - \frac{\Delta}{4} > c - \frac{\Delta}{2}.
        \end{align*}
        Here, the second line uses the reverse triangle inequality, and the last line follows from \Cref{eq:general-k-perturbation}. The verifier thus accepts with probability at least $2/3$. 
        \item For \emph{no} instances, every non-negative product state, including $\ket{\Psi'}$, satisfies $\bra{\Psi'} Q_\Gamma \ket{\Psi'} \leq s$. \Cref{eq:StoqMAk-log-estimation-bound} yields the following bound, which holds with probability at least $2/3$, 
        \begin{align*}
            \widehat{Y} \leq \E\sbra{Y} + \frac{\Delta}{4} 
            = \bra{\Psi'} Q_\Gamma \ket{\Psi'} +  \frac{\Delta}{4}
            \leq s + \frac{\Delta}{4} < s+\frac{\Delta}{2} = c-\frac{\Delta}{2}.
        \end{align*}
        Consequently, the verifier accepts with probability at most $1/3$. \qedhere
    \end{itemize}
\end{proof}

%%%%%%%%%%%%%%%%%%%%%%%%%%%%%%%%%%%%%%%%%%%%%%%%%%%%%%%%%%%
%%%%%%%%%%%%%%%%%%%%%%%%%%%%%%%%%%%%%%%%%%%%%%%%%%%%%%%%%%%
%%%%%%%%%%%%%%%%%%%%%%%%%%%%%%%%%%%%%%%%%%%%%%%%%%%%%%%%%%%

\section{Upper bounds for \StoqMAk{} with nearly perfect completeness}
\label{sec:StoqMAtwo-perfect-completeness-upper-bound}

In this section, we establish upper bounds for \StoqMAk{} with \emph{nearly perfect} completeness, which are stronger than the general deterministic exponential-time upper bound for \StoqMAk{} stated in \Cref{thm:StoqMAk-in-EXP} and proved in \Cref{sec:sos} via the Sum-of-Squares algorithm of Barak--Kelner--Steurer~\cite{BKS14}: 

\begin{theorem}[\StoqMAk{} with doubly exponentially small completeness error is in \PSPACE{}]
    \label{thm:StoqMAtwo-perfect-completeness}
    Let $\epsilon(n)$ and $s(n)$ be efficiently computable functions  satisfying $0 \leq \epsilon(n) \leq \exp(-\exp(\poly(n)))$ and $1/2 \leq s(n) \leq 1-1/\poly(n)$, and define $\Delta(n) \coloneqq 1-\epsilon(n)-s(n)$. Then, for all efficiently computable positive integer-valued functions $k(n)$, $\ell(n)$, $m_0(n)$, and $m_+(n)$, it holds that:
    \[ \StoqMA_\ell\ssbra{k,1-\epsilon,\Delta} \subseteq \DTISP\sbra[\bigg]{ \exp\rbra[\Big]{\widetilde{O}\rbra[\Big]{\frac{k^2\ell^2}{\Delta^2}}} \cdot 2^{m_+}, \widetilde{O}\rbra[\Big]{\frac{k^2\ell^2}{\Delta^2}}+O(m_0+m_+)}. \]
    In particular, for every efficiently computable integer-valued function $k$ that is polynomially bounded in $n$, and for functions $\epsilon$ and $s$ satisfying the above conditions, we have: 
    \[ \StoqMA(k,1-\epsilon,s) \subseteq \PSPACE. \]
\end{theorem}

\begin{theorem}[\PreciseStoqMAk{} with triply exponentially small completeness error is in \EXP{}]
    \label{thm:PreciseStoqMAtwo-perfect-completeness}
    Let $\epsilon(n)$ and $s(n)$ be efficiently computable functions  such that 
    \[0 \leq \epsilon(n) \leq \exp(-\exp(\exp(\poly(n)))) \quad\text{and}\quad 1/2 \leq s(n) \leq 1-1/\exp(\poly(n)),\] 
    and define $\Delta(n) \coloneqq 1-\epsilon(n)-s(n)$. Then, for all efficiently computable positive integer-valued functions $k(n)$, $\ell(n)$, and $m_+(n)$, the following inclusion holds:
    \[ \StoqMA_\ell\ssbra{k,1-\epsilon,\Delta} \subseteq \DTISP\sbra*{ \exp\rbra[\big]{\widetilde{O}\rbra{\max\cbra{k\ell,m_+}}}, O\rbra[\big]{2^{k\ell}} }. \]
    In particular, for every efficiently computable integer-valued function $k$ that is polynomially bounded in $n$, and for functions $\epsilon$ and $s$ satisfying the above conditions, it holds that: 
    \[ \PreciseStoqMA(k,1-\epsilon,s) \subseteq \EXP. \]
\end{theorem}

\vspace{1em}
In the remainder of this section, we first introduce the rectangular structure of non-negative product states in \Cref{subsec:rect-structure}, via the \StoqMAtwo{}-complete problem \SepRCD{}, and then explain how to test rectangular closure. 
Next, we provide two implementations of the formal algorithm (\algoref{algo:rectangle-closure-testing}) together with their performance guarantees (\Cref{thm:rectangular-closure-testing}) in \Cref{subsec:rect-closure-implementations}, and show how the complexity-theoretic consequences in \Cref{thm:StoqMAtwo-perfect-completeness,thm:PreciseStoqMAtwo-perfect-completeness} are derived. 
Lastly, we prove the correctness of the rectangular closure testing algorithm in \Cref{subsec:rect-closure-testing-analysis}. 

\subsection{Rectangular structure of non-negative product states}
\label{subsec:rect-structure}

We begin by describing the \emph{rectangular structure} of non-negative witness states underlying \emph{yes} instances of \StoqMAtwo{} with perfect completeness. In particular, we consider the \StoqMAtwo{}-complete problem \SepRCD{}, introduced in \Cref{subsubsec:SepRCD}.

\begin{definition}[Perfectly agreeing pair]
    For any \SepRCD{} instance $(R_0,R_1)$, we say that $(R_0, R_1)$ is a \emph{perfectly agreeing pair} if there exists a non-negative product state $\ket{\psi}$ such that 
    \begin{equation}
        \label{eq:perfectly-agreeing-pair-cond}
        \innerprod{R_0(\psi)}{R_1(\psi)}=1, \quad\text{where } \ket{R_z(\psi)} \coloneqq R_z \ket{\psi}\ket{0}^{\otimes m_0}\ket{+}^{\otimes r} \text{ for each } z\in\binset.
    \end{equation}  
    Furthermore, we say that $(R_0,R_1)$ is an \emph{$\epsilon$-near perfectly agreeing pair} if the condition in \Cref{eq:perfectly-agreeing-pair-cond} is relaxed to $\innerprod{R_0(\psi)}{R_1(\psi)}\geq 1-\epsilon$. 
\end{definition}

\begin{lemma}[Perfectly agreeing pairs admit product subset states]
    \label{lemma:rectangular-structure}
    For every perfectly agreeing pair $(R_0,R_1)$ corresponding to a \emph{yes} instance of $\SepRCD_{1,\beta}(\ell,m_0,r)$, let $\ket{\psi^\star} \coloneqq \ket{\psi^\star_1}\otimes\ket{\psi^\star_2}$ be a non-negative state that satisfies $\innerprod{R_0(\psi^\star)}{R_1(\psi^\star)}=1$. Then there exist non-empty subsets 
    \[S_1 \subseteq \supp(\ket{\psi^\star_1}) \quad\text{and}\quad S_2 \subseteq \supp(\ket{\psi^\star_2})\] 
    such that the product subset state $\ket{\phi} \coloneqq \ket{S_1}\otimes\ket{S_2}$ satisfies $\innerprod{R_0(\phi)}{R_1(\phi)}=1$. 
\end{lemma}

\begin{proof}
    Following \Cref{lemma:SepRCD-StoqMAk-complete}, it suffices to consider a \SepRCD{} instance $(\Gamma,I)$ such that $\innerprod{R_0(\psi)}{R_1(\psi)} = \bra{\psi}\bra{\bar{0}}\bra{\bar{+}} \Gamma \ket{\psi}\ket{\bar{0}}\ket{\bar{+}}$ for every non-negative state $\ket{\psi}$. We now write
    \[\ket{\psi^\star_1}=\sum_{x_1}\alpha^{(1)}_{x_1}\ket{x_1} \text{ and } \ket{\psi^\star_2}=\sum_{x_2}\alpha^{(2)}_{x_2}\ket{x_2},
    \quad\text{where } \alpha^{(1)}_{x_1}\geq 0 \ \forall x_1 \text{ and } \alpha^{(2)}_{x_2}\geq 0 \ \forall x_2.
    \]
    
    For the non-negative states $\ket{\psi^\star_1}$ and $\ket{\psi^\star_2}$, we define their supports by
    \[ S_1 \coloneqq \set{x_1}{\alpha^{(1)}_{x_1}>0} \quad\text{and}\quad 
    S_2 \coloneqq \set{x_2}{\alpha^{(2)}_{x_2}>0}.\]
    Thus, the support of the initialized state is the product set $K \coloneqq S_1\times S_2\times\cbra{0^{m_0}}\times U$. Since $\bra{\Omega^\star} \Gamma \ket{\Omega^\star}=1$, where $\ket{\Omega^\star} \coloneqq \ket{\psi^\star}\ket{\bar{0}}\ket{\bar{+}}$, a direct calculation shows that
    \[ \norm*{ \ket{\Omega^\star} - \Gamma \ket{\Omega^\star} }^2_2 = \innerprod{\Omega^\star}{\Omega^\star} + \bra{\Omega^\star} \Gamma^\dagger\Gamma \ket{\Omega^\star} - 2 \bra{\Omega^\star} \Gamma \ket{\Omega^\star} = 2 - 2 \bra{\Omega^\star} \Gamma \ket{\Omega^\star} = 0.\]
    Consequently, $\Gamma \ket{\Omega^\star} = \ket{\Omega^\star}$. 
    Since $\Gamma$ is a permutation matrix, any basis vector with positive amplitude cannot be sent outside the support of positive amplitudes. Thus, we obtain
    \begin{equation}
        \label{eq:perfectly-agree-part-permutation}
        \Gamma(K) \subseteq K \quad\text{and}\quad \abs{\Gamma(K)}=\abs{K},
    \end{equation}
    where the second identity holds because $\Gamma$ is bijective. Therefore, \Cref{eq:perfectly-agree-part-permutation} implies that 
    \begin{equation}
        \label{eq:perfectly-agree-part-fixedPoint}
        \Gamma(K)=K \quad\text{and}\quad \Gamma\ket{K}=\ket{K}.
    \end{equation}
    Here, $\ket{K}$ is the initialized full state of $\ket{\phi}$, corresponding to the set $K$, where the product subset state $\ket{\phi}\coloneqq\ket{S_1}\otimes\ket{S_2}$ corresponds to the set $S_1\times S_2$. Finally, using \Cref{eq:perfectly-agree-part-fixedPoint}, we complete the proof by noting that $\bra{\Omega_\phi} \Gamma \ket{\Omega_\phi} = \innerprod{\Omega_\phi}{\Omega_\phi}=1$, with $\ket{\Omega_\phi} \coloneqq \ket{\phi}\ket{\bar{0}}\ket{\bar{+}}$ as desired. 
\end{proof}

\paragraph{Testing rectangular closure.} 
With this rectangular structure of non-negative product states established in \Cref{lemma:rectangular-structure}, we are now ready to understand how a reversible circuit, or equivalently a permutation, acts on such states in combinatorial terms. 

To formalize this perspective, we need to introduce some additional notation. Specifically, we refer to a Cartesian product $A\times B$ of two sets $A,B\subseteq \binset^\ell$ as a \emph{rectangle}. As stated in \Cref{lemma:SepRCD-StoqMAk-complete}, without loss of generality, for any \SepRCD{} instance $(R_0,R_1)$, it suffices to consider the reversible circuit pair $(\Gamma,I)$, which achieves the same parameters and thresholds. Here, $\Gamma$ can be viewed as a permutation matrix. We need the following notions: 
\begin{itemize}
    \item \textbf{Initialized sector}: We denote the $(m_0,r)$-\emph{initialized sector} over the rectangle $S\times T$, where $S,T\subseteq \binset^\ell$ and $U\coloneqq \binset^r$, by 
    \[K(S,T) \coloneqq S \times T \times \cbra{0^{m_0}} \times U.\] 
    \item \textbf{Good and bad transitions}: For any $(a,b,u) \in A\times B\times U$, where $A,B\subseteq \binset^\ell$, we define the transition under $\Gamma$ by 
    \[\Gamma(a,b,0^{m_0},u) = (a',b',z',u').\] 
    If $z' \neq 0^{m_0}$, we call this a \emph{bad transition}. Otherwise, we call it a \emph{good transition} and write 
    \[\GammaZero(a,b,u)=(a',b').\]
\end{itemize}

Next, we define the notions of \emph{closed rectangle} and \emph{rectangular closure}. The former notion is a static certificate: once the rectangle $S\times T$ is chosen, \Cref{def:closed-rectangle} asks whether $\Gamma$ maps the initialized sector over that rectangle back into the same sector. By contrast, the latter notion is a dynamic construction introduced in \Cref{def:rectangular-closure}: the underlying procedure starts with a pair $(a_0,b_0)$, repeatedly updates the sets $S^+$ and $T^+$ for the current round, and either reaches a closed rectangle or encounters a bad transition. 

\begin{definition}[Closed rectangle]
    \label{def:closed-rectangle}
    For any given permutation $\Gamma$, a rectangle $S\times T\subseteq A\times B$ is \emph{closed} with respect to $\Gamma$ if the following inclusion holds. 
    \[\Gamma(K(S,T)) \subseteq K(S,T)\] 
    Equivalently, this condition means that every transition that starts in the $(m_0,r)$-initialized sector over the rectangle $S\times T$ stays in the initialized sector over the same rectangle.   
\end{definition}

\begin{definition}[Rectangular closure from a starting pair]
    \label{def:rectangular-closure}
    Fix a starting pair $(a_0,b_0)\in A\times B$.  Let $S_0 \coloneqq \{a_0\}$ and $T_0 \coloneqq \{b_0\}$. For sets $S_t$ and $T_t$ that correspond to round $t$, we say that round $t$ is \emph{bad} if there exist $a\in S_t$, $b\in T_t$, and $u\in U$ such that
    \[ \Gamma(a,b,0^{m_0},u)=(a',b',z',u') \quad\text{with}\quad z'\ne 0^{m_0}. \]
    If round $t$ is not bad, we say that this round is \emph{good} and define
    \begin{align*}
        S_{t+1} &\coloneqq S_t\cup \set{a'}{\exists a\in S_t,\ b\in T_t,\ u\in U, \ \GammaZero(a,b,u)=(a',b')},\\
        T_{t+1} &\coloneqq T_t\cup \set{b'}{\exists a\in S_t,\ b\in T_t,\ u\in U, \ \GammaZero(a,b,u)=(a',b')}.
    \end{align*}
\end{definition}

\subsection{Rectangular closure testing and its consequences}
\label{subsec:rect-closure-implementations}

We now formalize our intuition for testing rectangular closure and present the actual algorithm in \algoref{algo:rectangle-closure-testing}. In particular, the procedure described in \Cref{def:rectangular-closure} is implemented in Lines 5--16 of \algoref{algo:rectangle-closure-testing}. 

\LinesNumbered
\begin{algorithm}[ht!]
    \SetAlgorithmName{Algorithm}{algorithm}{List of Algorithms}
    \caption{Rectangular Closure Testing Algorithm.}
	\label{algo:rectangle-closure-testing}
    \SetEndCharOfAlgoLine{.}
    \setlength{\parskip}{5pt}
    \SetKwFor{While}{}{:}{}
    \SetKwFor{For}{For}{:}{}
    \SetKwIF{If}{ElseIf}{Else}{If}{:}{elif}{Else:}{}%
    \SetKwComment{Comment}{// }{}
    \SetKwInOut{Input}{Input}
    \SetKwInOut{Output}{Output}
    \SetKwInOut{Parameter}{Parameters}

    \Input{A reversible circuit $\Gamma$.}
    \Output{ACCEPT or REJECT.}
    \Parameter{The soundness parameter $\gamma\in(0,1)$, the witness length $2\ell$, the number of $\ket{0}$ ancillary qubits $m_0$, and the number of $\ket{+}$ ancillary qubits $r$.}
    Set $L \coloneqq \ceil*{ (2\ell\ln{2} + 1)/\ln(1+\gamma)}$\;
    Set $A\coloneqq\binset^{\ell}$, $B\coloneqq\binset^{\ell}$, and $U\coloneqq\binset^r$\;
    \For{\textnormal{each starting pair $(a_0,b_0) \in A\times B$}}{
        Initialize $S_0 \coloneqq \cbra{a_0}$ and $T_0 \coloneqq \cbra{b_0}$\;
        {\color{gray}\Comment{Rectangular closure from $(a_0,b_0)$.}}
        \For{each round $t \in \cbra{0,1,\cdots,L-1}$}{
            Set temporary sets $S^+ \coloneqq S_t$ and $T^+ \coloneqq T_t$\;
            \For{\textnormal{every triple $(a,b,u) \in S_t\times T_t\times U$}}{
                Compute $\Gamma(a,b,0^{m_0},u)=(a',b',z',u')$\;
                \If{$z'\neq 0^{m_0}$}{
                   Declare that round $t$ is \emph{bad}\;
                   Break\;
                }\Else{
                   Update $S^+\leftarrow S^+\cup\{a'\}$ and $T^+\leftarrow T^+\cup\{b'\}$\;
                }
            } 
            \If{\textnormal{round $t$ is \emph{bad}}}{
                Break\;
            }
            Set $S_{t+1} \coloneqq S^+$ and $T_{t+1} \coloneqq T^+$\;
        }
        \If{\textnormal{all rounds $t\in\cbra{0,1,\cdots,L-1}$ are \emph{good} for $(a_0,b_0)$}}{
            Return ACCEPT\;
        }
    }
    Return REJECT\; 
\end{algorithm}

Next, we present the performance guarantees and detailed complexity bounds of \algoref{algo:rectangle-closure-testing} under two different implementations, and the proof can be found in \Cref{subsec:rect-closure-testing-analysis}:

\begin{theorem}[Explicit rectangular closure testing algorithms for \SepRCD{}]
    \label{thm:rectangular-closure-testing}
    For every efficiently computable function $\gamma(n)$ taking values in $(0,1)$, there is an $L$-round rectangular closure testing algorithm for $\SepRCD_{1-\epsilon,1-\gamma}(\ell,m_0,r)$, with the parameters 
    \[L=\ceil*{ \frac{2\ell \ln{2} +1}{\ln(1+\gamma)}} = O\rbra*{\frac{\ell}{\gamma}} \quad\text{and}\quad \epsilon \leq 2^{5-r-2^{L+1}(\ell+4)}.\]
    In particular, let $T_\Gamma$ and $S_\Gamma$ denote the time and workspace needed to evaluate $\Gamma$ on the input $(a,b,0^{m_0},u)\in A\times B\times \cbra{0^{m_0}} \times U$, and let 
    \[\lambda \coloneqq 2\ell+m_0+r+\lceil\log(L+2)\rceil.\] 
    Then, \algoref{algo:rectangle-closure-testing} can be implemented deterministically in either of the following ways:
    \begin{enumerate}[label={\upshape(\arabic*)}]
        \item \label{thmitem:recursive} The recursive implementation takes time $2^{O(\ell L)+r}\cdot\poly(L)\cdot T_\Gamma$ and space $O\rbra*{L\rbra{\ell\!+\!\log(L\!+\!2)}\!+\!\lambda\!+\!S_\Gamma}$.
        \item \label{thmitem:explicit-table} The explicit-table implementation takes time $O\rbra*{  2^{4\ell+r}T_\Gamma }$ and space $O(2^\ell+\lambda+S_\Gamma)$. 
    \end{enumerate}
\end{theorem}

\begin{remark}[Optimal $r$-dependence under ETH]
The $r$-dependence in the time complexity of \Cref{thm:rectangular-closure-testing}\ref{thmitem:explicit-table} is \emph{optimal}, up to constant factors in the exponent, for $\ell=1$, perfect completeness, and $\gamma=2^{-r}$. To this end, given a $3$-CNF formula $F$ on $r$ variables, let $u\in\binset^r$ be a label for the $r$ ancilla qubits in the $\ket{+}$ state.  
Take $R_0=I$ and construct a reversible circuit $R_1$ satisfying 
\[ \forall a,b\in\binset, \; \forall u\in\binset^r, \quad R_1\ket{a,b,0^q,0,u} =\ket{a,b,0^q,F(u),u}. \]
Here, $a$ and $b$ label the two witness qubits and $q=O(\abs{F})$ is the number of work qubits.\footnote{This circuit $R_1$ is obtained by computing $F(u)$, copying its value into a separate qubit initialized to $\ket{0}$, and reversing the computation to restore the work qubits to $\ket{0^q}$.}
Consequently, we obtain $\innerprod{R_0}{R_1} = 1 - 2^{-r}\#\SAT(F)$, implying that improving the time complexity to $2^{o(r)}\poly(n)$ would refute ETH.
\end{remark}

\subsubsection{Complexity-theoretic consequences of \texorpdfstring{\Cref{thm:rectangular-closure-testing}}{}}
\label{subsubsec:rect-closure-testing-consequences}
Combined with \Cref{lemma:SepRCD-StoqMAk-complete}, \Cref{thm:rectangular-closure-testing} implies the complexity-theoretic consequences stated in \Cref{thm:StoqMAtwo-perfect-completeness,thm:PreciseStoqMAtwo-perfect-completeness}, upon identifying the parameter $r$ used below with the parameter $m_+$ used in those theorems. 
Let $g_\Gamma(n)$ be the number of gates in $\Gamma$. 
Since all parameters $\ell(n)$, $r(n)$, $m_0(n)$, and $g_\Gamma(n)$ are polynomially bounded, it holds that 
\[ T_\Gamma = O(g_{\Gamma}\cdot\rbra{2\ell+m_0+r}) = \poly(n) \quad\text{and}\quad S_\Gamma = O(2\ell+m_0+r+\log g_\Gamma) = \poly(n). \]
Consequently, we obtain the following bounds for \StoqMAk{} with nearly perfect completeness:\footnote{Applying prover compression (\Cref{corr:StoqMAk-eq-StoqMAtwo-small-yesError}\ref{thmitem:prover-compression-negl-lengthPreserved}) adds only $O(k)$ ancillary qubits initialized to $\ket{+}$ and $O(k\ell)$ initialized to $\ket{0}$.}
\begin{enumerate}[label={\upshape(\arabic*)}]
    \item For the recursive implementation in \Cref{thm:rectangular-closure-testing}\ref{thmitem:recursive}, the underlying algorithm runs in deterministic $\exp\rbra[\big]{\widetilde{O}\rbra{\ell^2/\gamma}}
    \cdot 2^r$ time and $\widetilde{O}\rbra{\ell^2/\gamma}+O(m_0+r)$ space. By combining with prover compression (\Cref{corr:StoqMAk-eq-StoqMAtwo-small-yesError}\ref{thmitem:prover-compression-negl-lengthPreserved}), we obtain an algorithm for \StoqMAk{} that runs in deterministic $\exp\rbra[\big]{\widetilde{O}\rbra{k^2\ell^2/\gamma^2}}\cdot 2^r$ time and $\widetilde{O}\rbra{k^2\ell^2/\gamma^2}+O(m_0+r)$ space.

    When $\gamma(n)\geq 1/\poly(n)$, the allowed completeness error $\epsilon(n)$ is inverse doubly exponential in $n$, and the space complexity is polynomial in $n$, yielding the \PSPACE{} upper bound. 

    \item For the explicit-table implementation in \Cref{thm:rectangular-closure-testing}\ref{thmitem:explicit-table}, the underlying algorithm runs in deterministic $O\rbra[\big]{2^{4\ell+r}} \cdot \poly(n) = \exp\rbra[\big]{\widetilde{O}\rbra{\max\cbra{\ell,r}}}$ time and $O(2^{\ell})$ space. By combining with prover compression (\Cref{corr:StoqMAk-eq-StoqMAtwo-small-yesError}\ref{thmitem:prover-compression-negl-lengthPreserved}), we obtain an algorithm for \StoqMAk{} that runs in deterministic $\exp\rbra[\big]{\widetilde{O}\rbra{\max\cbra{k\ell,r}}}$ time and $O(2^{k\ell})$ space.

    When $\gamma(n)\geq 1/\exp(\poly(n))$, the allowed completeness error $\epsilon(n)$ is inverse triply exponential in $n$, and the deterministic runtime is $\exp\rbra[\big]{\widetilde{O}\rbra{\max\cbra{k\ell,r}}} \leq 2^{\poly(n)}$, which leads to the \EXP{} upper bound. 
\end{enumerate}

\subsection{Analysis of \texorpdfstring{\algoref{algo:rectangle-closure-testing}}{}}
\label{subsec:rect-closure-testing-analysis}

We will analyze the correctness of \algoref{algo:rectangle-closure-testing} separately for \emph{yes} and \emph{no} instances, namely the completeness and soundness, and then put everything together. 

\paragraph{Completeness.} We begin with the perfect completeness case as a warm-up:
\begin{lemma}[Perfectly agreeing pairs imply rectangular closure]
    \label{lemma:rectangle-closure-testing-perfect-completeness}
    Let $(\Gamma,I)$ be a perfectly agreeing pair, viewed as a \emph{yes} instance of $\SepRCD_{1,\beta}(\ell,m_0,r)$. Then, there exists a starting pair whose rectangular closure never encounters a bad round. 
\end{lemma}

\begin{proof}
By utilizing \Cref{lemma:rectangular-structure}, there exist non-empty sets $S,T\subseteq\binset^\ell$ such that $\Gamma \ket{K(S,T)}=\ket{K(S,T)}$, where the initialized sector $K(S,T)=S\times T\times \cbra{0^{m_0}}\times U$ and $U=\binset^r$. 
Since $\Gamma$ is a permutation, it follows that
\begin{equation}
    \label{eq:perfect-completeness-fixedPoint}
    \Gamma(K(S,T))=K(S,T)
\end{equation}

Next, pick any starting pair $(a_0,b_0)\in S\times T$. We will prove by induction that
\[S_t\subseteq S \;\text{and}\; T_t\subseteq T \text{ hold for all rounds } t,\] 
and hence no round is bad. It is easy to see that the induction hypothesis holds at round $t=0$, because $S_0=\cbra{a_0}$ and $T_0=\cbra{b_0}$. 

Now assume that the induction hypothesis holds at round $t$, equivalently, it holds that:
\[\text{If } a\in S_t\subseteq S, \; b\in T_t\subseteq T, \;\text{and}\; u\in U, \quad\text{then } (a,b,0^{m_0},u)\in K(S,T).\]

According to \Cref{eq:perfect-completeness-fixedPoint}, the image of $\Gamma$ lies again in $S\times T\times\{0^{m_0}\}\times U$, and thus the transition in this round is not bad. Therefore, any newly added coordinates $(a',b')$ satisfy $a'\in S$ and $b'\in T$, which establishes the induction hypothesis at round $t+1$ via
\[S_{t+1}\subseteq S \quad\text{and}\quad T_{t+1}\subseteq T. \qedhere\] 
\end{proof}

Next, we relax the perfect completeness requirement in \Cref{lemma:rectangle-closure-testing-perfect-completeness} and extend the result to the setting in which the completeness error is sufficiently small relative to the number of rounds: 

\begin{lemma}[Near perfectly agreeing pairs imply rectangular closure in bounded rounds]
    \label{lemma:rectangle-closure-testing-nearly-perfect-completeness}
    Let $(\Gamma,I)$ be an $\epsilon$-near perfectly agreeing pair, viewed as a \emph{yes} instance of $\SepRCD_{1-\epsilon,\beta}(\ell,m_0,r)$. 
    If the completeness error satisfies the upper bound
    \[ \epsilon \leq \tau^2_L/2^{r+3}, \quad\text{where } \tau_0 \coloneqq 2^{-\ell} \;\text{and}\; \tau_{t+1} \coloneqq \tau^2_t/16,\] 
    then there exists a starting pair whose rectangular closure does not encounter a bad round up to round $L$. 
\end{lemma}

\begin{proof}
    Let $\ket{\psi_1}\ket{\psi_2}$ be a non-negative product state whose initialized full state $\ket{\Omega}$ satisfies 
    \begin{equation}
        \label{eq:near-perfect-overlap}
        \bra{\Omega}\Gamma\ket{\Omega} \geq 1-\epsilon. 
    \end{equation}
    Let $p_1$ and $p_2$ be probability distributions induced by the states $\ket{\psi_1}$ and $\ket{\psi_2}$, respectively, defined by
    \[ \forall  x\in\binset^\ell, \quad p_1(x) \coloneqq \abs*{\innerprod{x}{\psi_1}}^2 \quad\text{and}\quad p_2(x) \coloneqq \abs*{\innerprod{x}{\psi_2}}^2. \]
    Since both distributions $p_1$ and $p_2$ are defined over $[2^{\ell}]$, there exist $a_0$ and $b_0$ such that 
    \[ p_1(a_0) \geq 2^{-\ell} = \tau_0 \quad\text{and}\quad p_2(b_0) \geq 2^{-\ell} = \tau_0. \]
    We use this $(a_0,b_0)$ as the starting pair.

    \vspace{1em}
    We first bound the global near-invariance error. Since $\Gamma$ is a permutation and $\innerprod{\Omega}{\Omega}=1$, together with \Cref{eq:near-perfect-overlap}, we obtain
    \[ \norm*{\Gamma\ket{\Omega}-\ket{\Omega}}_2^2 
    = \bra{\Omega}\Gamma^\dagger\Gamma\ket{\Omega} + \innerprod{\Omega}{\Omega} - 2\bra{\Omega}\Gamma\ket{\Omega}
    = 2 - 2\bra{\Omega}\Gamma\ket{\Omega}
    \leq 2-2(1-\epsilon) = 2\epsilon.\]
    Consequently, for every coordinate $y\in\binset^{2\ell+m_0+r}$, it follows that
    \begin{equation}
        \label{eq:near-perfect-coordinate-bound}
        \abs*{ \bra{y}\Gamma\ket{\Omega} - \innerprod{y}{\Omega}} \leq \sqrt{2\epsilon}.
    \end{equation}

    Using the assumed threshold on $\epsilon$ and the fact that $\tau_t\ge \tau_L$ for every $t\leq L$, we obtain
    \begin{equation}
        \label{eq:near-perfect-error-threshold}
        \forall t\in\cbra{0,1,\dots,L}, \quad \sqrt{2\epsilon} \leq \sqrt{\frac{\tau^2_L}{4\cdot 2^r}} = \frac{\tau_L}{2\cdot 2^{r/2}} \leq \frac{\tau_t}{2\cdot 2^{r/2}}.
    \end{equation}

    \vspace{1em}
    Next, we prove by induction that
    \[ S_t\subseteq\set{a}{p_1(a) \geq \tau_t} \;\text{and}\; T_t\subseteq\set{b}{p_2(b) \geq \tau_t} \text{ for all rounds } t\leq L, \]
    and hence no round before round $L$ is bad. The induction hypothesis holds at round $t=0$, as $S_0=\cbra{a_0}$ and $T_0=\cbra{b_0}$. 

    Now assume that the induction hypothesis holds at round $t<L$. Then, if $a\in S_t$, $b\in T_t$, and $u\in U$, and we write $x\coloneqq (a,b,0^{m_0},u)$, we have
    \begin{equation}
        \label{eq:near-perfect-hypothesis-t}
        \innerprod{x}{\Omega} = \sqrt{\frac{p_1(a)p_2(b)}{2^r}} \geq \sqrt{\frac{\tau^2_t}{2^r}} = \frac{\tau_t}{2^{r/2}}. 
    \end{equation}
    
    Let $y \coloneqq \Gamma x$, which implies $\bra{y}\Gamma\ket{\Omega} = \innerprod{x}{\Omega}$. Consequently, combining the bounds in \Cref{eq:near-perfect-coordinate-bound,eq:near-perfect-error-threshold} yields that 
    \begin{equation}
        \label{eq:near-perfect-difference-bound}
        \abs*{\innerprod{x}{\Omega}-\innerprod{y}{\Omega}} = \abs*{\bra{y}\Gamma\ket{\Omega}-\innerprod{y}{\Omega}} \leq \sqrt{2\epsilon} \leq \frac{\tau_t}{2\cdot 2^{r/2}}.
    \end{equation}

    We then show the transition in the current round is not bad by contradiction. More precisely, if $y$ were outside the initialized sector, i.e.~if $z\neq 0^{m_0}$, then $\innerprod{y}{\Omega}=0$, and hence \Cref{eq:near-perfect-hypothesis-t} would imply
    \[ \abs*{ \innerprod{x}{\Omega} - \innerprod{y}{\Omega} } = \innerprod{x}{\Omega} \geq \frac{\tau_t}{2^{r/2}}, \]
    which contradicts \Cref{eq:near-perfect-difference-bound}, as desired. 

    To establish the induction hypothesis at round $t+1$, we write $y=(a',b',0^{m_0},u')$. Then, combining \Cref{eq:near-perfect-difference-bound,eq:near-perfect-hypothesis-t} and the triangle inequality, it follows that:
    \begin{equation}
        \label{eq:near-perfect-overlapNew}
        \sqrt{\frac{p_1(a')p_2(b')}{2^r}} = \innerprod{y}{\Omega} \geq \innerprod{x}{\Omega} - \abs*{\innerprod{x}{\Omega} - \innerprod{y}{\Omega}} \geq \frac{\tau_t}{2^{r/2}} - \frac{\tau_t}{2\cdot 2^{r/2}} = \frac{\tau_t}{2\cdot 2^{r/2}}. 
    \end{equation}
    Since $p_1(a') \leq 1$ and $p_2(b')\leq 1$ hold for all $a',b'\in\binset^\ell$, squaring \Cref{eq:near-perfect-overlapNew} gives:
    \[  p_1(a') \geq \frac{\tau_t^2}{4} \geq \frac{\tau_t^2}{16} = \tau_{t+1} \quad\text{and}\quad p_2(b') \geq \frac{\tau_t^2}{4} \geq \frac{\tau_t^2}{16} = \tau_{t+1}. \]
    Therefore, we establish the induction hypothesis by concluding that all newly inserted coordinates are heavy at threshold $\tau_{t+1}$, and no bad transition occurs in this round. 
\end{proof}

\paragraph{Soundness.} We now prove the soundness part:

\begin{lemma}[Rectangle size growth for \emph{no} instances before encountering a bad round]
    \label{lemma:rectangle-closure-testing-soundness}
    Let $(\Gamma,I)$ be a \emph{no} instance of $\SepRCD_{\alpha,1-\gamma}(\ell,m_0,r)$ for some $\gamma>0$. Then, for every starting pair and every round $t$, one of the following holds: either round $t$ is bad, or the rectangle size grows by a multiplicative factor of $1+\gamma$: 
    \begin{equation}
        \label{eq:rect-size-expansion}
        \abs*{S_{t+1}\times T_{t+1}} \geq (1+\gamma) \abs*{S_t \times T_t}.
    \end{equation}
    Consequently, no starting pair can avoid encountering a bad round before round 
    \[ L \coloneqq \ceil*{\frac{2\ell\ln{2}+1}{\ln(1+\gamma)}}. \]
\end{lemma}

\begin{proof}
Fix a starting pair and a round $t$. 
Assume that round $t$ is not bad. It therefore suffices to prove that \Cref{eq:rect-size-expansion}. To do so, we show that the rectangle size grows exponentially even when one restricts attention to product subset state witnesses. 

Consider a valid product subset witness state $\ket{S_t}\ket{T_t}$, whose initialized full state is uniform on $K(S_t,T_t) = S_t\times T_t\times \cbra{0^{m_0}}\times U$. A direct calculation shows that
\begin{subequations}
\label{eq:rect-size-growth-SepRCDval}
\begin{align}
    &\bra{K(S_t,T_t)} \Gamma \ket{K(S_t,T_t)}\\
    =~& \frac{1}{|K(S_t,T_t)|} \abs*{\set{y\in K(S_t,T_t)}{\Gamma y \in K(S_t,T_t)}}\\
    =~& \frac{1}{ 2^r|S_t\times T_t| } \abs*{\set{(a,b,u)\in S_t\times T_t\times U}{\GammaZero(a,b,u) \in S_t\times T_t}}. 
\end{align}
\end{subequations}
By combining the soundness promise $\bra{K(S_t,T_t)} \Gamma \ket{K(S_t,T_t)} \leq 1-\gamma$ with \Cref{eq:rect-size-growth-SepRCDval}, we obtain a lower bound on the number of triples escaping from $S_t\times T_t\times U$: 
\begin{equation}
    \label{eq:rect-size-escaping-bound}
    \mathrm{Esc}(S_t,T_t,U)\coloneqq \abs*{\set{(a,b,u)\in S_t\times T_t\times U}{\GammaZero(a,b,u) \notin S_t\times T_t}} \geq \gamma 2^r \abs*{S_t\times T_t}. 
\end{equation}

Because $\Gamma$ is injective, these leaving triples $(a,b,u)$ have distinct outputs $(a',b')=\GammaZero(a,b,u)$. For each output pair $(a',b')$, there are at most $|U|=2^r$ possible choices of output random strings $u'$.
Consequently, the number of distinct new output pairs outside $S_t\times T_t$ is at least $\mathrm{Esc}(S_t,T_t,U) / |U|$. Since all these new output pairs  are included in $S_{t+1}\times T_{t+1}$, it follows that
\begin{align*}
    \abs*{S_{t+1}\times T_{t+1}} &\geq \abs*{S_{t}\times T_{t}} + \frac{\mathrm{Esc}(S_t,T_t,U)}{|U|}\\
    &\geq \abs*{S_{t}\times T_{t}} + \frac{\gamma 2^r \abs*{S_t\times T_t}}{2^r} = (1+\gamma) \abs*{S_t\times T_t}.
\end{align*}
Here, the second line uses \Cref{eq:rect-size-escaping-bound}, thereby establishing \Cref{eq:rect-size-expansion}. 

If a starting pair were to reach round $L$ without encountering a bad round, then combining \Cref{eq:rect-size-expansion} with $\abs*{S_0\times T_0} = 1$ would yield
\begin{equation}
    \label{eq:rect-size-contradiction}
    \abs*{S_L \times T_L} \geq (1+\gamma)^L \abs*{S_0\times T_0}
    = (1+\gamma)^L > 2^{2\ell}.
\end{equation}
Here, the last inequality follows from the definition of $L$. Since $S_L\times T_L \subseteq \binset^\ell \times \binset^\ell$, \Cref{eq:rect-size-contradiction} leads to a contradiction. Therefore, no starting pair can reach round $L$ without encountering a bad round, as desired. 
\end{proof}

\paragraph{Putting everything together.} Finally, we show the correctness of \algoref{algo:rectangle-closure-testing}:

\begin{proof}[Proof of \Cref{thm:rectangular-closure-testing}.]
    We first prove correctness separately for \emph{yes} and \emph{no} instances:
    \begin{itemize}
        \item For \emph{yes} instances, $(\Gamma,I)$ is an $\epsilon$-near perfectly agreeing pair, so there exists a non-negative product witness state whose initialized full state $\ket{\Omega}$ satisfies $\bra{\Omega}\Gamma\ket{\Omega} \geq 1-\epsilon$. By \Cref{lemma:rectangle-closure-testing-nearly-perfect-completeness}, the error bound $\epsilon \leq \tau^2_L/2^{r+3}$ implies that the rectangular closure starting from some starting pair reaches round $L$ without encountering a bad round. Therefore, \algoref{algo:rectangle-closure-testing} accepts in this case. 
        \item For \emph{no} instances, because $(\Gamma,I)$ is a \emph{no} instance, it follows that $\bra{\Omega_\psi}\Gamma\ket{\Omega_\psi} \leq 1-\gamma$ for all non-negative product witness states $\ket{\psi}$. Using \Cref{lemma:rectangle-closure-testing-soundness}, no starting pair can reach round $L$ without encountering a bad round. Hence, \algoref{algo:rectangle-closure-testing} rejects. 
    \end{itemize}
    
    Next, we analyze two deterministic implementations of \algoref{algo:rectangle-closure-testing}: 
    
    \parheadingItem{Recursive implementation.} For every fixed starting pair $(a_0,b_0)$, define indicator predicates
    \[ \forall t\in\cbra{0,1,\cdots,L}, \qquad \chi_{S_t}(a)\coloneqq \bbI[a\in S_t] \quad\text{and}\quad \chi_{T_t}(b)\coloneqq \bbI[b\in T_t]. \]
    Initially, it holds that $\chi_{S_0}(a)=\bbI[a=a_0]$ and $\chi_{T_0}(b)=\bbI[b=b_0]$. 

    For each round $t< L$, we first test whether round $t$ is bad by evaluating the predicate
    \[ \mathrm{Bad}_t \coloneqq \exists a,b,u \ \bbI[\chi_{S_t}(a)\wedge \chi_{T_t}(b)\wedge \Gamma(a,b,0^{m_0},u)\notin A\times B\times\{0^{m_0}\}\times U]. \]

    If $\mathrm{Bad}_t$ is false, then the next-round predicates are simultaneously given by
    \begin{align*}
        \chi_{S_{t+1}}(a') &= \chi_{S_t}(a') \ \vee \ \exists a,b,b',u \ \bbI\sbra*{\chi_{S_t}(a)\wedge \chi_{T_t}(b)\wedge \GammaZero(a,b,u)=(a',b')},\\
        \chi_{T_{t+1}}(b') &= \chi_{T_t}(b') \ \vee \ \exists a,a',b,u \ \bbI\sbra*{\chi_{S_t}(a)\wedge \chi_{T_t}(b)\wedge \GammaZero(a,b,u)=(a',b')}.
    \end{align*}

    We use a depth-first deterministic evaluation. For each pair $(a,b)$, first evaluate the transition condition by scanning $u\in U$; if it holds, evaluate $\chi_{S_t}(a)$ and $\chi_{T_t}(b)$ sequentially. This reordering preserves the predicates above because both membership predicates are independent of $u$. Specifically, let $Q_t$ denote the maximum time for one membership query at round $t$. Then
    \[ Q_0=O(\ell) \quad\text{and}\quad Q_{t+1} \leq \rbra{1+2^{2\ell+1}}Q_t+O\rbra[\big]{2^{2\ell+r}\rbra{T_\Gamma+\log(L+2)}}. \] 
    The same evaluation order is used when checking for bad transitions. Across all rounds and starting pairs, the total time is $2^{O(\ell L)+r}\poly(L)T_\Gamma$.

    To bound the space complexity, we keep only the Boolean answer after each scan, allowing the counter and circuit-evaluation workspace to be reused across recursive calls. At each recursion level, we store the current pair, the target coordinate, the level number, and a constant amount of control information, using $O\rbra{\ell+\log(L+2)}$ bits. Since the recursion depth is at most $L$, the total space is $O\rbra[\big]{L\rbra{\ell+\log(L+2)}+\lambda+S_\Gamma}$.

    \parheadingItem{Explicit-table implementation.} For each starting pair, \algoref{algo:rectangle-closure-testing} stores $S_t$ and $T_t$ alongside the newly added sets $\Delta S_t\coloneqq S_t\setminus S_{t-1}$ and $\Delta T_t\coloneqq T_t\setminus T_{t-1}$, as $2^\ell$-bit tables, with the convention $S_{-1}=T_{-1}=\emptyset$.  This implementation is still \emph{deterministic}. 
    During round $t$, instead of scanning all triples in $S_t\times T_t\times U$, it suffices to scan, without duplication, only those triples in 
    \[ \rbra[\big]{ \rbra*{ \Delta S_t\times T_t } \cup \rbra*{ S_t\times \Delta T_t } } \times U.\] 
    Indeed, every pair outside $\rbra*{ \Delta S_t\times T_t } \cup \rbra*{ S_t\times \Delta T_t }$ already belonged to $S_{t-1}\times T_{t-1}$, and hence had already been scanned in an earlier round. If that earlier scan had found a bad transition, the algorithm would already have rejected; otherwise, the corresponding output coordinates had already been inserted into the current tables, so \emph{rescanning such a pair cannot create any new element}. Hence, for a fixed starting pair, each active pair in $A\times B$ is scanned \emph{at most once}.

    For one starting pair, there are at most $2^{2\ell}$ active pairs, and scanning one active pair requires enumerating all $u\in U$ and evaluating $\Gamma$, which costs $O(2^r T_\Gamma)$. Consequently, one starting pair takes time $O(2^{2\ell+r}T_\Gamma)$. Since there are $\abs{A\times B}=2^{2\ell}$ starting pairs, the total time is $O(2^{4\ell+r}T_\Gamma)$. The space consists of a constant number of bit tables for $S_t$, $T_t$, $\Delta S_t$, and $\Delta T_t$, the current basis strings and counters, and the workspace for evaluating $\Gamma$, namely $O(2^\ell+\lambda+S_\Gamma)$. 

    \vspace{1em}
    Finally, by the inequality $\gamma/2 \leq \ln(1+\gamma) \leq \gamma$ for $\gamma \in (0,1)$, we simplify the round bound to 
    \[ L = \ceil*{\frac{2\ell\ln{2}+1}{\ln(1+\gamma)}} = O\rbra*{\frac{\ell}{\gamma}}. \qedhere\]
\end{proof}

%%%%%%%%%%%%%%%%%%%%%%%%%%%%%%%%%%%%%%%%%%%%%%%%%%%%%%%%%%%
%%%%%%%%%%%%%%%%%%%%%%%%%%%%%%%%%%%%%%%%%%%%%%%%%%%%%%%%%%%
%%%%%%%%%%%%%%%%%%%%%%%%%%%%%%%%%%%%%%%%%%%%%%%%%%%%%%%%%%%

\section{A deterministic exponential-time upper bound for \StoqMAk{}}
\label{sec:sos}
In this section, we record an upper bound for \StoqMAk{}, which follows from the Sum-of-Squares algorithm of Barak--Kelner--Steurer~\cite{BKS14}. 

\begin{restatable}[{\cite[Theorem 3.1]{BKS14}} with improved $t$-dependence]{theorem}{BKSsSoS}
\label{thm:bks-sos}
For any entrywise non-negative, real-symmetric matrix $M$ acting on $(\mathbb C^d)^{\otimes t}$ with $\norm{M}\le 1$, there is a degree-$O(t^2 \log d/\epsilon^2)$ Sum-of-Squares algorithm that approximates
\begin{equation*}
  \max_{x\in \bbR^d:\|x\|=1}  \Tr M (x x^T)^{\otimes t}
\end{equation*}
to within additive error $\epsilon$.   
\end{restatable}

For completeness, we describe the BKS algorithm in the appendix. In addition, we slightly sharpen BKS's analysis, improving the dependence on the tensor order $t$ from $O(t^3 \log{d}/\epsilon^2)$, as stated in~\cite[Theorem 3.1]{BKS14}, to the \emph{optimal} dependence $O(t^2 \log d/\epsilon^2)$ under ETH.\footnotemark
\footnotetext{The symmetry of the tensors is not important for this algorithm. The same algorithm works for asymmetric tensors. But for our analysis to obtain $O(t^2)$ dependence, symmetry is crucial.}

Since the matrix $M$ arising from a stoquastic verification circuit is entrywise non-negative and symmetric, we have the following upper bound:
\begin{theorem}[$\StoqMAk\subseteq \EXP$]
\label{thm:StoqMAk-in-EXP}
For all efficiently computable functions $c(n)$ and $\Delta(n)$ with $1/2<c(n)\leq 1$ and $c(n)-\Delta(n) \geq 1/2$, and for every efficiently computable positive integer-valued function $k(n)$ and $\ell(n)$, the following inclusions hold:
    \begin{align}
        \SymStoqMA_\ell\ssbra*{k,c,\Delta}
        \subseteq\DTIME\sbra[\Big]{\exp\Bigl(O\Bigl(\frac{k^2\ell^2}{\Delta^2}\Bigr)\Bigr)};
        \label{eq:SymStoqMAk-DTIME-upper}\\ 
        \StoqMA_\ell\ssbra*{k,c,\Delta} \subseteq \DTIME\sbra[\Big]{\exp\Bigl(\widetilde O\Bigl(\frac{k^2\ell^2}{\Delta^2}\Bigr)\Bigr)}. \label{eq:StoqMAk-DTIME-upper}
    \end{align} 
In particular, for every efficiently computable integer-valued function $k$ that is polynomially bounded in $n$, the following inclusion holds: 
    \[
        \StoqMAk \subseteq \EXP. 
    \]
\end{theorem}
Here, \Cref{eq:SymStoqMAk-DTIME-upper} is an application of \Cref{thm:bks-sos} with $t=k, d=2^\ell$, and $\epsilon=\Delta/3$, which gives rise to a $\exp(O(k^2\ell^2/\Delta^2))$-size SDP; \Cref{eq:StoqMAk-DTIME-upper} is a straightforward corollary of \Cref{eq:SymStoqMAk-DTIME-upper} and \Cref{thm:StoqMAk-in-SymStoqMAk}. % 
Notably, the deterministic time bounds in \Cref{eq:SymStoqMAk-DTIME-upper,eq:StoqMAk-DTIME-upper} have an implicit \emph{exponential dependence} on the number of $\ket{+}$ ancillary qubits, arising from transforming a stoquastic verification circuit $V_x$ into the matrix $M_\Gamma$:
\begin{remark}[Implicit exponential dependence on the number of $\ket{+}$ ancillary qubits]
    \label{remark:m+-dependence-in-BKS}
    Let $m_0$ and $m_+$ be the numbers of $\ket{0}$ and $\ket{+}$ ancillary qubits. After contracting the initialized ancillary registers, the Sum-of-Squares algorithm in \Cref{thm:bks-sos} is applied to the matrix 
    \[ M_\Gamma = \bra{0}^{\otimes m_0}\bra{+}^{\otimes m_+} \Gamma \ket{0}^{\otimes m_0}\ket{+}^{\otimes m_+}, \quad\text{where } \Gamma \coloneqq V_x^\dagger X_{\sf O}V_x. \] 
    Consequently, the resulting tensor order and local dimension are $k$ and $2^\ell$, respectively, so the degree and size of the resulting SDP have no explicit dependence on $m_+$. Nevertheless, in the succinct-circuit representation, the coefficients of $M_\Gamma$ satisfy
    \[
        \bra{b}M_\Gamma\ket{a} =
        2^{-m_+} \abs*{ \set{u \in \binset^{m_+}}{\Gamma(a,0^{m_0},u) \in \cbra{b}\times\cbra{0^{m_0}}\times\binset^{m_+}} }.
    \]
    Therefore, the $m_+$-independence applies to the SDP-size bound once $M_\Gamma$ is available, whereas deterministically constructing its coefficients directly from $V_x$ incurs a $2^{O(m_+)}$ preprocessing cost unless an efficient coefficient-access procedure is assumed.\footnote{More precisely, a single scan over $(a,u)\in\binset^{k\ell}\times\binset^{m_+}$ constructs all coefficients in deterministic time $O\rbra*{2^{k\ell+m_+}T_\Gamma+2^{2k\ell}(m_++1)}$, where $T_\Gamma$ is the circuit-evaluation time. Adding the SDP cost gives a total running time of $\exp\rbra[\big]{\widetilde O\rbra*{k^2\ell^2/\Delta^2}}
    +2^{m_++\widetilde O(k\ell)}$, up to polynomial factors in the circuit size and $1/\Delta$.}
\end{remark}

\paragraph{ETH-based optimality.}
A remarkable consequence is the essential optimality under ETH of the above upper bound result, and both lower bounds~\Cref{thm:power-stoqma-np,thm:power-stoqmak-np} obtained in~\Cref{sec:stoqma-power}.
Since Dinur's PCP (\Cref{thm:dinur-gapcg-form}) can be reduced from $\SAT$ instances in near-linear time, the two $\StoqMA$ protocols for Dinur's PCP (\Cref{thm:power-stoqma-np,thm:power-stoqmak-np}) together with \Cref{eq:SymStoqMAk-DTIME-upper,eq:StoqMAk-DTIME-upper} both imply a $\exp(\widetilde{O}(n))$ algorithm for $\SAT$. 

\begin{restatable}[ETH-based optimality of~\Cref{thm:StoqMAk-in-EXP,thm:power-stoqma-np,thm:bks-sos}]{corollary}{ETHopt}
    Under ETH, up to lower-order factors, the following statements hold: 
    \begin{enumerate}[label=\upshape(\arabic*)]
        \item\label{thmitem:eth-opt-t} \Cref{eq:SymStoqMAk-DTIME-upper} in~\Cref{thm:StoqMAk-in-EXP} is tight with respect to the number of provers $k$ in view of~\Cref{thm:power-stoqmak-np}; consequently, the BKS algorithm in~\Cref{thm:bks-sos} is tight with respect to the order of tensor $t$.
        \item\label{thmitem:eht-opt-l} \Cref{eq:SymStoqMAk-DTIME-upper} in~\Cref{thm:StoqMAk-in-EXP} is tight with respect to the proof length $\ell$ in view of~\Cref{thm:power-stoqma-np}; consequently, the BKS algorithm in~\Cref{thm:bks-sos} is tight with respect to the log dimension $\log d$.
    \end{enumerate}
    The same holds for \Cref{eq:StoqMAk-DTIME-upper}. 
\end{restatable}
\begin{proof}
By Dinur's PCP~\Cref{thm:dinur-gapcg-form} and our protocols~\Cref{thm:power-stoqma-np,thm:power-stoqmak-np}, we have for some $\epsilon >0$ and $\Delta>0$,
\begin{align}
    \SAT&\in 
        \SymStoqMA_{O(\log n)} \ssbra*{O(\sqrt n), 1-\epsilon,\frac{1}{2}-2\epsilon};
    \label{eq:eth-opt-k}
    \\
    \SAT &\in 
    \StoqMA_{\widetilde{O}(\sqrt n)}\ssbra*{2,1-\epsilon,\frac12-2\epsilon}
    \subseteq \SymStoqMA_{\widetilde{O}(\sqrt n)} \ssbra*{2,1-\epsilon,\Delta}.
    \label{eq:eth-opt-l}
\end{align}

We first prove \Cref{thmitem:eth-opt-t}. Applying \Cref{thm:StoqMAk-in-EXP} to \Cref{eq:eth-opt-k} with $k= O(\sqrt n)$ and $\ell=O(\log n)$, we obtain $\SAT\in\DTIME(\exp(\widetilde O(n)))$, matching the ETH barrier up to lower-order factors. Therefore for any constant $\eta>0$, an improvement of the $k^2$ dependence in \Cref{eq:SymStoqMAk-DTIME-upper} to $k^{2-\eta}$, or an improvement in~\Cref{thm:power-stoqmak-np} with $n^{1/2-\eta}$ provers, $O(\log n)$-length proofs and a constant gap would imply a subexponential-time algorithm for $\SAT$.

Now we show \Cref{thmitem:eht-opt-l}. Similarly, applying \Cref{thm:StoqMAk-in-EXP} to \Cref{eq:eth-opt-l} with $k=2$ and $\ell=\widetilde O(\sqrt n)$ also saturates ETH. Therefore, for any $\eta>0$, an improvement in \Cref{eq:SymStoqMAk-DTIME-upper} with order $\exp(\ell^{2-\eta})$, or an improvement in~\Cref{thm:power-stoqma-np} with two provers, sending proofs of length $n^{1/2-\eta}$ and a constant gap would also imply a subexponential-time algorithm for $\SAT$. 

A completely analogous argument works for~\Cref{eq:StoqMAk-DTIME-upper}, since the simulations between $\StoqMAk$ and $\SymStoqMAk$ (\Cref{thm:StoqMAk-in-SymStoqMAk,thm:sym-to-stoq-k}) are length-efficient, prover-preserving, and constant-gap-preserving.
\end{proof}

%%%%%%%%%%%%%%%%%%%%%%%%%%%%%%%%%%%%%%%%%%%%%%%%%%%%%%%%%%%
%%%%%%%%%%%%%%%%%%%%%%%%%%%%%%%%%%%%%%%%%%%%%%%%%%%%%%%%%%%
%%%%%%%%%%%%%%%%%%%%%%%%%%%%%%%%%%%%%%%%%%%%%%%%%%%%%%%%%%%

\section*{Acknowledgments}
The authors are grateful to Thomas Vidick for insightful discussions and helpful comments on an earlier version of our manuscript, which helped improve its readability and led to Question \ref{open:QCMA-lower-bound}. The authors also thank Aviv Taller for asking about the computational power of \StoqMAk{} with logarithmic-size witnesses, which eventually led to \Cref{sec:StoqMAk-log-upper-bound}.
YL was supported by funding from the Swiss State Secretariat for Education, Research and Innovation (SERI). 

%%%%%%%%%%%%%%%%%%%%%%%%%%%%%%%%%%%%%%%%%%%%%%%%%%%%%%%%%%%
%%%%%%%%%%%%%%%%%%%%%%%%%%%%%%%%%%%%%%%%%%%%%%%%%%%%%%%%%%%
%%%%%%%%%%%%%%%%%%%%%%%%%%%%%%%%%%%%%%%%%%%%%%%%%%%%%%%%%%%

\section*{AI use disclosure}
\noindent
Large language model tools were used interactively to explore possible approaches to research problems in the main technical sections, particularly those concerning upper bounds, identify relevant references, and proofread the manuscript. All content, including mathematical statements and proofs, was written or substantially revised by the authors, who remain solely responsible for the final manuscript.

%%%%%%%%%%%%%%%%%%%%%%%%%%%%%%%%%%%%%%%%%%%%%%%%%%%%%%%%%%%
%%%%%%%%%%%%%%%%%%%%%%%%%%%%%%%%%%%%%%%%%%%%%%%%%%%%%%%%%%%
%%%%%%%%%%%%%%%%%%%%%%%%%%%%%%%%%%%%%%%%%%%%%%%%%%%%%%%%%%%

% Reference
\bibliographystyle{alphaurlQ}
\bibliography{StoqMA(2)}
\appendix
% Fix appendices label
\crefalias{section}{appendix}
\crefalias{subsection}{appendix}
\crefalias{subsubsection}{appendix}

%%%%%%%%%%%%%%%%%%%%%%%%%%%%%%%%%%%%%%%%%%%%%%%%%%%%%%%%%%%
%%%%%%%%%%%%%%%%%%%%%%%%%%%%%%%%%%%%%%%%%%%%%%%%%%%%%%%%%%%
%%%%%%%%%%%%%%%%%%%%%%%%%%%%%%%%%%%%%%%%%%%%%%%%%%%%%%%%%%%

\section{\texorpdfstring{$\PSPACE \subseteq \PreciseStoqMA_1$}{}}
\label{sec:cleanCC-in-PreciseStoqMA}

In this section, we establish the equivalence between $\PreciseStoqMA_1$, \PreciseStoqMA{}, and \PSPACE{}, where $\PreciseStoqMA_1$ denotes \PreciseStoqMA{} with perfect completeness: 

\begin{theorem}[\PreciseStoqMA{} with perfect completeness contains \PSPACE{}]
    \label{thm:PreciseStoqMA-perfect-completeness}
    There exists an explicit efficiently computable function $s_\star(n) = 1-2^{-\poly(n)}$ such that
    \[ \PSPACE \subseteq \PreciseStoqMA(1,s_\star). \]
\end{theorem}

Then, combining \Cref{thm:PreciseStoqMA-perfect-completeness} with $\PreciseQMA\subseteq\PSPACE$~\cite{FL16,FL18} immediately yields the following equivalence, thereby giving a positive answer to and fully resolving the open problem in~\cite[Section 1.2]{Liu21} asking whether \PreciseStoqMA{} captures the full \PSPACE{} power:
\begin{corollary}
    $\PreciseStoqMA_1 = \PreciseStoqMA = \PSPACE$. 
\end{corollary}

\vspace{1em}
To establish \Cref{thm:PreciseStoqMA-perfect-completeness}, we begin with the \PSPACE{}-complete problem \textsc{Clean Connected Component}, introduced in~\cite{AG20}: 

\begin{definition}[Clean Connected Component, \CleanCC{}, adapted from~{\cite[Definition 3.1]{AG20}}]
\label{def:cleanCC}
Let $G=(V,E)$ be a graph with vertex set $V = \binset^n$, and suppose that every vertex $v\in V$ has degree $d(v) \leq d_G\leq \poly(n)$. We say that a pair of polynomial-size classical circuits $(C_G,C_M)$, where  $C_G \colon V\times\ssbra{d_G} \to V$ and $C_M \colon V \to \binset$, is an instance of $\CleanCC(n,d_G)$ if the following conditions hold for every $v\in V$:
\begin{enumerate}[label={\upshape(\arabic*)}]
    \item For each $j\in\ssbra{d_G}$, $C_G(v,j)$ outputs the $j$-th neighbor of $v$. Furthermore, every neighbor of $v$ appears \emph{exactly once} in the list $C_G(v,0),\cdots,C_G\rbra*{v,d_G-1}$, and every remaining slot is padded with the self-loop value $v$. 
    \item $C_M(v)$ outputs $1$ if the vertex $v$ is \emph{marked}; otherwise, it outputs $0$, and $v$ is \emph{unmarked}. 
\end{enumerate}
The problem is to decide which of the following holds: 
\begin{itemize}
    \item \emph{\textbf{Yes.}} There exists a non-empty connected component $H_\star\subseteq V$ such that all vertices in $H_\star$ are \emph{unmarked}. 
    \item \emph{\textbf{No.}} Every connected component $H \subseteq V$ contains at least one \emph{marked} vertex. 
\end{itemize}
\end{definition}

\begin{lemma}[Adapted from~{\cite[Appendix D]{AG20}}]
    \label{lemma:CleanCC-PSPACE-hard}
    \CleanCC{} is \PSPACE{}-complete. 
\end{lemma}

Now we are ready to state the main technical result in this section, which, together with  \Cref{lemma:CleanCC-PSPACE-hard}, establishes \Cref{thm:PreciseStoqMA-perfect-completeness}:
\begin{theorem}[$\CleanCC\in\PreciseStoqMA_1$] 
    \label{thm:cleanCC-in-StoqMA}
    For any efficiently computable function $d_G(n)$ bounded by a \emph{polynomial} in $n$, the following holds:
    \[\CleanCC\rbra*{n,d_G} \in \StoqMA\rbra*{1,1-\frac{1}{2^{2n+2}(d_G+1)}}.\] 
\end{theorem}

\subsection{Constructing the stoquastic verification circuit}
\label{subsec:cleancc-in-stoqma}

We first make each move from a vertex $u \in V$ to its $j$-th neighbor $C_G(u,j)$ reversible by introducing the return-index circuit $\wCG$, whose efficient construction is given below:

\begin{proposition}[Return-index circuit can be efficiently constructed]
    \label{prop:return-index-circuit}
    Given the circuit $C_G$ that describes the graph $G=(V,E)$ with $|V|=2^n$ and the maximum degree $d_G \leq \poly(n)$, one can efficiently construct a classical circuit $\wCG \colon V \times \ssbra{d_G} \to \ssbra{d_G}$ satisfying the following:
    \begin{itemize}
        \item For every edge $(u,v)\in E$ and every neighbor index $j$ with $C_G(u,j)=v$, define $\wCG(u,j)$ to be the unique \emph{return index} at $v$, namely the unique index satisfying $C_G\rbra[\big]{v, \wCG(u,j) }=u$. 
        \item For every padded self-loop port with $C_G(u,j)=u$, define $\wCG(u,j)=j$. 
    \end{itemize}
\end{proposition}

\begin{proof}
    For every edge $(u,v)\in E$ such that $C_G(u,j)=v$, $\wCG(u,j)$ can be implemented efficiently by simply checking all $j'\in\ssbra{d_G}$ and finding the unique index $j'_\star$ satisfying $C_G(v,j'_\star)=u$. This construction is efficient because the maximum degree $d_G \leq \poly(n)$.
\end{proof}

Next, we introduce the classical reversible circuit $\Gamma_G$, built from the classical circuits $C_G$ and $\wCG$ and made reversible by using additional $\ket{0}$ ancillary qubits, as defined by
\begin{equation}
    \label{eq:CG-reversible}
    \forall j\in\ssbra{d_G}, \; \forall v\in V, \quad \Gamma_G \ket{j}\ket{v}\ket{\bar{0}} = \ket{\wCG(v,j)}\ket{C_G(v,j)}\ket{\bar{0}}.
\end{equation}
By construction, it follows that $\Gamma_G^2\ket{j}\ket{v}\ket{\bar{0}} = \ket{j}\ket{v}\ket{\bar{0}}$. Similarly, we define the classical reversible circuit $\Gamma_M$, built from the classical circuit $C_M$, by
\begin{equation}
    \label{eq:CM-reversible}
    \forall v\in V,\; \forall c\in\binset, \quad \Gamma_M \ket{v}\ket{c}\ket{\bar{0}} = \ket{v}\ket{c \oplus C_M(v)}\ket{\bar{0}}. 
\end{equation}
We also note that the reversible circuits $\Gamma_G$ and $\Gamma_M$ are both efficiently implementable.\footnote{Since the classical circuits $C_G$, $\wCG$, and $C_M$ are all efficiently implementable, it follows that $\Gamma_G$ and $\Gamma_M$ are also efficiently implementable using the standard compute-swap-uncompute trick (cf.~\cite[Theorem 7.3]{KSV02}).}

With these reversible circuits in place, we present the stoquastic verification circuit for \CleanCC{} used in the proof of \Cref{thm:cleanCC-in-StoqMA}, as given in \Cref{protocol:CleanCC-in-StoqMA}. 

\LinesNotNumbered
\begin{algorithm}[ht!]
    \UseProtocolCounter
    \caption{A stoquastic verification circuit $V_{C_G,C_M}$ for \CleanCC{}.}
    \label{protocol:CleanCC-in-StoqMA}
    \SetEndCharOfAlgoLine{.}
    \setlength{\parskip}{5pt}
    \SetKwFor{For}{For}{:}{}
    \SetKwIF{If}{ElseIf}{Else}{If}{:}{elif}{Else:}{}%
    \SetKwInOut{Input}{Input}
    \SetKwInOut{Output}{Output}
    \SetKwInOut{Registers}{Registers}
    \SetKwInOut{Parameter}{Parameters}

    \Input{A $\CleanCC(n,d_G)$ instance $(C_G,C_M)$.}
    \Parameter{$q\coloneqq \ceil*{\log(d_G+1)}$ and $J\coloneqq 2^q$.}
    \Registers{
    \begin{minipage}[t]{0.7\linewidth}
        \hangindent=3.2em
        \hangafter=1
        $\sfB$: the $q$-qubit branch-index register\; 
        $\sfV$: the $n$-qubit vertex-label register\;
        $\sfM$: the single-qubit marking register\;
        $\sfA$: the register containing $\ket{0}$ ancillary qubits.
    \end{minipage}
    }
    
    \textbf{1.} Construct the classical reversible circuits $\Gamma_G$ and $\Gamma_M$, as expressed in \Cref{eq:CG-reversible,eq:CM-reversible}, using \Cref{prop:return-index-circuit}\;

    \medskip
    \textbf{2.} Define a classical reversible circuit $\Gamma$ on the quantum registers $(\sfB,\sfV,\sfM,\sfA)$ as follows:
    \begin{enumerate}[label={\upshape(\arabic*)}, itemsep=0.2em, topsep=0.2em, parsep=0.2em]
        \item Apply $\Gamma_G$ to the registers $(\sfB,\sfV,\sfA)$ when $0\leq j < d_G$; 
        \item Apply $\Gamma_M$ to the registers $(\sfV,\sfM,\sfA)$ when $j=d_G$;
        \item Do nothing otherwise.
    \end{enumerate}
    Equivalently, $\Gamma$ can be expressed as
    \begin{equation}
        \label{eq:Gamma-CleanCC}
        \Gamma \ket{j}_\sfB \ket{v}_\sfV \ket{c}_\sfM \ket{\bar{0}}_\sfA
        =
        \begin{cases}
            \ket{\wCG(v,j)}_\sfB \ket{C_G(v,j)}_\sfV \ket{c}_\sfM \ket{\bar{0}}_\sfA , & 0\leq j<d_G,\\
            \ket{j}_\sfB \ket{v}_\sfV \ket{c\oplus C_M(v)}_\sfM \ket{\bar{0}}_\sfA, & j=d_G,\\
            \ket{j}_\sfB \ket{v}_\sfV \ket{c}_\sfM \ket{\bar{0}}_\sfA, & d_G< j < J.
        \end{cases}
    \end{equation}
    \textbf{3.} Apply the branch-overlap test (\Cref{lemma:branch-overlap-test}) to the reversible circuit pair $(\Gamma,I)$ with the input state stored in the registers $(\sfB,\sfV,\sfM)$.
\end{algorithm}

\subsection{Analysis of \texorpdfstring{\Cref{protocol:CleanCC-in-StoqMA}}{}}

In this subsection, we prove \Cref{thm:cleanCC-in-StoqMA}. 
Before presenting the proof, we first specify the acceptance probability of the stoquastic verification circuit in \Cref{protocol:CleanCC-in-StoqMA}. 
For a non-negative witness $\ket\psi=\sum_{v\in V}\alpha_v\ket v$, where $\alpha_v\geq 0$ for each $v\in V$ and $\sum_{v\in V}\alpha_v^2=1$, let $\ket{\Omega_\psi}$ denote the initialized state corresponding to $\ket{\psi}$. Then \Cref{lemma:SepRCD-StoqMAk-complete} guarantees that
\begin{equation}
    \label{eq:cleanCC-StoqMA-pacc}
    \Pr\sbra*{V_{C_G,C_M} \text{ accepts } \ket{\psi}} = \frac{1}{2} + \frac{1}{2} \bra{\Omega_\psi} \Gamma \ket{\Omega_\psi}, \quad\text{where } \ket{\Omega_\psi} \coloneqq \frac{1}{\sqrt J}\sum_{j=0}^{J-1}\ket{j}_\sfB\ket{\psi}_\sfV\ket{0}_\sfM\ket{\bar{0}}_\sfA.
\end{equation}

We now proceed with the proof of \Cref{thm:cleanCC-in-StoqMA}. At a high level, for \emph{yes} instances, the subset state $\ket{H_\star}$ serves as the witness state. For \emph{no} instances, however, the analysis is more delicate, since it must apply to \emph{all} non-negative witness states. The key observation is that, in any connected component containing a marked vertex, large witness mass must create either noticeable mass on a marked vertex or noticeable amplitude differences across edges. This argument is a path form of the Poincar\'e principle (cf.~\cite[Section 5.2]{Jerrum03}): local edge differences, together with marked-vertex mass, control the total mass in the component.

\begin{proof}[Proof of \Cref{thm:cleanCC-in-StoqMA}]

    We begin by expanding $\bra{\Omega_\psi} \Gamma \ket{\Omega_\psi}$ according to the branch-index register $\sfB$. For the first $d_G$ neighbor-index branches, we obtain the overlap $\sum_{v} \alpha_v\alpha_{C_G(v,j)}$. For the marking branch $j=d_G$, the initialized bit $c=0$ remains in the initialized sector exactly on unmarked vertices, contributing $\sum_{v:C_M(v)=0}\alpha_v^2$. The remaining $J-d_G-1$ branches are identities and each contributes $\sum_{v}\alpha_v^2=1$. Consequently, we conclude the following identity:
    \begin{equation}
        \label{eq:cleanCC-StoqMA-pacc-decomp}
        \bra{\Omega_\psi} \Gamma \ket{\Omega_\psi} = \frac{1}{J} \rbra*{ \sum_{j\in \ssbra{d_G}}\sum_{v\in V}\alpha_v\alpha_{C_G(v,j)} +  \sum_{v:C_M(v)=0}\alpha_v^2 + (J-d_G-1) }.
    \end{equation}

    \paragraph{Completeness.} For \emph{yes} instances, let the witness state be the subset state 
    \[\ket{H_\star} = \frac{1}{\sqrt{\abs{H_\star}}} \sum_{v\in H_\star} \ket{v},\] 
    where $H_\star$ is an unmarked connected component. Then $\alpha_v = 1/\sqrt{\abs{H_\star}}$ for each $v\in H_\star$ and $\alpha_v=0$ otherwise. Consequently, every neighbor-index branch maps vertices of $H_\star$ back into $H_\star$ and the marking branch also leaves the initialized state unchanged. Thus, a direct calculation from \Cref{eq:cleanCC-StoqMA-pacc-decomp} shows that 
    \begin{subequations}
    \label{eq:cleanCC-yes}
    \begin{align}
        \bra{\Omega_{H_\star}} \Gamma \ket{\Omega_{H_\star}} &= \frac{1}{J} \rbra*{ \sum_{j\in\ssbra{d_G}}\sum_{v\in H_\star} \frac{1}{\abs{H_\star}} +  \sum_{v:C_M(v)=0} \frac{1}{\abs{H_\star}} + (J-d_G-1) }\\ 
        &= \frac{d_{G}+1+(J-d_{G}-1)}{J} = 1.
    \end{align}
    \end{subequations}
    
    Substituting \Cref{eq:cleanCC-yes} into \Cref{eq:cleanCC-StoqMA-pacc}, we obtain the desired equality
    \[\Pr\sbra*{V_{C_{G},C_M} \text{ accepts } \ket{H_\star}}=\frac{1}{2}+\frac{1}{2} \cdot 1 =1.\]  
    
    \paragraph{Soundness.} For \emph{no} instances, we first derive an expression for the rejection probability: 
    \begin{equation}
    \label{eq:cleanCC-prej}
        1-\Pr\sbra*{V_{C_{G},C_M} \text{ accepts } \ket{\psi}} = \frac{1}{2J} \rbra*{ \sum_{(u,v)\in E}(\alpha_u-\alpha_v)^2 + \sum_{v:C_M(v)=1}\alpha_v^2 }.
    \end{equation}
    
    To see \Cref{eq:cleanCC-prej}, we use \Cref{eq:cleanCC-StoqMA-pacc-decomp} and compute
    \begin{align*}
        J\rbra*{ 1-\bra{\Omega_\psi} \Gamma \ket{\Omega_\psi} } &= d_G  - \sum_{j\in\ssbra{d_G}}\sum_{v\in V}\alpha_v\alpha_{C_G(v,j)} + \sum_{v:C_M(v)=1} \alpha_v^2\\
        &= \frac{1}{2}\sum_{j\in\ssbra{d_G}} \sum_{v\in V}\rbra*{\alpha_v-\alpha_{C_G(v,j)}}^2 + \sum_{v:C_M(v)=1} \alpha_v^2\\
        &= \sum_{(u,v)\in E}(\alpha_u-\alpha_v)^2 + \sum_{v:C_M(v)=1} \alpha_v^2.
    \end{align*}
    Here, the second line uses $\sum_v\alpha_v^2=1$ and 
    \[ \sum_{j\in\ssbra{d_G}}\sum_{v\in V}\alpha_{C_G(v,j)}^2=d_G, \] 
    which follows because the map $(v,j) \mapsto \rbra[\big]{\wCG(v,j),C_G(v,j)}$ is a permutation of $\ssbra{d_G}\times V$. The third line uses the fact that the graph $G$ is undirected: every non-loop edge $(u,v)\in E$ is counted twice in the directed-edge sum, while padded self-loop terms contribute zero.  

    \vspace{1em}
    Equipped with \Cref{eq:cleanCC-prej} and $|V|=2^n$, it now suffices to prove the following inequality, which follows from a path-based Poincar\'e-type estimate of Jerrum's canonical-path inequality~\cite[Theorem~5.2 and Remark~5.3(a)]{Jerrum03}:  
    \begin{equation}
        \label{eq:cleanCC-loss-bound}
        \sum_{(u,v)\in E}(\alpha_u-\alpha_v)^2 + \sum_{v:C_M(v)=1}\alpha_v^2 \geq \frac{1}{|V|^2}.
    \end{equation}

    Fix one connected component $H\subseteq V$ and define its weight $\mu_H \coloneqq \sum_{v\in H}\alpha_v^2$. Since $(C_G,C_M)$ is a \emph{no} instance, the connected component $H$ contains at least one marked vertex $s$. Also, by an averaging argument, there exists a vertex $t\in H$ such that $\alpha_t^2 \geq \mu_H/\abs{H}$. Now we consider a simple path in $H$ from $s$ to $t$, denoted by $u_0 \sim u_1 \sim  \cdots \sim u_m$, with $u_0 \coloneqq s$ and $u_m \coloneqq t$. Since $m+1\leq \abs{H} \leq \abs{V}$, we obtain
    \begin{equation}
        \label{eq:cleanCC-no-path-identity}
        \alpha_t = \alpha_s + \sum_{i\in\ssbra{m}} \rbra*{\alpha_{u_{i+1}}-\alpha_{u_i}}.
    \end{equation}
    
    Applying the Cauchy--Schwarz inequality to \Cref{eq:cleanCC-no-path-identity} gives
    \begin{equation}
        \label{eq:cleanCC-no-path-bound}
        \alpha_t^2 
        \leq (m+1)\rbra*{ \alpha_s^2 + \sum_{i\in\ssbra{m}}\rbra*{\alpha_{u_{i+1}}-\alpha_{u_i}}^2 }
        \leq \abs{H}\rbra*{ \alpha_s^2 + \sum_{i\in\ssbra{m}}\rbra*{\alpha_{u_{i+1}}-\alpha_{u_i}}^2 }
    \end{equation}

    Since the chosen path from $s$ to $t$ in $H$ lies inside $H$ and $s$ is marked, the contribution of the component $H$ to the left-hand side of \Cref{eq:cleanCC-loss-bound} satisfies 
    \begin{subequations}
    \label{eq:cleanCC-one-component}
    \begin{align}
        \sum_{(u,v)\in E(H)}(\alpha_u-\alpha_v)^2 + \sum_{v\in H:C_M(v)=1}\alpha_v^2 
        &\geq \alpha_s^2 + \sum_{i\in\ssbra{m}} \rbra*{\alpha_{u_{i+1}}-\alpha_{u_i}}^2\\
        &\geq \frac{\alpha_t^2}{\abs{H}}\\
        &\geq \frac{\mu_H}{\abs{H}^2} \geq \frac{\mu_H}{\abs{V}^2}. 
    \end{align}
    \end{subequations}
    Here, the second line is obtained via rearranging \Cref{eq:cleanCC-no-path-bound}, and the last line follows from the facts that $\alpha_t^2 \geq \mu_H/\abs{H}$ and $H$ is a subgraph of $G$.  

    Finally, summing \Cref{eq:cleanCC-one-component} over all connected components of $G$ and using $\sum_v \alpha_v^2=1$, we establish \Cref{eq:cleanCC-loss-bound} and conclude the desired lower bound:
    \[ \sum_{(u,v)\in E}(\alpha_u-\alpha_v)^2 + \sum_{v:C_M(v)=1}\alpha_v^2
    \geq \sum_{H} \frac{\mu_H}{\abs{V}^2} 
    = \frac{1}{\abs{V}^2} \sum_{v\in V} \alpha^2_v= \frac{1}{\abs{V}^2}. \qedhere \]
\end{proof}

%%%%%%%%%%%%%%%%%%%%%%%%%%%%%%%%%%%%%%%%%%%%%%%%%%%%%%%%%%%
%%%%%%%%%%%%%%%%%%%%%%%%%%%%%%%%%%%%%%%%%%%%%%%%%%%%%%%%%%%
%%%%%%%%%%%%%%%%%%%%%%%%%%%%%%%%%%%%%%%%%%%%%%%%%%%%%%%%%%%

\section{The Barak-Kelner-Steurer algorithm and a refined analysis}
\label{sec:BKS-analysis}

We refer the reader to the original paper of Barak, Kelner, and Steurer~\cite[Section 1]{BKS14} for the background on the sum-of-squares algorithm. Our goal is to prove the following theorem for the purpose of completeness:
\BKSsSoS*

Our treatment largely follows Barak--Kelner--Steurer's proof, with similar notation.
We use ``degree'' and ``level'' interchangeably: a degree-$D$ pseudodistribution (resp., pseudoexpectation) is what they
would call a level-$D$ pseudodistribution (resp., pseudoexpectation). We improve the analysis to have the optimal dependence on $t$, the order of the tensor, but otherwise the same. The main tool underlying their algorithm is the ``reweighting'':
\begin{definition}[Conditioning, or reweighting] 
Let $X$ be a random variable over the $(d-1)$-dimensional sphere, and  $\EE$ be a degree-$D$ pseudoexpectation of $X$. Let $B(x)=S(x)^2$ be a square polynomial with $\EE_{x\sim X}[B]>0$. The reweighted pseudoexpectation for
$X\mid B$ is defined by
\[
\EE_{x\sim X\mid B}[P(x)]
 \coloneqq 
\frac{\EE_{x\sim X}[P(x)B(x)]}{\EE_{x\sim X}[B(x)]}
\]
for every polynomial $P$ of degree at most $D-\deg(B)$.
\end{definition}
For a tuple $i_1,\dots,i_{t-1}\in[d]$, we write
\[
X_{i_1,\dots,i_{t-1}}  \coloneqq  X\mid \prod_{j=1}^{t-1} x_{i_j}^2.
\]
This is exactly the conditioning operation used in~\cite{BKS14}.

\subsection{SoS algorithm and rounding}
\paragraph{The feasibility SDP.} Binary search for value $\nu$ such that the following degree-$D$ pseudoexpectation is feasible, for $D=O(t^2\log d/\epsilon^2)$:
\begin{align}
    &\EE \vDash ~\Tr M (x\cdot x^T)^{\otimes t} \ge \nu,\label{eq:hsep-constraint}
    \\
    &\EE \vDash ~\|x\|^2 = 1.\label{eq:sphere-constraint}
\end{align}
Here $\EE \vDash P(x)\ge 0 \iff \EE [P(x)Q^2(x)] \ge 0$  for any $\deg P + 2\deg Q \le D;$ while $\EE \vDash P(x)=0 \iff \EE [P(x)Q(x)]=0$ for all $\deg P + \deg Q\le D.$

\paragraph{The random variables $A(X),A_1,\dots,A_t$.}
Fix any level-$D$ pseudodistribution on $X$ with $D\ge 2t$. Define random variables $A_1,\dots,A_t$ over $[d]$ by
\[
\Pr[(A_1,\dots,A_t)=(i_1,\dots,i_t)]
 \coloneqq 
\EE_{x\sim X}\Big[\prod_{r=1}^t x_{i_r}^2\Big].
\]
Since $\prod_{r=1}^t x_{i_r}^2=(\prod_{r=1}^t x_{i_r})^2$ is a square of degree $2t$, these numbers are
non-negative. Also,
\[
\sum_{i_1,\dots,i_t}
\EE_{x\sim X}\Big[\prod_{r=1}^t x_{i_r}^2\Big]
=
\EE_{x\sim X}\Big[\Big(\sum_{i=1}^d x_i^2\Big)^t\Big]=1,
\]
so $A_1,\dots,A_t$ are honest random variables. Their common marginal is the distribution $A(X)$ on $[d]$
given by
\[
\Pr[A(X)=i]  \coloneqq  \EE_{x\sim X}[x_i^2].
\]
We also write $A$ for a random variable distributed according to $A(X)$, and define
\[
\Psi(X) \coloneqq \H(A(X))=\H(A).
\]
Because the joint moments are symmetric under permutation of the $t$ coordinates, the tuple
$(A_1,\dots,A_t)$ is exchangeable.

\paragraph{The direct rounding.}
For $\alpha=(i_1,\dots,i_t)\in[d]^t$, write $x_\alpha \coloneqq \prod_{r=1}^t x_{i_r}$. 

Define vectors
$y,z\in\mathbb R^{[d]^t}$ by
\[
y_\alpha  \coloneqq  \Big(\EE_{x\sim X}[x_\alpha^2]\Big)^{1/2}
\quad\text{and}\quad
z_\alpha  \coloneqq  \prod_{r=1}^t x_{i_r}^*  \coloneqq \prod_{r=1}^t \Big(\EE_{x\sim X}[x_{i_r}^2]\Big)^{1/2}.
\]
Then $\norm{y}=\norm{z}=1$, because
\[
\sum_\alpha y_\alpha^2
=
\EE_{x\sim X}\Big[\Big(\sum_{i=1}^d x_i^2\Big)^t\Big]=1
\quad\text{and}\quad
\sum_\alpha z_\alpha^2
=
\Big(\sum_{i=1}^d (x_i^*)^2\Big)^t=1.
\]
Here $x^*$ corresponds to the \emph{direct} rounding ignoring the correlation arising from the tensor product.

Abbreviate $M(x) \coloneqq \sum_{\alpha,\beta\in[d]^t} M_{\alpha\beta}x_\alpha x_\beta.$
Because $M_{\alpha\beta}\ge 0$ and pseudoexpectations satisfy Cauchy--Schwarz,
\[
\EE_{x\sim X}[x_\alpha x_\beta]
\le
\Big(\EE_{x\sim X}[x_\alpha^2]\,\EE_{x\sim X}[x_\beta^2]\Big)^{1/2}
=
y_\alpha y_\beta.
\]
Hence
\[
\nu\le \EE_{x\sim X}[M(x)]
\le
\sum_{\alpha,\beta} M_{\alpha\beta}y_\alpha y_\beta
=
\innerprodF{y}{My}.
\]
On the other hand, by our rounding and the choice of $z$, we have
\[
M(x^*)=\innerprodF{z}{Mz}.
\]
So for symmetric $M$,
\[
\nu-M(x^*)
\le
\innerprodF{y}{My}-\innerprodF{z}{Mz}
=
\innerprodF{y-z}{M(y+z)}
\le
\norm{M}\,\norm{y-z}\,\norm{y+z}.
\]
Since $y$ and $z$ are unit vectors, $\norm{y+z}\le 2$; and $\|M\|\le1$, one concludes that the direct rounding is worse than the SDP value by an additive term $2\|y-z\|$:
\begin{equation}
    \nu-M(x^*) \le 2\|y-z\|
\end{equation}

The Hellinger distance is the bridge from $\|\cdot\|$ to an entropic argument. In particular, 
\begin{align*}
\norm{y-z}^2
&=
2-2\sum_\alpha \sqrt{\Pr[(A_1,\dots,A_t)=\alpha]\,\Pr[\{A_1\}\cdots\{A_t\}=\alpha]}
\\
&=
2\,\dH\big(\{A_1\cdots A_t\},\{A_1\}\cdots\{A_t\}\big)^2.   
\end{align*}
Here, $\{X\}$ denotes the law of $X$, and $\dH$ is the Hellinger distance.
So the above analysis on the direct rounding can be summarized below:
\begin{lemma}[Direct rounding]\label{lem:direct}
Let $X$ be a level-$2t$ pseudodistribution satisfying \Cref{eq:hsep-constraint,eq:sphere-constraint}. If
\[
\dH\big(\{A_1\cdots A_t\},\{A_1\}\cdots\{A_t\}\big)\le \epsilon,
\]
then the unit vector $x^*\in\mathbb R^d$ with
\[
x_i^*  \coloneqq  \Big(\EE_{x\sim X}[x_i^2]\Big)^{1/2}
\]
satisfies
\[
M(x^*)\ge \nu-2\sqrt{2}\,\epsilon.
\]
\end{lemma}

\paragraph{Reweighting.}
If direct rounding fails to provide a good approximation, then $$\dH\big(\{A_1\cdots A_t\},\{A_1\}\cdots\{A_t\}\big) \gtrsim \epsilon.$$ Hence the variables $A_i$ are somewhat correlated, and one can pin some variables to make the remaining variables more deterministic. Updating the pseudodistribution by reweighting (conditioning on the pinnings) iteratively decreases the entropy of the variables $A_i$. Since entropy is bounded, within a bounded number of iterations, the Hellinger distance must be small and at that time the direct rounding would work.
This is the part where we tighten the analysis in BKS exploiting the chain rule and symmetry bypassing a wasteful triangle inequality.
\begin{lemma}[Entropy decrement per conditioning round]\label{lem:progress}
If
\[
\dH\big(\{A_1\cdots A_t\},\{A_1\}\cdots\{A_t\}\big)>\epsilon,
\]
then
\begin{equation}
\H(A_t\mid A_1,\dots,A_{t-1}) \le \H(A)-\frac{2\epsilon^2}{t}.\label{eq:entropy-decrement}
\end{equation}
\end{lemma}

\begin{proof}
Let $P \coloneqq \{A_1\cdots A_t\}$ and $Q \coloneqq \{A_1\}\cdots\{A_t\}$. 
Using \Cref{lemma:Hsquare-leq-KL}, we obtain
\begin{equation}
\KL(P\|Q)\ge 2\,\dH(P,Q)^2 > 2\epsilon^2.\label{eq:KL-bound}
\end{equation}
By the chain rule for relative entropy,
\begin{equation}
    \label{eq:chain-rule-symmetry}
    \KL(P\|Q) 
    = \sum_{j=1}^t \rmI(A_j;A_1,\dots,A_{j-1}) 
    \leq  t \cdot \rmI(A_t;  A_1,\dots,A_{t-1})
\end{equation}
The inequality holds, because $(A_1,\dots,A_t)$ is exchangeable with marginal $A$, and for every $j$,
\[
\rmI(A_j; A_1,\dots,A_{j-1})
\le
\rmI(A_t; A_1,\dots,A_{t-1}) = \H(A)-\H(A_t\mid A_1,\dots,A_{t-1}).
\]
Then \Cref{eq:entropy-decrement} in the statement follows by combining the above inequalities.
\end{proof}

\begin{corollary}[BKS conditioning succeeds when direct rounding fails]\label{cor:progress}
If
\[
\dH\big(\{A_1\cdots A_t\},\{A_1\}\cdots\{A_t\}\big)>\epsilon,
\]
then there exist $i_1,\dots,i_{t-1}\in[d]$ with
\[
\Pr[A_1=i_1,\dots,A_{t-1}=i_{t-1}]>0
\]
such that
\[
\Psi(X_{i_1,\dots,i_{t-1}})
\le
\Psi(X)-\frac{2\epsilon^2}{t}.
\]
\end{corollary}

\begin{proof}
\Cref{lem:progress} says that the average, over $(A_1,\dots,A_{t-1})$, of the conditional entropy
$\H(A_t\mid A_1,\dots,A_{t-1})$ is at most $\H(A)-2\epsilon^2/t$. Hence there exists a particular
$(i_1,\dots,i_{t-1})$ of positive probability such that
\[
\H(A_t\mid A_1=i_1,\dots,A_{t-1}=i_{t-1})
\le
\H(A)-\frac{2\epsilon^2}{t}.
\]
It remains to identify this conditional law with $A(X_{i_1,\dots,i_{t-1}})$. For every $i\in[d]$,
\begin{align*}
\Pr[A_t=i\mid A_1=i_1,\dots,A_{t-1}=i_{t-1}]
&=
\frac{\EE_{x\sim X}[x_i^2\prod_{j=1}^{t-1}x_{i_j}^2]}{\EE_{x\sim X}[\prod_{j=1}^{t-1}x_{i_j}^2]}
\\
&=
\EE_{x\sim X_{i_1,\dots,i_{t-1}}}[x_i^2]
=
\Pr[A(X_{i_1,\dots,i_{t-1}})=i].
\end{align*}
Therefore
\[
\Psi(X_{i_1,\dots,i_{t-1}})=\H(A(X_{i_1,\dots,i_{t-1}}))
=
\H(A_t\mid A_1=i_1,\dots,A_{t-1}=i_{t-1}),
\]
and the claim follows.
\end{proof}

\paragraph{Putting things together.}
Set
\[
\delta  \coloneqq  \frac{\varepsilon}{2\sqrt2}
\qquad\text{and}\qquad
R  \coloneqq  \left\lceil \frac{t\log d}{2\delta^2}\right\rceil+1.
\]
We now imitate the BKS combining algorithm, but with the corrected entropy decrement.

Start with $X^{(0)} \coloneqq X$. At stage $s$, let $A^{(s)}(X)$, $A_1^{(s)},\dots,A_t^{(s)}$, and $\Psi(X^{(s)})$ denote the
objects associated with the current pseudodistribution $X^{(s)}$.
If
\[
\dH\big(\{A_1^{(s)}\cdots A_t^{(s)}\},\{A_1^{(s)}\}\cdots\{A_t^{(s)}\}\big)\le \delta,
\]
then \Cref{lem:direct} yields a unit vector $x^*$ with
\[
M(x^*)\ge \nu-2\sqrt2\,\delta\,\norm{M}
\ge
\nu-\varepsilon\norm{M},
\]
and we stop.
Otherwise, direct rounding fails at stage $s$, so \Cref{cor:progress} gives a tuple $i_1^{(s)},\dots,i_{t-1}^{(s)}\in[d]$ such that
\[
\Psi\big(X^{(s)}_{i_1^{(s)},\dots,i_{t-1}^{(s)}}\big)
\le
\Psi(X^{(s)})-\frac{2\delta^2}{t}.
\]
We then set
\[
X^{(s+1)}  \coloneqq  X^{(s)}_{i_1^{(s)},\dots,i_{t-1}^{(s)}}.
\]

If direct rounding were to fail for every stage $s=0,1,\dots,R-1$, then repeated application of \Cref{cor:progress} would give
\[
\Psi(X^{(R)})
\le
\Psi(X^{(0)})-R\cdot \frac{2\delta^2}{t}
\le
\log d - R\cdot \frac{2\delta^2}{t}
<0,
\]
which is impossible because entropy is always nonnegative. Hence some stage must succeed, and the resulting
vector $x^*$ satisfies
$M(x^*)\ge \nu-\varepsilon.$

Each conditioning step reweights by the square monomial $\prod_{j=1}^{t-1} x_{i_j}^2$, whose degree is
$2(t-1)$. Hence after $s$ rounds, the total degree loss is $2s(t-1)$. To keep direct rounding available at
stage $s$, we still need level at least $2t$ at that stage. Therefore it is enough to start with
$D \ge 2t+2(t-1)R.$ By our choice of $R$,
\[
D = O\left(\frac{t^2\log d}{\varepsilon^2}\right).
\]
Thus, we have completed the proof of \Cref{thm:bks-sos}.

%%%%%%%%%%%%%%%%%%%%%%%%%%%%%%%%%%%%%%%%%%%%%%%%%%%%%%%%%%%
%%%%%%%%%%%%%%%%%%%%%%%%%%%%%%%%%%%%%%%%%%%%%%%%%%%%%%%%%%%
%%%%%%%%%%%%%%%%%%%%%%%%%%%%%%%%%%%%%%%%%%%%%%%%%%%%%%%%%%%

\section{Omitted proofs}

\subsection{Omitted proofs in \texorpdfstring{\Cref{sec:parallel-repetition-StoqMAk}}{Section 3}}
\label{subsec:omitted-parallel-repetition-StoqMAk}

We first show the multiplicativity of stoquastic separable values of product Hermitian matrices: 

\stoqSepValMultiplicativeProduct*

\begin{proof}
For simplicity, we prove only the case where $k=2$, and the same argument directly extends to general $k$. Similar to \Cref{prop:opt-stoqSepVal-nonNeg}, it suffices to consider only non-negative quantum states. Also, the lower bound is straightforward, following the same argument as that used for the corresponding part of the proof of \Cref{lemma:stoqSepVal-multiplicative-PSD}. 

We now establish the upper bound. Let the operator $M$ act on quantum registers $\sfA$ and $\sfB$, and let the operator $M'$ act on quantum registers $\sfA'$ and $\sfB'$. 
After regrouping the registers as $\sfA\sfA'$ and $\sfB\sfB'$, we obtain 
\[ M_{\sfA\sfB}\otimes M'_{\sfA'\sfB'} = (M_1 \otimes M'_1)_{\sfA\sfA'}\otimes (M_2\otimes M'_2)_{\sfB\sfB'}\]
from \Cref{eq:sep-vals-product-form}. Then it follows that:
    \begin{align*}
        &\hsepPlus{d_1d'_1}{d_2d'_2}(M \totimes M') \\
        =~& \max_{\ket{\phi_1}\in \calS_+(d_1d'_1),\,\ket{\phi_2}\in \calS_+(d_2d'_2)} \bra{\phi_1}(M_1\otimes M'_1)\ket{\phi_1} \cdot \bra{\phi_2}(M_2\otimes M'_2)\ket{\phi_2}  \\
        \leq~& \lambda_{\max}(M_1\otimes M'_1)\cdot \lambda_{\max}(M_2\otimes M'_2) \\
        \leq~& \lambda_{\max}(M_1) \cdot \lambda_{\max}(M_2) \cdot \lambda_{\max}(M'_1) \cdot \lambda_{\max}(M'_2)\\
        =~& \max_{\ket{\psi_1},\ket{\psi_2}} \bra{\psi_1}\totimes \bra{\psi_2}\rbra*{M_1\totimes M_2}\ket{\psi_1}\totimes\ket{\psi_2} \cdot \max_{\ket{\psi'_1},\ket{\psi'_2}} \bra{\psi'_1}\totimes\bra{\psi'_2} \rbra*{M'_1\totimes M'_2} \ket{\psi'_1}\totimes\ket{\psi'_2}\\
        =~& \hsepPlus{d_1}{d_2}(M) \cdot \hsepPlus{d'_1}{d'_2}(M').
    \end{align*}
Here, the fourth line follows from Theorem 3 in~\cite[Section VIII.3]{Gantmakher00}, which states that the absolute value of each eigenvalue of an entrywise non-negative matrix $A$ is at most the largest eigenvalue $\lambda_{\max}(A)$. The fifth line uses the fact that $\lambda_{\max}(A)=\max_{\ket{\psi}} \bra{\psi} A \ket{\psi} = \bra{\psi_\star} A \ket{\psi_\star}$ for some non-negative unit vector $\ket{\psi_\star}$ that attains the optimum. 
\end{proof}

Similar to the notion of \emph{separable} in~\cite{HM13}, one may wonder whether \Cref{lemma:stoqSepVal-multiplicative-product} extends to convex combinations of such matrices, which we refer to as \emph{product-convex Hermitian} matrices. An explicit counterexample can be found in an earlier version of our work~\cite[Remark 3.10]{LW26}.

\vspace{1em}
Next, we prove that $\SepRCD_k$ is complete for \StoqMAk{}: 

\SepRCDisComplteForStoqMA*

\begin{proof}
    The \StoqMAk{} containment follows directly from applying the branch-overlap test (\Cref{lemma:branch-overlap-test}) to the reversible circuit pair $(R_0,R_1)$, with the input state chosen as
    \[\ket{\Psi} = \ket{\psi_1}\otimes\cdots\otimes\ket{\psi_k}.\]

    The \StoqMAk{}-hardness follows from~\cite[Figure 2]{Liu21}. Let $V_x$ be a stoquastic verifier corresponding to $x\in(\calI_\yes,\calI_\no)$, where $(\calI_\yes,\calI_\no) \in \StoqMA\rbra[\big]{k,\frac{1+a}{2},\frac{1+b}{2}}$.  Let the reversible circuits $R_0$ and $R_1$ be defined by $R_0 \coloneqq V_x^\dagger X_{\sfO} V_x$ and $R_1 \coloneqq I$. Using the construction $\calA$ for these choices of $R_0$ and $R_1$ to define a new stoquastic verifier $\widehat{V}_x$, \Cref{lemma:branch-overlap-test} yields the corresponding acceptance probability as follows:   
    \begin{subequations}
    \label{eq:StoqMAk-pacc-alternative}
    \begin{align}
        \Pr\sbra*{\widehat{V}_x\text{ accepts }\ket{\psi'}} &= \frac{1}{2} + \frac{\innerprod{R_0}{R_1}}{2}\\
        &= \bra{\psi'}\bra{\bar{0}}\bra{\bar{+}} V_x^\dagger \rbra[\Big]{\frac{X_{\sfO}+I}{2}} V_x \ket{\psi'}\ket{\bar{0}}\ket{\bar{+}}\\ 
        &= \Pr\sbra*{V_x\text{ accepts }\ket{\psi'}}.
    \end{align}
    \end{subequations}
    This identity completes the proof, as desired.  
\end{proof}

%%%%%%%%%%%%%%%%%%%%%%%%%%%%%%%%%%%%%%%%%%%%%%%%%%%%%%%%%%%

\subsection{Omitted proof of the generalized birthday paradox}
\label{subsec:omitted-birthday}

\birthdayParadox*

\begin{proof}
Let $\bbI_{ij}$ be the indicator that $(X_i,X_j)\in B\cap(\Omega_0\times\Omega_0)$, and
\[
    Z \coloneqq \sum_{1\le i<j\le K} \bbI_{ij}.
\]
Then $\E[Z]=\binom{K}{2}\lambda.$
Because $\lambda\ge \beta/n$ and $K=C\sqrt n$, we have
$\E[Z]\ge \beta C^2/3$
for all sufficiently large $n$.

We claim that $\E[Z^2]=O(1)+O(\E[Z]^2)$, with constants depending only on
$\alpha,\beta,D, C$.
Expand
\begin{equation}
    Z^2=\sum_{i<j} \bbI_{ij} + 2\sum_{(i,j)<(k,\ell)} \bbI_{ij}\bbI_{k\ell}. \label{eq:Z-squared}
\end{equation}
There are three kinds of terms in the above equation:

\medskip
\noindent
\emph{Diagonal terms.}
These contribute exactly $\E[Z]$.

\medskip
\noindent
\emph{Disjoint pairs of pairs.}
If $\{i,j\}\cap\{k,\ell\}=\varnothing$, then $\bbI_{ij}$ and $\bbI_{k\ell}$ are independent, so
\[
    \E[\bbI_{ij}\bbI_{k\ell}]=\lambda^2.
\]
The total contribution of all disjoint pairs of pairs is therefore $O(K^4\lambda^2)$, which is
$O(\E[Z]^2)$.

\medskip
\noindent
\emph{Overlapping pairs of pairs.}
It remains to bound terms such as $\E[\bbI_{ij}\bbI_{ik}]$.
Define
\[
    s(x) \coloneqq \sum_{y\in\Omega_0:(x,y)\in B} \mu(y) \le \frac{D \alpha}{n},
\]
where the inequality is due to our assumption (i) and (ii). 
Hence
\begin{align*}
    \E[\bbI_{ij}\bbI_{ik}]
    &=\sum_{x\in\Omega_0} \mu(x)s(x)^2 
    \le \frac{D \alpha}{n}\sum_{x\in\Omega_0} \mu(x)s(x) 
    =\frac{D \alpha}{n}\lambda.
\end{align*}
Because $\lambda\le D \alpha/n$ as well, this contribution is $O(1/n^2)$.
There are only $O(K^3)=O(n^{3/2})$ overlapping pairs of pairs, so their total contribution is
$O(K^3/n^2)=o(1)$.

Combining all contributions,
\[
    \E[Z^2]\le \E[Z] + O(\E[Z]^2) + o(1).
\]
Choose $C$ large enough so that $\E[Z]\ge 1$ for all sufficiently large $n$.
Then, for large $n$,
\[
    \E[Z^2]\le A\E[Z]^2
\]
for some constant $A=A(\alpha,\beta,D, C)>0$.
Paley--Zygmund now yields
\[
    \Pr[Z>0]\ge \frac{\E[Z]^2}{\E[Z^2]}\ge \frac{1}{A}.
\]

To further amplify the above probability to $1-\epsilon$, we increase the number of samples by a factor of $\log(1/\epsilon)$. More precisely, we repeat the experiment in $O(\log(1/\epsilon))$ independent blocks and accept if any block finds a bad pair; this procedure reduces the failure probability to at most $\epsilon$.  For constant $\epsilon$, this procedure only changes the absolute constant $C$. This gives the desired conclusion.
\end{proof}

%%%%%%%%%%%%%%%%%%%%%%%%%%%%%%%%%%%%%%%%%%%%%%%%%%%%%%%%%%%

\subsection{Omitted proofs of rETH-optimality of the \BPTIME{} upper bound for \StoqMA{}}
\label{subsec:omitted-rETH-opt-BPTIME-bound}

\rETHoptimal*

\begin{proof}
    Let $\varphi$ be an $n$-variable \SAT{} instance. We now provide an efficient construction of a single-prover stoquastic verification circuit with witness length $\ell(n)=n$, completeness $1$, and soundness $1/2$. The witness register $\sfW$ consists of $n$ qubits and is interpreted as an assignment $a\in\binset^n$. Let $\sfQ$ be a single-qubit register initialized to $\ket{+}$, let $\sfF$ be a single-qubit flag register initialized to $\ket{0}$, and let $\sfA$ be a $\poly(n)$-qubit work register initialized to $\ket{\bar{0}}$ that includes a single-qubit register $\sfB$. 

    The stoquastic verification circuit $V_{\varphi}$ does the following:
    \begin{enumerate}[label=\upshape(\arabic*)]
        \item Compute the bit $\SAT_\varphi(a) = \bbI\sbra{\varphi(a)\text{ is satisfied}}$ and write the resulting bit into $\sfB$.
        \item Apply $X_\sfB\,{\rm Toffoli}_{\sfQ,\sfB\to \sfF}\,X_\sfB$, which flips $\sfF$ when $\sfQ$ contains $\ket{1}$ and $\varphi(a)$ is unsatisfied. 
        \item Uncompute the work register and measure $\sfQ$ in the Hadamard basis.
    \end{enumerate}

    \parheading{Correctness analysis.}
    For a basis witness $\ket{a}$, if $\varphi(a)$ is satisfied, then $V_{\varphi}$ accepts with probability $1$. If $\varphi(a)$ is unsatisfied, then the resulting state in the registers $(\sfQ,\sfF)$ is $(\ket{00}+\ket{11})/\sqrt{2}$, and thus $V_\varphi$ accepts with probability $1/2$. 
    
    Now let us move on to the general case. Consider the witness state $\ket{\psi}=\sum_{a\in\binset^n} \alpha_a \ket{a}$, where $\sum_a \abs{\alpha_a}^2=1$. By linearity, the acceptance probability can be expressed as 
    \[ \Pr\sbra*{V_\varphi \text{ accepts } \ket{\psi}} = \frac{1}{2} + \frac{1}{2} \sum_{a\in\binset^n \,\wedge\, \varphi(a)\text{ is satisfied}} \abs*{\alpha_a}^2. \]
    
    Consequently, if $\varphi$ is satisfiable, a satisfying assignment basis state makes the verifier accept with probability $1$; whereas if $\varphi$ is unsatisfiable, then every witness is accepted with probability exactly $1/2$. Therefore, $\varphi$ is efficiently mapped to a stoquastic verification circuit for $\StoqMA_\ell\ssbra{1,1,1/2}$ with $\ell(n)=n$ and $\abs{V_\varphi}=\poly(n)$. Together with the hypothesis, this reduction yields a randomized $2^{o(n)}\poly(n)$-time algorithm for \SAT{}, which contradicts randomized ETH~\cite[Section 1.3]{DHMTW14} (see also~\cite[Hypothesis 1]{CRTY23}).
\end{proof}

\end{document}